\documentclass[a4paper,12pt]{article}

\input{preamble}

\title{
Evaluating the Efficiency of Regulation in Matching Markets with Distributional Disparities
\thanks{We thank Alfred Galichon, Fuhito Kojima, and Yeon-Koo Che for helpful comments. 
We also acknowledge the participants of NYU econometrics seminar, JEA annual meeting, APIOC, and IIOC.}}
\author{Kei Ikegami\thanks{New York University, \texttt{ki2047@nyu.edu}.} \and Atsushi Iwasaki\thanks{University of Electro-Communications, \texttt{atsushi.iwasaki@uec.ac.jp}} \and Akira Matsushita\thanks{Kyoto University, \texttt{amatsushita@i.kyoto-u.ac.jp}} \and Kyohei Okumura\thanks{Northwestern University, \texttt{kyohei.okumura@gmail.com}}}

\begin{document}

\begin{titlepage}

\newgeometry{top=0.3in,bottom=0.7in,right=1in,left=1in}
\maketitle
\thispagestyle{empty}
\vspace{-20pt}

\begin{abstract}
\noindent
Cap-based regulations are widely used to address distributional disparities in matching markets, but their efficiency relative to alternative instruments such as subsidies remains poorly understood. This paper develops a framework for evaluating policy interventions by incorporating regional constraints into a transferable utility matching model. We show that a policymaker with aggregate-level match data can implement a taxation policy that maximizes social welfare and outperforms any cap-based policy. Using newly collected data from the Japan Residency Matching Program, we estimate participant preferences and simulate counterfactual match outcomes under both cap-based and subsidy-based policies. The results reveal that the status quo cap-based regulation generates substantial efficiency losses, whereas small, targeted subsidies can achieve similar distributional goals with significantly higher social welfare.


\vspace{12pt}
    \noindent JEL Classification Codes: C25; C78; H20; I14.\\ 
    \noindent Keywords: Matching with Distributional Constraints; Discrete Choice; Assignment Problems; Two-sided Matching; Subsidies; Tax; Redistribution. 
\end{abstract}
\end{titlepage}

\restoregeometry

\section{Introduction}
\label{sec:intro}

Outcomes in real-world matching markets often diverge from socially desirable allocations, particularly in terms of distributional equity. To address such imbalances, policymakers often restrict the number of matches for certain categories. Examples of such regulations include race-based affirmative action in U.S. college admissions aimed at promoting diversity \citep{Ellison2021-fr}, gender quotas in electoral systems designed to enhance female representation \citep{Besley2017-ay}, and the residency market in Japan analyzed in \cite{kamada:aer:2015}, in which the government currently caps the number of matches in high-demand urban areas to redirect applicants to underserved rural areas.

While such cap-based regulations are widely used, monetary interventions such as taxes and subsidies offer a natural, yet comparatively underexplored, alternative. The potential advantage of monetary interventions lies in their ability to account for the intensity of preferences. A cap is a blunt instrument that may prevent high-surplus matches, whereas a targeted subsidy could, in principle, influence marginal participants without generating large welfare losses. It remains a central and unresolved question for both theory and policy whether this theoretical potential translates into significant efficiency gains, and under what conditions.

Answering these questions requires moving beyond the non-transferable utility (NTU) framework that dominates the literature on matching with distributional constraints. The NTU framework is ill-suited for policy analysis for two primary reasons. First, its reliance on ordinal preferences cannot readily accommodate interventions involving taxes or subsidies, or quantify their welfare consequences. Second, it typically treats transfers, such as wages, as exogenous, failing to capture how they endogenously adjust in response to policy interventions. These limitations underscore the need for a modeling approach toward more efficient and credible policy intervention.

In this paper, we develop a transferable utility matching model that incorporates \emph{regional constraints}, defined as upper and lower bounds on the number of matches in each region. Within this framework, a system of taxes and subsidies can implement the matching that maximizes participants' welfare among all those satisfying the constraints. This welfare-maximizing outcome establishes a clear welfare benchmark against which other policies can be evaluated. We then provide a path to empirical implementation by embedding our model in the aggregate matching framework of \citet{galichon_cupids_2021}. This enables us to identify model primitives from past match data, compute optimal subsidy policies, and conduct counterfactual simulations comparing their performance to existing cap-based regulations. Applying this framework to newly collected data from the Japan Residency Matching Program, our counterfactual simulations reveal that the current cap-based policy generates substantial efficiency losses, whereas small, targeted subsidies can achieve similar distributional goals with significantly higher social welfare.

Distributional concerns have played a central role in Japan’s residency matching market, prompting various proposals from both practitioners and researchers. In 2005, Japan introduced a mandatory two-year residency program and adopted a centralized matching system known as the Japan Residency Matching Program (JRMP). However, geographic disparities in physician distribution have posed persistent challenges to equitable access to healthcare since the program’s inception. In response, the JRMP began capping the number of residency positions in urban areas in 2008 to encourage matches in underserved regions. 

The data reveal important limitations of the current cap-based policy and point to potential sources of its inefficiency. Despite progressively tighter caps on urban positions over the past decade, geographic imbalances in residency placements persist. Moreover, the introduction of regional caps has been associated with a rising rate of unmatched applicants. These patterns suggest that the cap-based policy currently employed by the JRMP may be inefficient. One potential source of this inefficiency is the failure to account for endogenous transfer formation. The data show substantial variation in salaries across residency programs in Japan. In particular, less popular programs, especially those located in rural areas, appear to offer higher compensation to attract medical students. This observation motivates our modeling approach, which endogenizes transfers between matched pairs.

In our model, each position belongs to a region, and caps and floors on the number of matches in each region are exogenously determined. A policymaker aims to satisfy these constraints, but they may not be satisfied in an equilibrium matching formed by agents without intervention. If the policymaker has complete knowledge of market participants’ preferences, she can induce the matching that maximizes social surplus among all matchings that satisfy the regional constraints via a taxation policy. The taxes and subsidies correspond to the Lagrange multipliers associated with the regional constraints in the social surplus maximization problem. 
This result implies that the taxation policy strictly outperforms any cap-based policy, serving as a useful benchmark for evaluating the efficiency of alternative policies, as it achieves the constrained-optimal.

In practice, policymakers rarely have full knowledge of individual preferences and must infer them from past matching outcomes. To address this, we embed our framework in the model introduced by \citet{galichon_cupids_2021}, which enables identification of the primitives of the transferable utility model using only aggregate-level matching data. This allows us to compute the optimal taxation policy from observed data (Theorem \ref{theorem:EAE_opt}) and to conduct counterfactual simulations to compare welfare across alternative interventions, including cap-based regulations (Proposition \ref{prop:ae}). These simulations also allow us to quantify the relative inefficiencies of different policy tools.

For the empirical analysis, we start by defining the transfer between matched pairs and then consider a measurement model for this transfer. To this end, we model the baseline utilities of both sides for each pair of agents. The transfer from the hospital to the matched resident represents the gap between this baseline utility and the equilibrium utility level. By aggregating this transfer across the pairs within each hospital, we construct an aggregate-level transfer. We then introduce a measurement model for this aggregate-level transfer, which is linear in observed monthly salary. Our estimation proceeds in two steps: first, we estimate the aggregate-level surplus split for each pair of a medical school and a hospital, following \cite{galichon_cupids_2021}; second, based on these first-step estimates, we recover the parameters in the baseline utilities and the measurement model.

The estimation results partly align with those in the existing literature while also presenting a departure from its assumptions. On the doctor’s side, our estimates indicate that factors such as the distance between the hospital and the doctor’s alma mater, as well as the hospital size, are significant. Furthermore, we find that the number of previous matches is an important determinant of doctors’ preferences. These findings are consistent with the existing literature, which suggests that hospitals are horizontally differentiated from the perspective of doctors. On the hospital’s side, we observe that hospitals exhibit horizontal preferences similar to those of doctors: doctors from distant regions are less preferred, in addition to considerations of quality. Such horizontal preference structure on the hospital side is not allowed in the existing literature, which hinders the identification of preferences from matching data.

Our empirical estimates inform a series of counterfactual simulations designed to quantify the inefficiency of existing regulations and evaluate policy alternatives. First, a comparison between the status quo and an unregulated benchmark reveals that the current urban caps, while effective at redirecting some residents, generate a substantial aggregate welfare loss. Second, to distinguish the cost of the distributional goal itself from the cost of the policy instrument, we simulate an optimal subsidy policy designed to achieve the same resident distribution as the current cap-based system. We find that a relatively small, targeted subsidy can meet this goal at a significantly higher level of social welfare. This result indicates that the vast majority of inefficiency stems not from the distributional constraint itself, but from the choice of regulatory instrument, the caps.
Our findings therefore suggest that policymakers can achieve distributional goals more efficiently by replacing blunt quantity restrictions with carefully designed monetary interventions.
Finally, we test the limits of non-monetary interventions by simulating the Flexible Deferred Acceptance mechanism, a sophisticated cap-setting algorithm proposed by \cite{kamada:aer:2015}. Our results show that even this advanced mechanism generates substantial welfare losses compared to an optimal subsidy policy. This finding reinforces our conclusion that quantity restrictions are an inherently inefficient instrument for achieving distributional goals.

\paragraph{Related literature}

Matching with constraints, initiated by \citet{kamada:aer:2015}, has attracted considerable attention across various fields due to its applicability to real-world settings \citep{Abdulkadiroglu:AER:2003,ehlers::2012,kojima2012school,hafalir2013effective, Fragiadakis2017-jn}.
A key departure from their approach is our adoption of the TU framework in a matching market with distributional constraints.

Recent studies by \citet{Kojima2020-wt} and \citet{Jalota2025-rh} investigate broad classes of distributional constraints, primarily focusing on the existence of stable outcomes.
If a firm in their model is regarded as a region, the two models appear similar, but they differ in certain aspects. For instance, in our model, a distinction should be made between matching a doctor to a single hospital in a region and matching the same doctor to a different hospital in the same region; however, such a distinction cannot be described in their models.\footnote{This is analogous to the difference between the models for affirmative action such as \cite{Abdulkadiroglu:AER:2003} and the model of \cite{kamada:aer:2015}.} Our analysis focuses on a specific class of constraints, regional constraints, and studies the problem of a policymaker who designs a taxation policy to ensure that these constraints are satisfied. We also extend the model to accommodate empirical analysis using aggregate-level matching data.


Next, we highlight our contribution from the viewpoint of the empirical analysis of matching markets.
The structural analysis of matching markets is widely accepted in many fields of economics, including labor economics and industrial organization, which also adopts a type of matching model to describe trade networks \citep{fox_estimating_2018,fox_unobserved_2018}. 
Notable methodological contributions are developed in \cite{galichon_cupids_2021}: the nonparametric identification result of the social surplus function in a transferable utility matching market, which is robust to distributional assumptions on unobserved heterogeneity, along with the corresponding estimators\footnote{The methodology in \cite{galichon_cupids_2021} is also distinct compared to other methodologies dependent on distributional assumptions, such as the minimum score estimator proposed by \cite{Fox2010,fox_2017,fox_measuring_2013} and the maximum likelihood estimator of \cite{choo_who_2006}.}.
This study builds on the general framework of \cite{galichon_cupids_2021}, proposing an extended model that accommodates regional constraints. Our framework is thus applicable without assuming a specific distributional form for unobserved heterogeneity terms. Furthermore, diverging from \cite{galichon_cupids_2021}, we propose a formal estimation strategy that exploits a measurement of the transfer to quantify the parameter values in a monetary unit.

Lastly, we mention \cite{agarwal_empirical_2015} as the closest empirical research to this study.
\cite{agarwal_empirical_2015} takes a non-transferable utility matching model to analyze the doctor-hospital matching market in the U.S. National Residency Matching Program (NRMP), where a centralized mechanism determines the matching, and salaries are determined almost exogenously. While \cite{agarwal_empirical_2015} addresses the endogeneity of salaries using a control function approach, our study fully models the salary determination process. This difference reflects the variation in market environments: in Japan, the concentration in urban areas presents more severe issues, and salaries are used as a tool to attract more candidates. 
Furthermore, treating transfers endogenously allows us to recover preferences under counterfactual regulations and simulate the resulting matching outcomes. This feature is particularly useful for policymakers when evaluating potential regulatory interventions, which a non-transferable utility matching model, which treats transfers as exogenous, cannot offer.

\subsection{Institutional Background of the JRMP}
\label{sec:JRMP_institutional_detail}

In 2004, Japan introduced a mandatory two-year residency training program. Under the new program, residents and hospitals are matched through a centralized mechanism called the Japan Residency Matching Program (JRMP). The JRMP is modeled after the National Resident Matching Program in the United States. Medical students and hospitals must register with the JRMP system and submit their respective preference rankings. Students are provided with information about hospitals, including hospital size, location, specifics of the training program, salary, and workload. Additionally, they can participate in job fairs and directly visit hospitals to obtain further details. Before submitting their preference lists, students are required to take examinations administered by each hospital they consider ranking. Once the preference lists are finalized, the JRMP employs the deferred acceptance algorithm to determine the matches. In 2023, the JRMP involved 10,202 students and 1,209 hospitals offering 10,895 positions. On average, each student listed 4.35 hospitals on their preference list. The algorithm successfully matched 87.9\% of the students; specifically, 64.3\% secured their first-choice hospitals, 16.3\% their second choice, and 9.0\% their third choice.\footnote{Unmatched students may individually contact hospitals with vacancies or reapply through the matching process in the subsequent year.}

Researchers and Japanese media have reported that the distributional disparity of physicians across regions has worsened following the introduction of the JRMP \citep{Iizuka2016-wp, Sakai2013-co, Endo2019}. One primary factor contributing to this imbalance is that the centralized matching system allowed medical students to freely express their location preferences, leading to a concentration of applicants in non-university hospitals. Under the previous system, a majority of graduates would begin their residency at the university hospital affiliated with their medical school.\footnote{Although it was not mandatory, most medical students chose to undertake residency training after graduation.} Many medical graduates were affiliated with university medical departments, which directed graduates to affiliated rural hospitals, maintaining a relatively balanced distribution. However, the JRMP effectively removed the obligatory connection between graduates and university departments, resulting in fewer graduates choosing positions in rural or underserved regions. Consequently, rural hospitals faced increasing difficulties recruiting physicians, exacerbating regional disparities in patient access to hospital services in these areas.

To mitigate these regional imbalances, the JRMP introduced regional caps starting in 2010. The process for setting regional caps involves several steps. First, the total number of residency positions nationwide is determined by multiplying the number of medical students by a fixed constant. This constant was approximately 1.22 in 2015 but is scheduled to be reduced gradually to 1.05 by 2025. After determining the nationwide cap, residency positions are allocated across prefectures using a formula that incorporates factors such as population size, medical school enrollment capacities, the current number of practicing doctors, geographic considerations, including physician density (the number of doctors per unit area), and populations in isolated or remote islands. This allocation method is intentionally designed to benefit underserved areas by disproportionately reducing the number of residency positions in urban regions. The rationale behind imposing regional caps is that by restricting positions available in urban areas and tightening overall capacity, more medical graduates will be encouraged or compelled to pursue training in rural regions, under the expectation that residents who train in these underserved areas are more likely to remain and practice there long-term, thereby addressing existing regional disparities.

\section{Data}
\label{sec:data}

Our analysis covers the four years of matching results generated by the JRMP from 2016 to 2019. To estimate our model, we need three key elements: the matching patterns between medical schools and hospitals, the characteristics of these institutions relevant to their preferences, and the salaries paid to residents during their internships. We begin by describing the data sources, present the descriptive statistics in Section~\ref{sec:desc_stats}, and discuss the empirical matching patterns observed in the market in Section~\ref{sec:empirical_pattern}.

The matching patterns between medical schools and hospitals, i.e., the numbers of matches between any given pair, are calculated based on the “Physician Registration Report,” a certificate that summarizes the personal information of licensed physicians, including residents. This data source allows us to calculate the annual number of matches between specific medical schools and hospitals using information on their respective graduating institutions and training hospitals.

We obtained the characteristics of hospitals from the JRMP website, which provides details such as hospital names, program offerings, and capacity. For the characteristics of medical schools, we used the national exam pass rate and whether the university is public, based on publicly available information from hospital websites. Additionally, to measure the expected ability of graduates from a medical school, we used the T-score of the entrance exam.\footnote{The T-scores of universities are published by cram schools. These scores are calculated using data from practice exams administered by the cram schools, which gather information on students' actual university entrance exam results. The T-scores reflect the relationship between students' performance on practice exams and their success on university entrance exams.} T-scores are widely recognized as an indicator of university entrance exam difficulty in Japan, with higher scores indicating more challenging universities. We used the most recent T-scores available for our estimation.\footnote{The data source is \url{https://www.keinet.ne.jp/university/ranking/}.}

Finally, we gathered salary data by crawling hospital websites. Due to limited data availability, we used the most recent salary information rather than data from 2016 to 2019, assuming that salary levels remained constant during this period. We also collected additional hospital-related information, such as location, number of beds, and emergency transport cases.

\subsection{Descriptive Statistics}
\label{sec:desc_stats}

Table~\ref{table:data_summary_stats} summarizes the environment and the outcomes of JRMP for the years 2017, 2018, and 2019. Since our estimation uses the matching patterns of the last year as one of the covariates, we exclude the fiscal year of 2016. 
Panel A and Panel B in Table \ref{table:data_summary_stats} summarize the results of JRMP. 
The environment of the matching market—characterized by the number of schools, students, hospitals, and available slots—remains stationary over these three years, with minimal entry or exit. The matching outcomes, such as the unmatched rate, also appear stable during this period. Based on these observations, we assume that the preference structures of schools and hospitals remain unchanged throughout the data period.

\begin{table}[t]
\centering
\small
\caption{Environments and Outcomes of JRMP}
\label{table:data_summary_stats}
\begin{threeparttable}
\begin{tabular}{@{}l p{0.15\textwidth} p{0.15\textwidth} p{0.15\textwidth}}
\toprule
 & \textbf{2017} & \textbf{2018} & \textbf{2019} \\ \midrule
\multicolumn{4}{l}{\hspace{-0.6em}\textbf{Panel A: Doctor side}} \\
Number of schools & 78 & 78 & 78 \\
Number of students & 9830 & 9916 & 9932 \\
Number of matched students & 8530 & 8369 & 8634 \\
Number of unmatched students & 1300 & 1547 & 1298 \\
Unmatch rate ($\%$)& 13.22 & 18.48 & 15.03 \\ \midrule
\multicolumn{4}{l}{\hspace{-0.6em}\textbf{Panel B: Hospital side}} \\
Number of hospitals & 1025 & 1025 & 1022 \\
Number of total seats & 11716 & 11468 & 11730\\ 
Number of matched seats & 8530 & 8369 & 8634 \\ 
Number of unmatched seats & 3186 & 3099 & 3096 \\
Unmatch rate ($\%$)& 27.19& 27.02 & 26.39 \\
Number of excess seats & 1886 & 1551 & 1798 \\
Excess rate ($\%$) & 16.01 & 13.52 & 15.33\\ 
\bottomrule
\end{tabular}
\end{threeparttable}
\end{table}

\begin{figure}[ht]
\centering
\begin{minipage}[b]{0.48\linewidth}
    \centering
    \includegraphics[width=\linewidth]{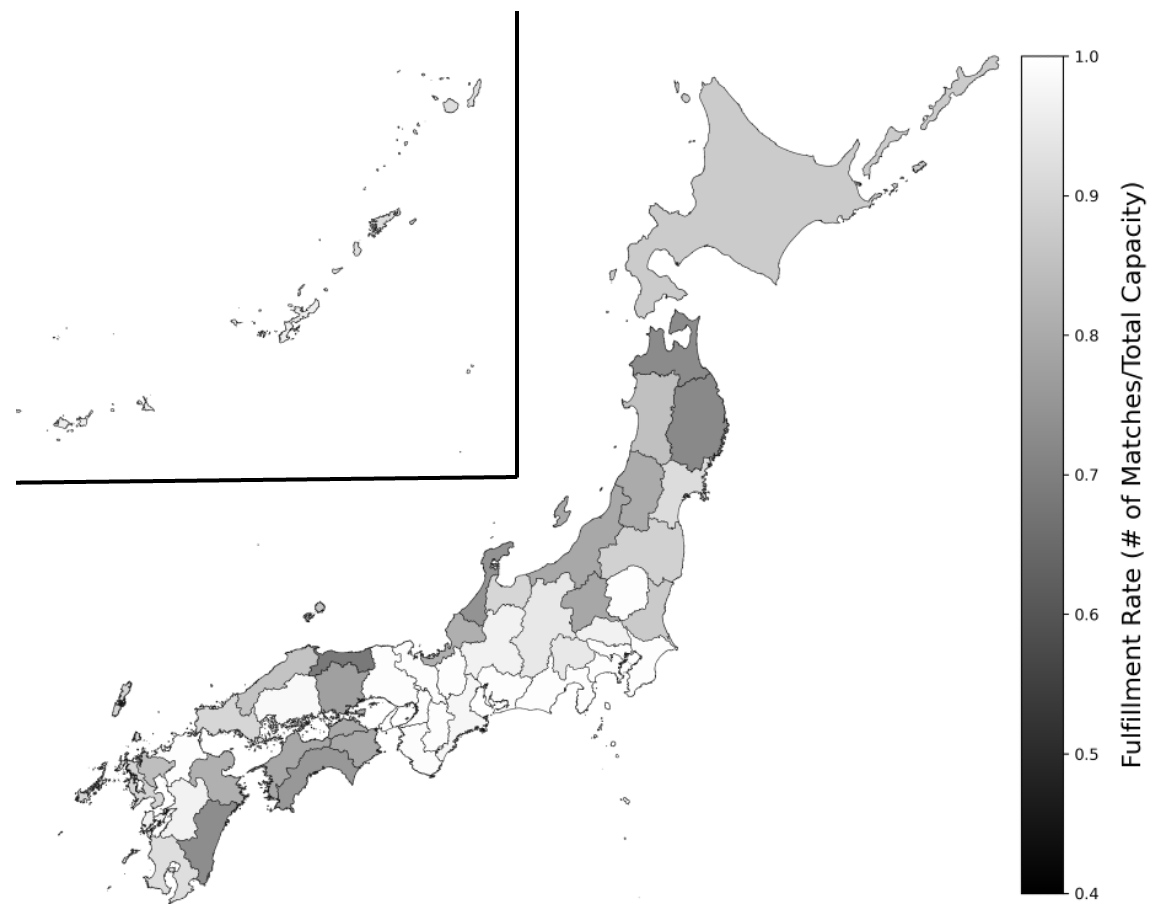}
    \subcaption{Fulfillment Rates}\label{fig:fulfillment}
\end{minipage}
\begin{minipage}[b]{0.48\linewidth}
    \centering
    \includegraphics[width=\linewidth]{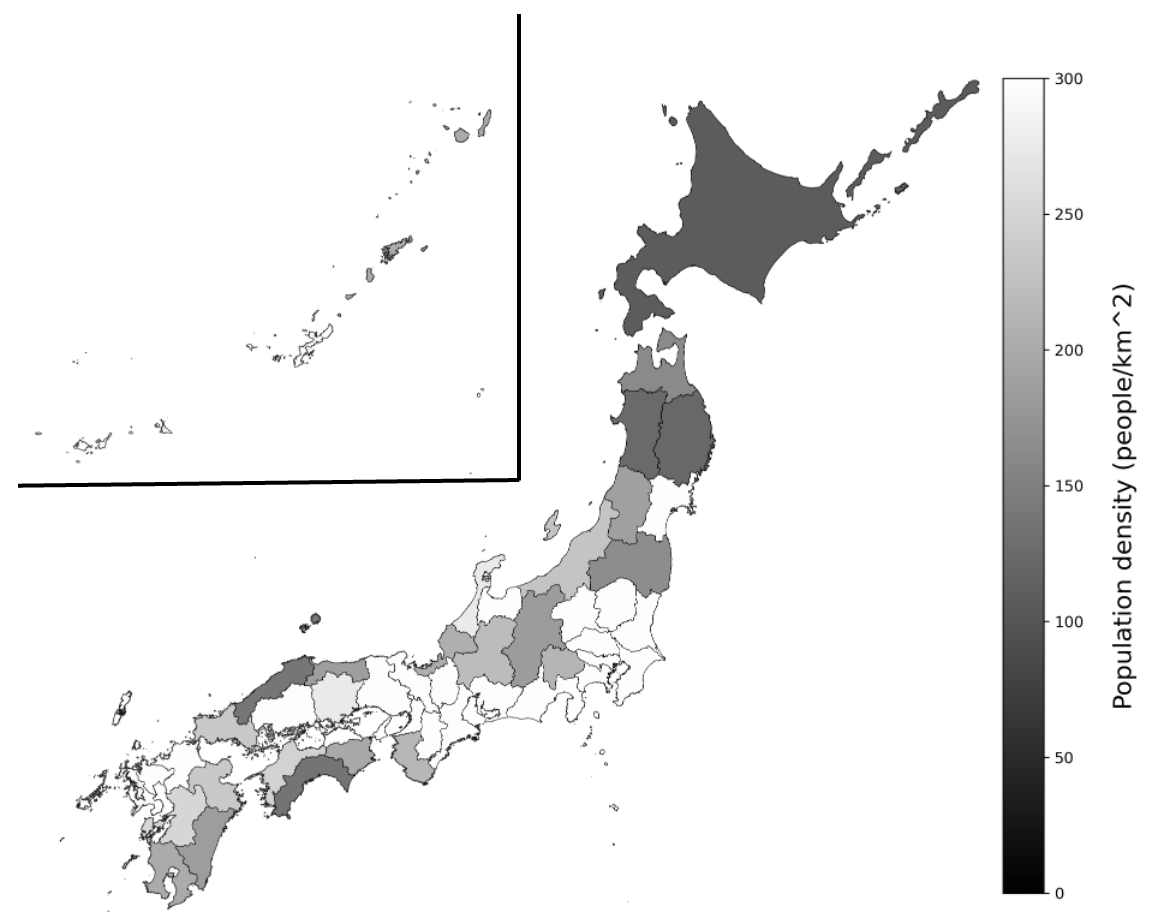}
    \subcaption{Population Density}\label{fig:density}
\end{minipage}
\caption{Fulfillment Rate and Population Density}\label{fig:match_rate}
\end{figure}

Despite the overall presence of unoccupied seats, as shown in Table \ref{table:data_summary_stats}, the fulfillment rate by prefecture, defined as the ratio of the number of matches to the total number of positions in each prefecture, exhibits substantial regional variation. Figure~\ref{fig:fulfillment} displays a choropleth map of fulfillment rates, while Figure~\ref{fig:density} illustrates population densities by prefecture in 2019. These figures suggest that rural areas, characterized by lower population density, are less popular and experience lower fulfillment rates, even with the current tight caps on urban areas.

Table \ref{table:data_summary_covariates} presents descriptive statistics for the characteristics of medical schools and hospitals.
Panel A summarizes the medical school variables. Medical schools are first categorized as either private or public. There are 51 public medical schools and 27 private ones. The T-score and graduation exam pass rate are used as proxies for student quality, with higher values indicating greater ability. Public medical schools tend to have higher average T-scores, suggesting that graduating from a public institution signals stronger student qualifications. This difference is statistically significant.

Panel B summarizes the hospital-side variables. Hospitals are classified in two ways. First, based on their affiliation with a university, they are categorized as either university or non-university hospitals. Approximately 88.3\% of hospitals have no university affiliation. University hospitals tend to have more beds on average, reflecting their typical role as regional flagship hospitals.
The second classification is based on location. Hospitals situated in one of six prefectures—Tokyo, Kanagawa, Aichi, Kyoto, Osaka, and Fukuoka—where official caps on the number of matches are imposed under the JRMP, are defined as being in urban areas. The remaining hospitals are considered rural. Although this definition includes only six of the forty-seven prefectures in total, 33.7\% of all hospitals fall into the urban category. Urban hospitals also tend to be larger in size.

We examine the total variation in salaries in the JRMP market. The latter half of Panel B reports significant variation in the salaries of medical interns across Japan. For example, in 2017, the average annual salary was approximately \$31,714, with a standard deviation of \$8,166.\footnote{The conversion from yen to dollars was based on the exchange rate as of August 30, 2024.} For comparison, Table 1 of \cite{agarwal_empirical_2015} reports that the mean salary for similar medical interns in the United States was \$47,331, with a standard deviation of \$2,953. While average salaries are higher in the U.S., the standard deviation in Japan is 2.77 times larger, indicating greater salary dispersion. 
This variation reflects salary differentials across regions: on average, rural hospitals offer significantly higher monthly wages—by \$430 (63,760 JPY)—than their urban counterparts. At the same time, affiliation with a university is associated with significantly lower pay, with university hospitals paying \$692 (101,000 JPY) less on average.

\begin{table}[t]
    \centering
    \small
    \caption{Summary Statistics of Medical‐School and Hospital Covariates}
    \label{table:data_summary_covariates}
    \begin{threeparttable}
        \begin{tabular}{@{}p{0.36\textwidth}lccccc@{}}
            \toprule
            & \textbf{Category} & \textbf{Count} & \textbf{Mean} & \textbf{Std} & \textbf{Min} & \textbf{Max} \\ 
            \midrule
            \multicolumn{7}{l}{\textbf{Panel A. Medical Schools}} \\[0.2em]
            T-score           & Private & 27  & 64.96 &  2.17 & 62  & 72   \\
                               & Public  & 51  & 66.38 &  2.66 & 63  & 74   \\[0.2em]
            \midrule
            \multicolumn{7}{l}{\textbf{Panel B. Hospitals}} \\[0.2em]
            Number of beds    
            & Rural           & 684 & 417.99 & 155.92 &   36 & 1195 \\
                               & Urban           & 348 & 472.48 & 217.84 &   38 & 1379 \\
                               & Non‐university & 911 & 411.08 & 147.73 &  36  & 1097 \\
                               & University      & 121 & 626.74 & 272.77 & 295  & 1379 \\
                               [0.2em]
            Monthly Salary (×100{,}000 JPY)   & Total & 1032 & 3.864 & 0.995 & 1.800 & 8.550\\
            & Rural           & 684 & 4.079   & 0.984   & 1.800 & 8.550 \\
                               & Urban           & 348 & 3.452   & 0.875   & 1.900 & 6.847 \\
                               & Non‐university & 911 & 3.986   & 0.978   & 2.225 & 8.550 \\
                               & University      & 121 & 2.976   & 0.566   & 1.800 & 4.720 \\
            \bottomrule
        \end{tabular}
    \begin{tablenotes}
    \footnotesize
    \item[*] The summary statistics for monthly salary are based on the three-year average from 2017 to 2019 for each of the 1,032 hospitals.
    \end{tablenotes}
    \end{threeparttable}
\end{table}

\subsection{Empirical Matching Pattern}

\label{sec:empirical_pattern}
In line with \cite{agarwal_empirical_2015}, we estimate a set of regressions to characterize the matching patterns between hospitals and students.
Due to the unavailability of individual matching data, we first calculate the weighted averages of the characteristics of matched partners within each medical school and hospital. These aggregated characteristics are then regressed on the covariates of the medical schools and hospitals.
We also regress the unmatch rate of each hospital and school on their characteristics to identify the types of students who are more likely to go unmatched in this market.
We use random effects models to account for the correlation between unobserved error terms and observed covariates in our panel matching data.\footnote{Since our main interest lies in the coefficients on time-invariant covariates, such as the urban hospital dummy and public university dummy, we do not use a fixed effects model.}
We compute cluster-robust standard errors where the cluster is set to the unit level: the hospital level for the hospital perspective and the school level for the school perspective.

Table~\ref{table:agarwal_table2} reports the regression results on matching patterns from the perspective of hospitals. 
The first and second column dependent variables represent the quality of matched students, so the positive correlations with hospital size shown in the first row are in line with expectations. 
The second row shows a notable finding about the role of urban location. It is natural to expect that urban hospitals attract higher-quality students---as urban areas are generally more preferred by students—this pattern is indeed observed for difficulty. However, urban hospitals are less likely to match with highly qualified students from public universities. This pattern suggests that public universities play a significant role in supplying physicians to rural areas.
There is no statistically significant pattern in the distance between the matched students and their universities. As for the unmatch rate, urban hospitals exhibit fewer vacancies, as expected.

\begin{table}[tbp]
\centering
\def\sym#1{\ifmmode^{#1}\else\(^{#1}\)\fi}
\caption{Matching from Hospital Viewpoint}
\label{table:agarwal_table2}
\begin{tabular}{l*{4}{c}}
\hline\hline
            &\multicolumn{1}{c}{(1)}   &\multicolumn{1}{c}{(2)}   &\multicolumn{1}{c}{(3)}   &\multicolumn{1}{c}{(4)}   \\
            &      Public   &  Difficulty   &    Distance   &Unmatch~Rate   \\
\hline
$\ln$ Beds   &      0.0576***&       6.286***&       9.081   &      -0.146***\\
            &    (0.0222)   &     (1.190)   &     (10.29)   &    (0.0227)   \\
[1em]
Urban Hospital     &      -0.113***&       3.179***&       2.837   &      -0.148***\\
            &    (0.0176)   &     (0.822)   &     (9.661)   &    (0.0150)   \\
\hline
Observations&        3072   &        3072   &        3072   &        3072   \\
Year dummy  &     $\surd$   &     $\surd$   &     $\surd$   &     $\surd$   \\
\hline\hline
\multicolumn{5}{l}{\footnotesize Standard errors in parentheses}\\
\multicolumn{5}{l}{\footnotesize * p<0.1, ** p<0.05, *** p<0.01}\\
\end{tabular}
\end{table}

One straightforward explanation for the frequent matching between public universities and rural hospitals is geographic proximity: every prefecture hosts at least one public medical school. At the same time, higher salaries appear to serve as a strong incentive to retain talented students from public universities in rural areas after graduation. 
As shown in Table~\ref{table:agarwal_table3}, which presents the regression results from the perspective of medical schools, graduates from public universities are more likely to match with rural hospitals and receive higher salaries. 
This finding supports the view that salary is a significant factor in students’ preferences to sustain an equilibrium in which rural hospitals offer higher wages to retain talented graduates. 
Table~\ref{table:agarwal_table3} also shows that graduates from more competitive universities tend to match with nearby hospitals, suggesting that distance is an important factor in residents’ choices. As expected, graduates from more prestigious universities are less likely to go unmatched.

\begin{table}[tbp]\centering
\def\sym#1{\ifmmode^{#1}\else\(^{#1}\)\fi}
\caption{Matching from School Viewpoint}
\label{table:agarwal_table3}
\begin{tabular}{l*{4}{c}}
\hline\hline
            &\multicolumn{1}{c}{(1)}   &\multicolumn{1}{c}{(2)}   &\multicolumn{1}{c}{(3)}   &\multicolumn{1}{c}{(4)}   \\
            &Urban Hospital   &  $\ln$ Wage   &    Distance   &Unmatch~Rate   \\
\hline
Public University&      -3.055***&       1.335***&       425.7   &      -1.416** \\
            &     (1.088)   &     (0.391)   &     (539.0)   &     (0.633)   \\
[1em]
Difficulty  &      0.0115   &     0.00452   &      -16.16** &     -0.0220** \\
            &    (0.0146)   &   (0.00433)   &     (6.540)   &   (0.00956)   \\
\hline
Observations&         234   &         234   &         234   &         234   \\
Year dummy  &     $\surd$   &     $\surd$   &     $\surd$   &     $\surd$   \\
\hline\hline
\multicolumn{5}{l}{\footnotesize Standard errors in parentheses}\\
\multicolumn{5}{l}{\footnotesize * p<0.1, ** p<0.05, *** p<0.01}\\
\end{tabular}
\end{table}

\section{Model}
\label{sec:model}

\paragraph{Matching market with regional constraints}
We consider a two-sided matching market.
Let $I$ denote the set of doctors (medical students) and $J$ the set of job slots owned by hospitals.
Each doctor $i \in I$ can be matched with at most one slot $j \in J$, and each slot can accommodate at most one doctor. If a doctor $i$ is unmatched, they are paired with an outside option $j_0$. Similarly, an unmatched slot $j$ is paired with an outside option $i_0$. A \emph{matching} is represented by a $0$-$1$ matrix $d = (d_{ij})_{i \in I,j \in J}$, where $d_{ij} = 1$ if and only if doctor $i$ is matched with slot $j$. A matching $d$ is \emph{feasible} if each doctor is matched to exactly one slot or the outside option, and each slot is matched to exactly one doctor or the outside option: $\sum_{j \in J} d_{ij} \leq 1$ for all $i \in I$, and $\sum_{i \in I} d_{ij} \leq 1$ for all $j \in J$.

There is a policymaker who faces an additional condition called regional constraints. There are $L$ regions, denoted by $Z = \{ z_1, z_2, \ldots, z_L \}$, with each job slot $j$ assigned to one region. Additionally, we define a special region $z_0$ that contains only the outside option $j_0$. With a slight abuse of notation, let ${z} \colon J \cup \{j_0\} \to Z \cup \{z_0\}$ be the mapping where $z(j)$ indicates the unique region to which job slot $j$ belongs. Each region $z$ has a cap and a floor, $\underline{o}_z \in \mathbb{R}_+$ and $\bar{o}_z \in \mathbb{R}_+ \cup \{\infty\}$, where $\uo_z < \bo_z$ for each $z$. We say a feasible matching $d$ satisfies \emph{regional constraints} if it respects the caps and the floors: $\sum_{i \in I} \sum_{j \in z} d_{ij} \in [\underline{o}_z, \bar{o}_z]$ for each $z \in Z$. Throughout the paper, we assume that at least one feasible matching exists that satisfies regional constraints.\footnote{Formally, we assume
$\uo_z \leq |\{j \colon j \in z\}|$ for each $z \in Z$, and $|I| \geq \sum_z \uo_z.$
}

\paragraph{Stable outcome under taxation policy}
Agents form a stable outcome \`a la \cite{shapley_1971}. Without policy intervention, the realized matching may not meet the regional constraints. The policymaker can implement a taxation policy that influences the split of the joint surplus among agents to satisfy the regional constraints. When a doctor $i$ and a slot $j$ are matched, they generate an \emph{(individual-level) net joint surplus} $\Phi_{ij} \in \mathbb{R}$. Denote by $\Phi_{i, j_0}$ and $\Phi_{i_0, j}$ the payoffs $i \in I$ and $j \in J$ receive when unmatched, respectively. The tax $w_z \in \mathbb{R}$ is imposed on each match $(i,j)$ in region $z$, with negative taxes being interpreted as subsidies. We assume $w_{z_0} = 0$, i.e., no tax is imposed on the outside option. With taxation policy $w = (w_z)_{z \in Z}$, each matched pair divides the \emph{gross joint surplus} $\Phi_{ij} - w_{z(j)}$ instead of the net joint surplus.
\footnote{In principle, the policymaker could set pair-specific taxes. We restrict our definition of a taxation policy to be uniform within each region; however, this is without loss of generality: the welfare-maximizing taxation policy is, in fact, uniform, imposing the same tax on all pairs within a given region (see Section~\ref{sec:limit_cap_based} for the related discussion and Appendix~\ref{app:taxation_possibility} for the formal proof).}
The stable outcome under a taxation policy is defined as follows:
\begin{definition}[Stable outcome]
\label{def:indiv_eqm}
Given the matching market $(I,J,Z, z,\Phi)$,\footnote{The symbol $z$ denotes the mapping from job slots to regions.} a profile $(d, (u, v))$ of feasible matching $d$ and equilibrium payoffs $(u, v)$ forms a \emph{stable outcome} under taxation policy $w$ if it satisfies:
\begin{enumerate}
\item Individual rationality: For all $i \in I$, $u_i \geq \Phi_{i,j_0}$, with equality if $i$ is unmatched. For all $j \in J, v_j \geq \Phi_{i_0,j}$, with equality if $j$ is unmatched.
\item No blocking pairs: For all $i \in I$ and $j \in J$, $u_i + v_j \geq \Phi_{ij} - w_{z(j)}$, with equality if $d_{ij}=1$.
\end{enumerate}
\end{definition}

While it is reasonable to assume that participants form a stable outcome in many frictionless and decentralized environments, such as certain labor or marriage markets, the appropriateness of this assumption depends on the specific application. In Appendix~\ref{app:JRMP_game}, we show that any stable outcome can be supported as an equilibrium of a stylized game that mirrors the JRMP's matching process.

A \emph{cap-based policy} is represented by removing a subset of slots from the market. Let \( J' \subseteq J \) denote the set of slots that remain after the policy is applied; this subset characterizes the cap-based policy. Operationally, this removal is equivalent to imposing an upper bound (a cap) on the number of slots. Let \( \Phi' \coloneqq (\Phi_{ij})_{i \in I,\, j \in J'} \) denote the submatrix of \( \Phi \) corresponding to the remaining doctors and slots. Given \( J' \), the agents form a stable outcome in the reduced market \( (I, J', Z, z, \Phi') \) under the assumption that taxation policy is absent, that is, \( w \equiv 0 \).

\paragraph{Unobserved heterogeneity}
Let $X = \{ x_1, x_2, \ldots, x_N \}$ represent the finite set of observable characteristics, or \emph{types}, of doctors. Each doctor $i \in I$ has a type $x(i) \in X$. Similarly, let $Y = \{ y_1, y_2, \ldots, y_M \}$ represent the finite set of observable characteristics of job slots, with each slot $j \in J$ having a type $y(j) \in Y$. Although agents with the same type are indistinguishable to the policymaker, there can be \emph{unobservable heterogeneity}: doctors of the same type $x$ or job slots of the same type $y$ may generate different joint surpluses when matched. For convenience, we denote $i \in x$ if $x(i) = x$ and $j \in y$ if $y(j) = y$. We define $x_0$ and $y_0$ as special types representing the outside options $i_0$ and $j_0$, respectively, and let $X_0 = X \cup \{ x_0 \}$ and $Y_0 = Y \cup \{ y_0 \}$ include these outside options. The set of all type pairs is denoted by $T = X_0 \times Y_0 \setminus \{ (x_0, y_0) \}$. We assume each job slot type $y \in Y$ belongs to a unique region, denoted by $z(y) \in Z$. Let $n_x$ be the number of doctors with type $x$, and $m_y$ be the number of job slots with type $y$.

Let $\mu_{xy}$ denote the number of matches between type-$x$ doctors and type-$y$ job slots, defined as $\mu_{xy} = \sum_{i \in x} \sum_{j \in y} d_{ij}$. An \emph{aggregate-level matching} $\mu = (\mu_{xy})_{x \in X, y \in Y}$ is said to be \emph{feasible} if it satisfies the population constraints $\sum_y \mu_{xy} = n_x$ and $\sum_x \mu_{xy} = m_y$ for each $x$ and $y$. Furthermore, we say $\mu$ satisfies regional constraints if $\sum_{y \in z}\sum_{x \in X} \mu_{xy} \in [\underline{o}_z, \bar{o}_z]$ for each $z \in Z$.

Types $y \in Y$ and regions $z \in Z$ can be interpreted in various ways. For example, in our application, a type $y$ may correspond to a hospital, and a region $z$ may correspond to a district (e.g., a prefecture). In other contexts, a type could represent a subcategory of occupation (e.g., registered nurse, physician assistant), and a region could represent a broader occupational category (e.g., healthcare).

\section{Theoretical Results}
\label{sec:theoretical_results}

\subsection{Relative Efficiency of Caps and Taxation}
\label{sec:limit_cap_based}

To build intuition for our main analysis, this subsection begins with a stylized model that abstracts from unobserved heterogeneity. In this benchmark case, we provisionally assume that the policymaker knows all joint surplus values, \( \Phi_{ij} \). While this assumption is relaxed later, it clearly illustrates the properties of the \emph{optimal taxation policy}, which is defined below, and its relationship to cap-based policies, which remain valid even in the presence of unobserved heterogeneity.

Given a matching \( d \), the \emph{social surplus} is defined as \( \sum_{i,j} d_{ij} \Phi_{ij} \), the total net joint surplus generated by the matching \( d \). Consider the following social surplus maximization problem subject to regional constraints:
\begin{optproblem}{$P_0$}
\label{opt:IE_primal}
\begin{alignat}{3}
 &\underset{d \in \{0,1\}^{I \times J}}{\text{maximize}}  & \quad & \sum_{(i,j) \in I } \Phi_{ij} d_{ij} + \sum_{i \in I} \left(1 - \sum_{j \in J} d_{ij}\right) \Phi_{i, j_0} +  \sum_{j \in I} \left(1 - \sum_{i \in I} d_{ij}\right) \Phi_{i_0, j} &  \\
 &\text{subject to} & \quad & \sum_{j \in J} d_{ij} \leq 1                                &\forall i \in I,\\
 &                  & \quad & \sum_{i \in I} d_{ij} \leq 1                               &\forall j \in J,\\
 &                  & \quad & \uo_z \leq \sum_{j \in z} \sum_{i \in I} d_{ij} \leq \bo_z  &\forall z \in Z, 
\end{alignat}
\end{optproblem}
Let \( d^* \) denote a solution to \(\mathrm{(P_0)}\); that is, \( d^* \) is the matching that maximizes social surplus subject to the regional constraints.

Given full knowledge of the joint surplus matrix \( \Phi \), there exists an \emph{optimal} taxation policy that implements \( d^* \) as a stable outcome. This implies that, within our framework, the optimal taxation policy outperforms any cap-based policy: no cap-based policy that satisfies the same regional constraints can generate a higher social welfare than the optimal taxation policy. The taxes and subsidies are characterized by the Lagrange multipliers on the binding regional constraints and can be computed via linear programming. (see Appendix~\ref{app:taxation_possibility} for details.) A key feature of this optimal taxation policy is its structure: it consists of a single tax or subsidy $w_z$ for each region $z \in Z$, and the same amount is applied uniformly to all matches within a region, despite the theoretical possibility of setting pair-specific taxes. This significantly simplifies implementation.

The following example illustrates that the relative efficiency of a cap-based policy, measured as the ratio of its social welfare to the optimal taxation benchmark, can take any value in $(0,1)$, depending on the model's parameters. This underscores the need for empirical analysis to quantify the actual welfare loss in any given application.

\begin{example}
Consider two doctors, \( i_1 \) and \( i_2 \); two regions, \( z_1 \) and \( z_2 \); and two slots, \( j_1 \in z_1 \) and \( j_2 \in z_2 \). Region \( z_2 \) is subject to a floor constraint \( \underline{o}_{z_2} = 1 \), while there are no other regional constraints. The value of outside options is normalized to zero, and the joint surplus matrix is given by  
\[
\Phi \coloneqq
    \kbordermatrix{
    & j_1 & j_2 \\
    i_1 & a & \delta \\
    i_2 & b & -\epsilon
    },
\]
where \( a, b, \delta, \epsilon > 0 \) and \( a \geq \delta + b + \epsilon \).

\paragraph{No intervention}
The stable matching pairs \( i_1 \) with \( j_1 \), leaving \( i_2 \) unmatched. The resulting social welfare is \( a \), but the regional constraint is violated.

\paragraph{Cap-based policy}
The policymaker meets the constraint by removing the urban slot, $j_1$. The reduced surplus matrix becomes
\[
\Phi' \coloneqq
    \kbordermatrix{&j_2\\
        i_1 & \delta \\
        i_2 & -\epsilon
    }
    .
\]
In this case, \( i_1 \) is matched with \( j_2 \), yielding social welfare \( \delta \).

\paragraph{Optimal subsidy policy}
The policymaker retains both slots and subsidizes matches in region \( z_2 \) by \( \epsilon \). The resulting gross surplus matrix is
\[
\kbordermatrix{
    &j_1 & j_2 \\
    i_1 & a & \delta + \epsilon \\
    i_2 & b & 0
}.
\]
The matching $i_1$ to $j_1$ and $i_2$ to $j_2$ forms a stable outcome, achieving social welfare of \( a - \epsilon \).

\paragraph{Comparison}
The ratio \( \delta / (a - \epsilon) \) compares the social welfare under the cap-based policy to that under the taxation policy. This ratio can take any value in \( (0, 1] \), while satisfying the condition $a \geq \delta + b + \epsilon$.\footnote{For example, the ratio equals $1$ if \( \delta = a - \epsilon \), \( b = \epsilon \), and $\epsilon \leq a/2$.
In contrast, suppose that $\delta = \epsilon$ and $b \leq a - 2\epsilon$. As \( \delta \) and \( \epsilon \) approach zero, the ratio converges to zero.}

\end{example}
This example also illustrates how the inefficiency of a cap-based policy arises. As a blunt instrument, the cap satisfies the rural constraint by eliminating the high-surplus urban slot, $j_1$. This forces the most productive doctor, $i_1$, into a significantly less-valued match, sacrificing the large surplus generated by the $(i_1, j_1)$ pair. A targeted subsidy, in contrast, achieves the goal more efficiently. By making the undesirable rural match, $(i_2, j_2)$, just palatable enough for the marginal doctor, $i_2$, it fulfills the constraint while leaving the most productive match in the market, $(i_1, j_1)$, undisturbed.

\subsection{Preliminary Results}

To account for unobserved heterogeneity in preferences, this subsection reviews the framework of \cite{galichon_cupids_2021}, which links the individual-level objects,
$(\Phi_{ij})_{ij}$, $(d_{ij})_{ij}$, $(u_i)_i$, and $(v_j)_j$, 
introduced in Section~\ref{sec:model} to their aggregate-level counterparts required for our analysis.

\paragraph{Discrete Choice Representation} 
For any pair $(i,j)$ with doctor $i$ of type $x$ and slot $j$ of type $y$, we assume that the individual-level joint surplus decomposes as
\[
\Phi_{ij} = \Phi_{xy} + \epsilon_{iy} + \eta_{xj},
\]
where $\Phi_{xy}$ is the \emph{aggregate-level joint surplus} and $\epsilon_{iy}$ and $\eta_{xj}$ are independent error terms. For each type $x$ and every doctor $i \in x$, the vector $(\epsilon_{iy})_{y \in Y_0}$ is drawn from a distribution $P_x \in \Delta(\R^{|Y|+1})$. Similarly, for each type $y$ and every slot $j \in y$, the vector $(\eta_{xj})_{x \in X_0}$ is drawn from a distribution $Q_y \in \Delta(\R^{|X|+1})$.

\begin{ass}[Independence]
\label{ass:indep_error}
The error terms are independent across all doctors and slots.
\end{ass}

\begin{ass}[Additive Separability]
\label{ass:additive_separability}
There exists a matrix $\left(\Phi_{xy}\right)_{(x,y) \in T}$ such that:
\begin{enumerate}
    \item For every $x \in X$, $y \in Y$, $i \in x$, and $j \in y$, 
    \[
    \Phi_{ij} = \Phi_{xy} + \epsilon_{iy} + \eta_{xj};
    \]
    \item For every $x \in X$ and $y \in Y$, 
    \[
    \Phi_{i,y_0} = \epsilon_{i,y_0} \quad \text{and} \quad \Phi_{x_0,j} = \eta_{x_0,j}.
    \]
\end{enumerate}
\end{ass}

We define the \emph{aggregate-level utilities} $U_{xy}$ and $V_{xy}$, which depend solely on types, as follows. For each $x \in X$ and $y \in Y$, let
\begin{align}
\label{eq:agg_level_util}
U_{xy} \coloneqq \min_{i \,:\, x(i)=x} \{ u_i - \epsilon_{iy} \}, \quad V_{xy} \coloneqq \min_{j \,:\, y(j)=y} \{ v_j - \eta_{jx} \},
\end{align}
with the normalization $U_{x,y_0} = V_{x_0,y} \coloneqq 0$. The following lemma shows that the matching market can be represented as a bilateral discrete choice problem.

\begin{lemma}[\cite{galichon_cupids_2021}]
\label{lemma:tu_discrete_choice}
Let $(u,v)$ be a payoff profile corresponding to a stable outcome. Under Assumption~\ref{ass:additive_separability}, for any doctor $i \in I$ and any slot $j \in J$, we have
\begin{align}
\label{eq:tu_util_equiv}
u_i &= \max_{y \in Y_0} \left\{ U_{x(i),y} + \epsilon_{iy} \right\}, \quad
v_j = \max_{x \in X_0} \left\{ V_{x,y(j)} + \eta_{xj} \right\}.
\end{align}
\end{lemma}
\begin{proof}
See Appendix~\ref{app:tu_discrete_choice}.
\end{proof}

\paragraph{Large Market Approximation} 
By Lemma~\ref{lemma:tu_discrete_choice}, the social welfare on the doctor side, defined as the sum of the equilibrium payoffs, can be expressed as
\[
\sum_{i \in I} u_i = \sum_{x \in X} n_x \left( \frac{1}{n_x} \sum_{i \in x} \max_{y \in Y_0} \left\{ U_{xy} + \epsilon_{iy} \right\} \right).
\]
When the number of doctors of type $x$, denoted by $n_x$, is large, the average
\[
\frac{1}{n_x} \sum_{i \in x} \max_{y \in Y_0} \left\{ U_{xy} + \epsilon_{iy} \right\}
\]
is well approximated by its expected value: $\E_{\epsilon_i \sim P_x} \left[ \max_{y \in Y_0} \left\{ U_{xy} + \epsilon_{iy} \right\} \right]$.
Thus, the social welfare on the doctor's side becomes
\[
G(U) \coloneqq \sum_{x \in X} n_x\, \E_{\epsilon_i \sim P_x} \left[ \max_{y \in Y_0} \left\{ U_{xy} + \epsilon_{iy} \right\} \right].
\]
Similarly, if the number of slots of type $y$, denoted by $m_y$, is large, the social welfare on the hospital side is approximated by
\[
H(V) \coloneqq \sum_{y \in Y} m_y\, \E_{\eta_j \sim Q_y} \left[ \max_{x \in X_0} \left\{ V_{xy} + \eta_{xj} \right\} \right].
\]
Under the assumption that the cumulative distribution functions of the error terms are continuously differentiable, the Williams-Daly-Zachary theorem \citep{McFadden1980-af} implies that
\[
\frac{\partial G}{\partial U_{xy}}(U) = n_x \Pr\!\left(\text{doctor of type } x \text{ chooses a slot of type } y\right).
\]
For sufficiently large $n_x$, the product $n_x\, \Pr\!\left(\text{doctor of type } x \text{ chooses a slot of type } y\right)$ provides a good approximation of $\mu_{xy}$.

\begin{ass}[Smooth Distribution]
\label{ass:smooth_cdf}
For each $x$ and $y$, the cumulative distribution functions $P_x$ and $Q_y$ are continuously differentiable.
\end{ass}

\subsection{Design and Evaluation of Policies with Unobserved Heterogeneity}

\label{sec:eae}
An aggregate-level matching market is characterized by the tuple
\[
\mM \coloneqq (X, Y, n, m, Z, z, \Phi, P, Q),
\]
where $n \coloneqq (n_x)_x$, $m \coloneqq (m_y)_y$, $\Phi = (\Phi_{xy})_{xy}$, $P = (P_x)_x$, and $Q = (Q_y)_y$. The policymaker, facing unobserved heterogeneity, seeks (i) to determine the optimal taxation policy $w$ that maximizes social welfare subject to regional constraints and (ii) to compute the matching and social welfare outcomes for various policies. For these purposes, we formulate optimization problems that depend solely on the aggregate-level surplus $(\Phi_{xy})_{x,y}$. An additional technical assumption on the error term distributions guarantees that $G$ and $H$ are strictly convex,\footnote{\label{fn:conjugacy_mu_UV}See Appendix~\ref{app:proof_thm_eae} for details.} thereby ensuring a one-to-one correspondence between an aggregate-level matching $\mu$ and the aggregate-level utilities $U$ and $V$ given taxation policy $w$.

\begin{ass}[Full Support]
\label{ass:full_supp}
For each $x$ and $y$, $\mathrm{supp}(P_x) = \R^{|Y_0|}$ and $\mathrm{supp}(Q_y) = \R^{|X_0|}$.
\end{ass}

We now define the following concave programming problem, which is a counterpart of the social surplus maximization problem $\mathrm{(P_0)}$ for the case without unobserved heterogeneity:
\begin{optproblem}{P}
\label{opt:natural_agg_eq_primal}
\begin{alignat}{3}
 &\underset{\mu \geq 0}{\text{maximize}}  & \quad & \sum_{(x, y) \in T } \mu_{xy} \Phi_{xy} + \mathcal{E}(\mu)&  \\
 &\text{subject to} & \quad & \sum_{y \in Y_0} \mu_{xy} = n_x                                 &\forall x \in X,\\
 &                  & \quad & \sum_{x \in X_0} \mu_{xy} = m_y                               &\forall y \in Y,\\
 &                  & \quad & \uo_z \leq \sum_{y \in z} \sum_{x \in X} \mu_{xy} \leq \bo_z  &\forall z \in Z, 
\end{alignat}
\end{optproblem}
where
\[
\mathcal{E}(\mu) \coloneqq -G^*(\mu) - H^*(\mu),
\]
and $G^*$ and $H^*$ denote the Legendre-Fenchel transforms of $G$ and $H$, respectively.\footnote{Since $G$ and $H$ are proper convex functions, their Legendre-Fenchel transforms are well-defined.} Its dual problem is given by
\begin{optproblem}{D}
\label{opt:natural_agg_eq_dual}
\begin{alignat}{2}
\underset{U, V, \bw_z, \uw_z}{\text{minimize}} \quad & G(U) + H(V) + \sum_{z \in Z} \bo_z \bw_z - \sum_{z \in Z} \uo_z \uw_z \ \\
 \text{subject to}\quad &U_{xy} + V_{xy} \geq \Phi_{xy} - \bw_{z(y)} + \uw_{z(y)} &\forall (x, y) \in T, \\
 &\bw_z \geq 0, \ \uw_z \geq 0 &\forall z \in Z.
\end{alignat}
\end{optproblem}

The primal problem \optprobref{opt:natural_agg_eq_primal} admits an optimal solution because its objective function is continuous and its feasible set is compact and nonempty. By strong duality,\footnote{The constraints of the primal problem satisfy the weak Slater condition, as they are affine in $\mu$.} the dual problem \optprobref{opt:natural_agg_eq_dual} attains the same optimal value. Moreover, both \optprobref{opt:natural_agg_eq_primal} and \optprobref{opt:natural_agg_eq_dual} possess unique solutions, as we show below.

The optimal value of \optprobref{opt:natural_agg_eq_primal} and \optprobref{opt:natural_agg_eq_dual} represents the social surplus generated under equilibrium matching $\mu$. \footnote{For any $w$, we can first rewrite the dual problem of $(P_0)$ using the aggregate level objects. We can show that $(D_w)$ defined in the statement of Proposition~\ref{prop:ae} coincides with the dual problem under the large market approximation (see Section~1.4 of \cite{Galichon_poisson} for more details.) Then, this statement follows from the second part of Theorem~\ref{theorem:EAE_opt}. See also Appendix~\ref{sec:comparison} for more details about the large market approximation.} The following theorem states that, given regional constraints $(\bo_z, \uo_z)_z$, the optimal taxation policy $w^*$ together with the corresponding aggregate-level matching and utilities $(\mu, U, V)$ are characterized by the solutions to \optprobref{opt:natural_agg_eq_primal} and \optprobref{opt:natural_agg_eq_dual}.

\begin{theorem}
\label{theorem:EAE_opt}
Assume that Assumptions \ref{ass:indep_error}--\ref{ass:full_supp} hold. Fix any aggregate-level matching market $\mM$. Then,
\begin{enumerate}
    \item There is a unique solution to \optprobref{opt:natural_agg_eq_primal} and \optprobref{opt:natural_agg_eq_dual}, denoted by $\mu^*$ and $(U^*, V^*, \bar{w}^*, \underline{w}^*)$, respectively.
    \item Define, for each $z$, 
    \[
    w_z^* \coloneqq \mathbbm{1}\{\bar{w}_z^* > 0\}\bar{w}_z^* - \mathbbm{1}\{\underline{w}_z^* > 0\}\underline{w}_z^*.
    \]
    Then, $w^*$ is the optimal taxation policy, and $(\mu^*, U^*, V^*)$ constitute the corresponding aggregate-level matching and utilities.
\end{enumerate}
\end{theorem}
\begin{proof}
See Appendix~\ref{app:proof_thm_eae}.
\end{proof}

We can also obtain the aggregate-level matching and utilities $(\mu, U, V)$ corresponding to any given taxation policy $w$ by solving the following pair of optimization problems.

\begin{prop}
\label{prop:ae}
Assume that Assumptions \ref{ass:indep_error}--\ref{ass:full_supp} are satisfied. For any aggregate-level matching market $\mM$ and any taxation policy $w$, the resulting aggregate-level matching and utilities, denoted by $(\mu(w), U(w), V(w))$, are characterized as the solutions to
\begin{optproblem}{$P_w$}
\label{opt:w_agg_eq_primal}
\begin{alignat}{3}
\text{maximize}_{\mu \geq 0} \quad & \sum_{(x,y) \in T} \mu_{xy} \Bigl( \Phi_{xy} - w_{z(y)} \Bigr) + \mathcal{E}(\mu) \\
\text{subject to} \quad & \sum_{y \in Y_0} \mu_{xy} = n_x, &&\quad \forall x \in X, \\
& \sum_{x \in X_0} \mu_{xy} = m_y, &&\quad \forall y \in Y,
\end{alignat}
\end{optproblem}
and
\begin{optproblem}{$D_w$}
\label{opt:w_agg_eq_dual}
\begin{alignat}{2}
\text{minimize}_{U,V} \quad & G(U) + H(V) \\
\text{subject to} \quad & U_{xy} + V_{xy} \geq \Phi_{xy} - w_{z(y)}, \quad &&\forall (x,y) \in T.
\end{alignat}
\end{optproblem}
\end{prop}

For a given taxation policy $w$, the tuple $(\mu(w), U(w), V(w))$ defined in Proposition~\ref{prop:ae} is referred to as an \emph{aggregate equilibrium (AE)} under $w$. When $w$ is optimal, the equilibrium is termed the \emph{efficient aggregate equilibrium (EAE)}.

\section{Empirical Strategy}
\label{sec:empirical_strategy}
We begin by mapping the primitives and equilibrium objects in the model described in Section \ref{sec:model} and \ref{sec:theoretical_results} to their empirical counterparts in the doctor-hospital matching market. 
We use the following notation:\footnote{Type $x$ and $y$ in previous sections are replaced by $s$ (school) and $h$ (hospital), respectively.} a doctor is denoted by \(i\), and a job slot is denoted by \(j\). Each doctor belongs to a medical school, and each slot is offered by a hospital, where \(s\) represents a school and \(h\) represents a hospital. We consider \(s\) and \(h\) as observable types of doctors and slots, using \(s(i)\) to denote the medical school to which doctor \(i\) belongs, and \(h(j)\) to denote the hospital offering slot \(j\). The matching market operates over \(T \in \mathbb{N}\) periods, with \(t\) denoting each observation period. Let \(Z\) denote the set of regions and \(z(h)\) denote the region to which hospital \(h\) belongs, assuming \(z(h)\) remains constant over time. The aggregate-level joint surplus at time \(t\) is denoted by \(\Phi_{sht}\). The unobserved part (error term) of doctor \(i\)'s preference for hospital \(h\) is denoted by \(\epsilon_{iht}\), while the unobserved part of slot \(j\)'s preference for school \(s\) is denoted by \(\eta_{sjt}\).\footnote{In our model, different slots in a hospital may have different preferences. This is possible when, for example, the admission office is composed of members with varying tastes for schools, such as a strong preference for the school from which they graduated. \(\eta_{sjt}\) reflects these differences in the committee members' preferences.} The net joint surplus satisfies the equality: \(\Phi_{ijt} = \Phi_{s(i)h(j)t} + \epsilon_{ih(j)t} + \eta_{s(i)jt}\) for each \(i \in I\), \(j \in J\), and \(t \in [T]\).\footnote{\([T] \coloneqq \{1, 2, \dots, T-1, T\}\).}
The matching \((d_{ijt})_{i,j,t}\) is not observable. Instead, the available data comprises the aggregate-level matching \((\mu_{sht})_{s,h,t}\), which is the number of matches between medical school \(s\) and hospital \(h\) at time \(t\).

We suppose that the observed numbers of matches in a year constitute the aggregate equilibrium of the matching market of the year under no tax.
\footnote{In other words, we do not consider the current matching outcome to be affected by any monetary intervention aimed at satisfying regional constraints. In Appendix \ref{sec:appendix_AE_EAE}, we check if this assumption is valid in our data. We estimate our model under the assumption that the matching outcome forms the efficient aggregate equilibrium under a regional constraint. The estimated value of tax levied does not increase, whereas the regional constraint set by the policymaker is gradually tightened. This result indicates that ``implicit'' taxation has never been implemented during our data period.}
Hence, following \cite{galichon_cupids_2021}, we can identify and estimate the aggregate-level utilities of both sides by polynomial functions. This step is the first stage of our estimation described in Section~\ref{sec:strategy_first}.
Additionally, we use the data on doctors’ salaries paid by hospitals to evaluate welfare in monetary terms. For this objective, we impose an additional structure on the composition of the net joint surplus. We then introduce the concept of transfers between matched pairs in the individual-level market and define an aggregate-level transfer used in our estimation. In Section~\ref{sec:strategy_second}, we describe the second step of our estimation, which includes a measurement model for connecting the salary to the aggregate-level transfer.


\subsection{Transfer}
To define a model object corresponding to the observed salary, we impose an additional structure on the net joint surplus. 
The \emph{base utility} of $i$ when matched with $j$ at time $t$ is the utility felt by $i$ when matching with $j$ net of transfer, and is denoted by $U^{\mathrm{base}}_{ijt}$. Similarly, the base utility of $j$ when matched with $i$ at time $t$ is denoted by $V^{\mathrm{base}}_{ijt}$.
The net joint surplus is the sum of the base utilities of a doctor and a slot forming the match, i.e., 
\begin{align}
\label{eq:utility_structure}
\Phi_{ijt} = U^{\mathrm{base}}_{ijt} + V^{\mathrm{base}}_{ijt}.
\end{align}

Furthermore, in accordance with additive separability (Assumption \ref{ass:additive_separability}), we re-interpret the i.i.d.~error terms as the unobserved taste shocks of the agents on both sides of the market. In other words, we assume the following utility structure:
\begin{ass}
    \label{ass:surplus_composition}
\begin{align}
        U^{\mathrm{base}}_{ijt} = U^{\mathrm{base}}_{sht} + \epsilon_{iht}, \quad 
        V^{\mathrm{base}}_{ijt} = V^{\mathrm{base}}_{sht} + \eta_{sjt}.
\end{align}
\end{ass}


We call $U^{\mathrm{base}}_{sht}$ and $V^{\mathrm{base}}_{sht}$ by \emph{aggregate-level base utility}: they are a part of base utilities, which are determined by the observable characteristics.
As a direct implication of Assumption \ref{ass:additive_separability},  \ref{ass:surplus_composition} and \eqref{eq:utility_structure}, we have $\Phi_{sht} = U^{\mathrm{base}}_{sht} + V^{\mathrm{base}}_{sht}$.
Note that the aggregate-level base utility $U^{\mathrm{base}}_{sht}$ can be different from the aggregate-level utility $U_{sht}$ introduced in \eqref{eq:agg_level_util}.

Fix any period $t$. Consider a matched pair $(i,j)$ with $h(j)=h$ and $z(h(j)) = z$ for some $h$ and $z$.
We define \textit{individual-level transfer from hospital $h$ to doctor $i$}, denoted by $\tau_{iht}$, as follows:
\begin{align}
    \label{def:transfer}
        \tau_{iht} \coloneqq u_{it} - U_{ijt}^{\mathrm{base}}.
\end{align} 
In equilibrium, doctor $i$ enjoys equilibrium payoff $u_{it}$, which could be different from $U_{ijt}^{\mathrm{base}}$. We interpret the difference between equilibrium payoff and base utility as the individual-level transfer from the hospital side to the doctor side.

Now we define an \emph{aggregate-level transfer from hospital $h$} as the average of the individual-level transfer in a hospital $h$ and denote it by $\iota_{ht}$:
\begin{align}
    \iota_{ht} \coloneqq \frac{1}{\left|D(h)_t \right|} \sum_{i\in D(h)_t}  \tau_{iht},
\end{align}
where $D(h)_{t}$ is the set of doctors matched with any slot of hospital $h$ at time $t$.
We can show the following identities: the aggregate-level transfer from a hospital is equal to the weighted average of the gap between the aggregate-level utility and the aggregate-level base utility. We use these identities as moment conditions to identify the aggregate-level base utility.
\begin{prop}
\begin{align}
\label{eq:identity}
    \iota_{ht} = \sum_s \omega_{sht}\left( U_{sht} - U^{\mathrm{base}}_{sht} \right),\ \iota_{ht} = \sum_s \omega_{sht}\left( V^{\mathrm{base}}_{sht} - V_{sht} \right)
\end{align}
where $\omega_{sht} = \frac{\mu_{sht}}{\sum_{s'} \mu_{s' ht}}$.
\end{prop}

\subsection{Estimation}
Based on the observable characteristics of $s$ and $h$, we have a set of variables related to the preferences: we use $X^{U,\mathrm{base}}_{sht}$ as the variables for $U^{\mathrm{base}}_{sht}$, and $X^{V,\mathrm{base}}_{sht}$ as the variables for $V^{\mathrm{base}}_{sht}$.
We assume linear structure on both of the preferences: $U^{\mathrm{base}}_{sht} = X_{sht}^{U,\mathrm{base}\prime}\beta_U,\ V^{\mathrm{base}}_{sht} = X_{sht}^{V,\mathrm{base}\prime}\beta_V$.
Our parameters of interest are $\beta_U$ and $\beta_V$. We use $\theta$ to indicate the vector of these parameters: $\theta \coloneqq \left(\beta_U, \beta_V \right)$.

Our estimation consists of the following two steps:
\begin{enumerate}
    \item Estimate the aggregate-level utilities $U_{sht}$ and $V_{sht}$ for every $t$, and then
    \item Estimate $\theta$ using the estimated aggregate-level utilities and the observed salaries.
\end{enumerate}

\subsubsection{First Step}
\label{sec:strategy_first}
Despite the nonparametric identification results obtained in \cite{galichon_cupids_2021}, we estimate the parametrized version of the aggregate-level utilities. This is because some pairs of schools and hospitals have zero matches in practice.
Hence, in the first step, we use the moment matching estimator proposed in \cite{galichon_cupids_2021} to estimate $U_{sht}$ and $V_{sht}$.\footnote{Note that our estimation target is $U_{sht}$ and $V_{sht}$. The difference from the case of \cite{galichon_cupids_2021} is that we just need one side fixed effect.}

By formulating the aggregate matching outcome $\mu_{sht}$ as a realization of a Poisson distribution, it is possible to estimate the aggregate-level utilities by a Poisson regression with fixed effects.
For a regressor in the Poisson regressions, we make a set of polynomials for some degree from $X^{U}_{sht}$ and $X^{V}_{sht}$, which is denoted by $X^{\mathrm{poly}}_{sht}$. We model the aggregate-level utilities as follows:
\begin{align}
\label{eq:systematic_polynomial_approx}
    U_{sht} = X_{sht}^{\mathrm{poly}\prime } \beta^{\mathrm{poly}}_{U},\ V_{sht} = X_{sht}^{\mathrm{poly} \prime} \beta^{\mathrm{poly}}_{V}.
\end{align}
We use $\hat{\beta}^{\mathrm{poly}}_{U}$ and $\hat{\beta}^{\mathrm{poly}}_{V}$ as the estimated coefficients attached with the polynomials.
And we define the estimated aggregate-level utilities by $\hat{U}_{sht} \equiv X_{sht}^{\mathrm{poly}\prime } \hat{\beta}^{\mathrm{poly}}_{U}$ and $\hat{V}_{sht} \equiv X_{sht}^{\mathrm{poly}\prime } \hat{\beta}^{\mathrm{poly}}_{V}$.
In Appendix~\ref{app:poisson_derivation}, we describe the relationship between the moment matching estimator and the two-way fixed effect Poisson regression, which complements the discussion in \cite{Galichon_poisson}.

\subsubsection{Second Step}
\label{sec:strategy_second}
When we directly observe the values of $\iota_{ht}$ for all hospitals and periods, we can use \eqref{eq:identity} to construct an estimator of $\theta$. By inserting the estimation results in the first step, we construct the following moment conditions for $\theta$:
\begin{equation}
\label{eq:moment_conditions}
    \begin{aligned}
    \sum_{s} \omega_{sht} \left(X_{sht}^{U,\mathrm{base}\prime}\beta_U  \right) = \sum_{s}\omega_{sht} \hat{U}_{sht} - \iota_{ht},\  \forall \ h,t \\
    \sum_{s} \omega_{sht} \left(X_{sht}^{V,\mathrm{base}\prime}\beta_V \right)= \sum_{s}\omega_{sht} \hat{V}_{sht} + \iota_{ht},\  \forall \ h,t.
    \end{aligned}
\end{equation}
In Appendix \ref{sec:MC}, we show a Monte Carlo exercise adopting this approach to show how to recover the structural parameters.


In practice, we face a measurement problem: we cannot observe the aggregate-level transfer $\iota_{ht}$. Instead, we can only observe the realized salaries paid by hospitals every period. It is important to note that salary represents only one component of the total transfer in this market, which also includes non-monetary aspects. For example, the hospital accepts the risk of medical incidents by allowing the less-experienced medical interns to get more practice on the job. The workload in a hospital also comprises such unobserved transfers. Furthermore, we expect that the observed salaries correlate with these unobservable terms, which makes the identification more demanding.

For this problem, we introduce a measurement model to connect the observed salaries to $\iota_{ht}$.
Denoting the salary paid in a hospital $h$ at time $t$  by $S_{ht}$, we assume that both schools and hospitals have quasi-linear utilities with respect to monetary transfers:
\begin{equation}
    \begin{aligned}
        \iota_{ht} = \gamma_{0,U} + \gamma_{1,U} S_{ht} + \psi^{U}_{ht}\\
        -\iota_{ht} = \gamma_{0,V} + \gamma_{1,V} S_{ht} + \psi^{V}_{ht}.
    \end{aligned}
    \label{eq:quasi_linear}
\end{equation}
$\gamma_{1,V}$ is expected to be negative because the salary is the amount of money paid to the doctor by the hospital. $\psi^{U}_{ht}$ is the unobserved transfer from the hospital to the matched doctors, and $\psi^{V}_{ht}$ is the same unobserved transfer from the doctor to the hospital. 

The unobserved transfer is likely correlated with the observed monetary transfer. Hence, we need some instrumental variables that have an influence only on the salary.
As such instrumental variables, we use the characteristics of the surrounding hospitals as in \cite{BLP}.
The rationale behind these instruments is that a hospital considers the characteristics of other hospitals when setting its salary, whereas the unobserved transfer is not known to others. 
In practice, we use only the characteristics of nearby hospitals located within a 20 km radius of a given hospital as instrumental variables for the salary, even though our model accounts for all hospitals operating within the same market.

By combining the moment conditions and the measurement model, estimating equations are specified as follows:
\begin{equation}
    \begin{aligned}
    \label{eq:estimating_equation}
    &\sum_{s}\omega_{sht} \hat{U}_{sht} = \gamma_{0,U} + \gamma_{1, U}S_{ht} + \sum_{s} \omega_{sht} \left(X_{sht}^{U\prime}\beta_U \right) + \psi_{ht}^{U}\\
    &\sum_{s} \omega_{sht}\hat{V}_{sht} = \gamma_{0,V} + \gamma_{1,V} S_{ht} + \sum_{s} \omega_{sht} \left(X_{sht}^{V\prime}\beta_V \right) + \psi_{ht}^{V}.
    \end{aligned}
\end{equation}
When we take the weighted average of every variable in $X_{sht}^{U}$ and $X_{sht}^{V}$ as independent variables in the right-hand side, the above equations are just linear equations in $\theta$. 
We estimate these linear equations using the instrumental variables discussed above.

\section{Empirical Results}
\label{sec:empirical_results}
In this section, we show the estimation results.
In Section \ref{sec:results_first}, we show the estimation results of our first step.
In Section \ref{sec:results_second}, we show the estimation results of our second step: the preference parameters of both sides of the market.

\subsection{First Step}
\label{sec:results_first}
From 2017 to 2019, we estimate the aggregate-level utility of both sides separately. The degree of the polynomial approximating the aggregate-level utility is our tuning parameter. In this section, we show the results obtained when we choose three as the degree of polynomials, as this choice allows a more flexible functional form. In Appendix \ref{sec:appendix_empirical_results}, we show the results obtained when the degree of the polynomials is set to two.

\begin{figure}[tbp]
\begin{center}
    \centering
    \begin{subfigure}{\textwidth}
        \centering
        \includegraphics[width=\textwidth]{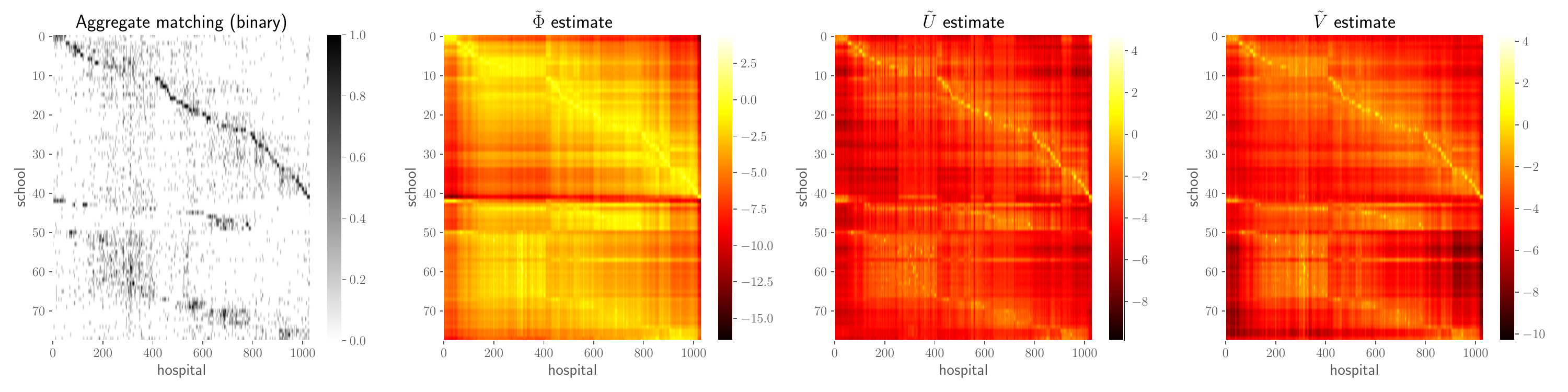}
        \caption{2017}
        \label{fig:label-2017}
    \end{subfigure}
    \vfill
    \begin{subfigure}{\textwidth}
        \centering
        \includegraphics[width=\textwidth]{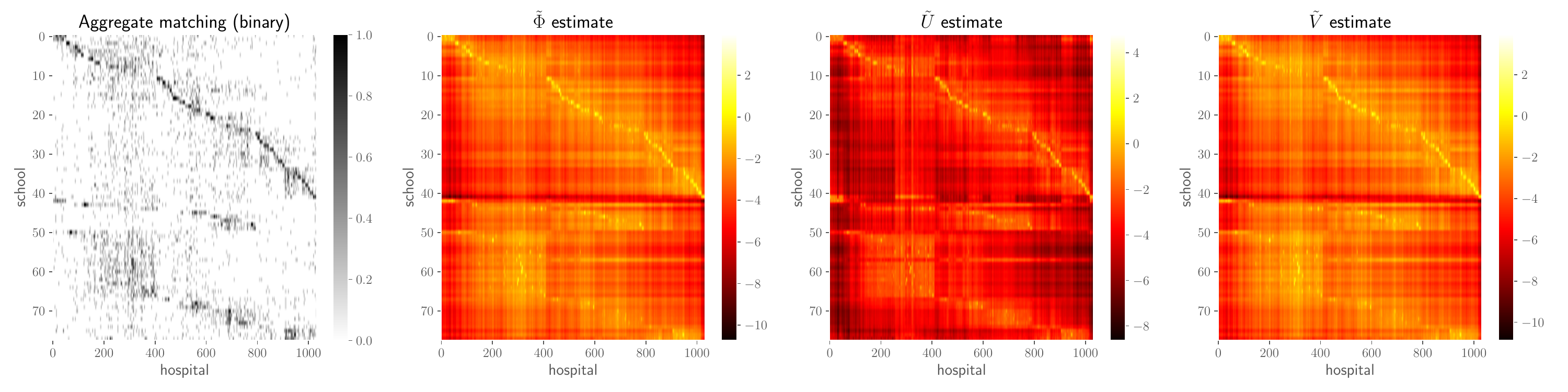}
        \caption{2018}
        \label{fig:label-2018}
    \end{subfigure}
    \vfill
    \begin{subfigure}
{\textwidth}
        \centering
        \includegraphics[width=\textwidth]{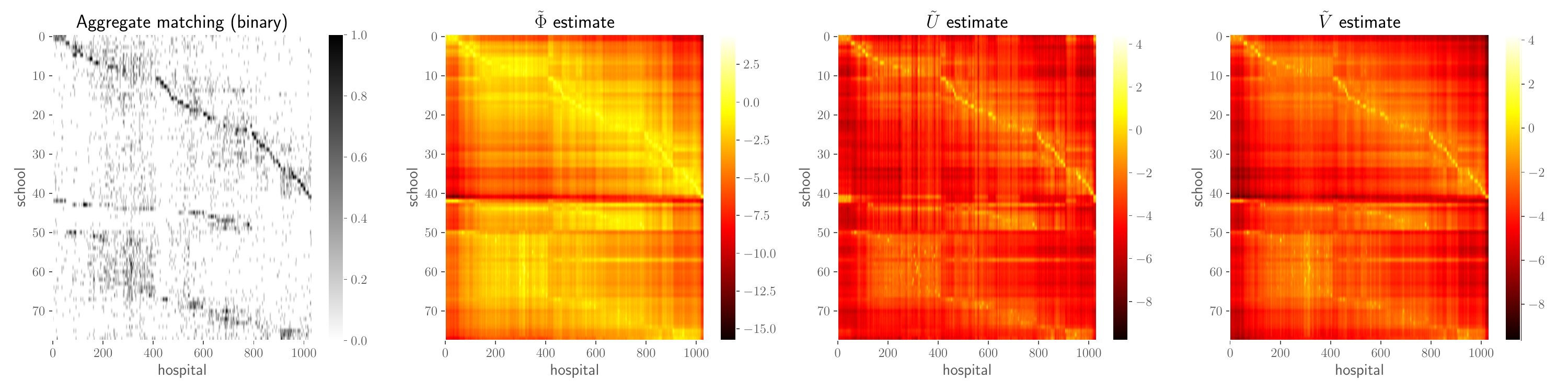}
        \caption{2019}
        \label{fig:label-2019}
    \end{subfigure}
    \caption{Aggregate Matchings, Aggregate-level Utilities and Social Surpluses.}
    \label{fig:first_stage_all_3}
    \vskip 10pt
    \begin{footnotesize}
    \end{footnotesize} 
\end{center}
\end{figure}

Each panel in Figure \ref{fig:first_stage_all_3} shows the observed matching pattern (leftmost figure), the estimated aggregate-level social surpluses (second from the left), the estimated aggregate-level utilities from the doctors’ side (second from the right), and the estimated aggregate-level utilities from the hospitals’ side (rightmost figure). All figures are heatmaps, where the vertical axis represents the indices of medical schools and the horizontal axis represents the indices of hospitals.\footnote{The way to set the index is described in Section \ref{sec:data}.} In the right three heatmaps, brighter colors indicate higher values.

The visible pattern in the aggregate matching is well captured by our estimation. Specifically, all the heatmaps reflect the likelihood of matches between graduates from local public universities and nearby hospitals. As illustrated in the rightmost figures, even from the hospitals' perspective, graduates from closer medical schools provide higher aggregate-level utility. \footnote{In contrast, \cite{agarwal_empirical_2015} highlights that distance influences doctors' decisions but does not account for preference heterogeneity concerning distance on the hospital side.} This pattern can be interpreted as evidence of the importance of local knowledge: knowledge of the local medical environment is so critical that hospitals prefer to hire local doctors.

\subsection{Second Step}
\label{sec:results_second}
Based on the estimation results from the first step, we calculate the marginal effects of covariates on the base utilities. We include the following hospital-specific variables as the characteristics in the preference of doctor side: \textit{logarithm of number of beds of a hospital}, which acts as the measure of the size and the quality of the hospital, \textit{dummy variable of university hospital}, \textit{dummy variable of governmental hospital}, \textit{dummy variable of urban area}, and \textit{dummy variable of Tokyo}.\footnote{Tokyo is by far the largest metropolitan area compared to other urban regions and holds a unique status as the capital, which is why we included a dedicated dummy variable for it.} Furthermore, we include the following pairwise variables: \textit{logarithm of distance}, \textit{logarithm of number of previous matches}, and \textit{dummy variable of affiliation relationship}. As the characteristics in the preference of hospital side, in addition to the pairwise variables, we include the following university-specific variables: \textit{dummy variable of public university}, \textit{T-score of the entrance exam}, \textit{dummy variable of urban areas}, and \textit{dummy variable of Tokyo}.

\begin{table}[tbp]\centering
\def\sym#1{\ifmmode^{#1}\else\(^{#1}\)\fi}
\captionsetup{justification=centering}
\caption{Preference Parameters, Degree of Polynomials = 3}
\label{tab:results_degree3}
\scalebox{0.9}{
\begin{tabular}{l*{4}{c}}
\hline\hline
            &\multicolumn{1}{c}{(1)}   &\multicolumn{1}{c}{(2)}   &\multicolumn{1}{c}{(3)}   &\multicolumn{1}{c}{(4)}   \\
            &  University   &University (IV)   &    Hospital   &Hospital (IV)   \\[1em]
\hline
Constant    &      -5.494***&      -6.724***&       1.194   &       1.954** \\
            &     (0.187)   &     (0.327)   &     (0.776)   &     (0.874)   \\
[1em]
Salary (million Yen)&       0.574***&       2.527***&       0.634***&      -1.780** \\
            &     (0.128)   &     (0.479)   &     (0.146)   &     (0.792)   \\
[1em]
Tokyo       &      0.0371   &       0.112** &      0.0251   &     -0.0940   \\
            &    (0.0376)   &    (0.0461)   &    (0.0618)   &    (0.0712)   \\
[1em]
urban       &     -0.0307   &      0.0572*  &       0.205***&       0.125***\\
            &    (0.0264)   &    (0.0333)   &    (0.0346)   &    (0.0447)   \\
[1em]
log(Distance)&      -0.438***&      -0.436***&      -0.409***&      -0.373***\\
            &    (0.0150)   &    (0.0121)   &    (0.0194)   &    (0.0217)   \\
[1em]
log(Previous Match)&       1.245***&       1.229***&       1.551***&       1.560***\\
            &    (0.0305)   &    (0.0231)   &    (0.0411)   &    (0.0463)   \\
[1em]
log(Beds)   &       0.744***&       0.814***&               &               \\
            &    (0.0303)   &    (0.0353)   &               &               \\
[1em]
Public university&               &               &       0.287***&       0.289***\\
            &               &               &    (0.0569)   &    (0.0610)   \\
[1em]
T-score    &               &               &      -1.660** &      -2.774***\\
            &               &               &     (0.734)   &     (0.810)   \\[1em]
\hline
\(N\)       &        2847   &        2627   &        2847   &        2627   \\
\hline\hline
\multicolumn{5}{l}{\footnotesize Standard errors in parentheses. * p<0.1, ** p<0.05, *** p<0.01}\\
\end{tabular}}
\end{table}

The estimation results are shown in Table \ref{tab:results_degree3}. Columns 1 and 3 correspond to the case of OLS. Columns 2 and 4 are the results obtained when we use BLP instruments for salary.
The direction of the estimated coefficients of salary is aligned with the expected signs when we use instruments. We adopt the results obtained using IV estimations as our main estimation results.
Although we cannot reject the null hypothesis that \(\left|\gamma_{1,U}\right| = \left|\gamma_{1,V}\right|\) at the $5\%$ significance level, we use different coefficients in the subsequent counterfactual analysis rather than assuming these two are equal.\footnote{The $p$-value of the test is \(0.22\).}

\begin{table}[tbp]\centering
\def\sym#1{\ifmmode^{#1}\else\(^{#1}\)\fi}
\captionsetup{justification=centering}
\caption{University Preference Parameters (Unit: Million Yen) \\ Degree of Polynomials = 3}
\label{tab:monetary_u_3}
\scalebox{0.9}{
\begin{tabular}{l*{3}{c}}
\hline\hline
            &\multicolumn{1}{c}{(1)}&\multicolumn{1}{c}{(2)}&\multicolumn{1}{c}{(3)}\\
Coefficient of Salary =           &\multicolumn{1}{c}{$2.527$} &\multicolumn{1}{c}{$2.412$} &\multicolumn{1}{c}{$2.519$}\\[1em]
\hline
log(Distance)&      -0.173\sym{***}&      -0.181\sym{***}&      -0.173\sym{***}\\
            &      (0.03)         &      (0.03)         &      (0.03)         \\[1em]
log(Previous Match)&       0.486\sym{***}&       0.508\sym{***}&       0.487\sym{***}\\
            &      (0.09)         &      (0.10)         &      (0.09)         \\[1em]
Affiliation &       0.182\sym{***}&       0.186\sym{**} &       0.182\sym{**} \\
            &      (0.06)         &      (0.06)         &      (0.06)         \\[1em]
University Hospital&      -0.005         &      -0.001         &      -0.005         \\
            &      (0.03)         &      (0.03)         &      (0.03)         \\[1em]
Governmental Hospital&       0.005         &       0.005         &       0.005         \\
            &      (0.01)         &      (0.01)         &      (0.01)         \\[1em]
log(Beds)   &       0.322\sym{***}&       0.338\sym{***}&       0.324\sym{***}\\
            &      (0.06)         &      (0.06)         &      (0.06)         \\[1em]
\hline
\(N\)       &        2627         &        2627         &        2627         \\
Urban $\times$ Year &  & $\surd$ & $\surd$  \\
Tokyo $\times$ Year &  &  & $\surd$  \\
\hline\hline
\multicolumn{4}{l}{\footnotesize Standard errors in parentheses. * p<0.1, ** p<0.05, *** p<0.01}\\
\end{tabular}}
\vskip 10pt
\end{table}

The distance between the university and the hospital negatively influences both the preferences of doctors and hospitals, which is aligned with the estimation results of the first step.
Furthermore, as intuitively and anecdotally validated, the previous number of matches has a strong influence on the preferences. The more previous matches lower the hurdle to apply for the doctor side, and the uncertainty about quality is cleared from the hospital's perspective.
The quality measure for both sides is also impactful. From the doctor's side, the number of beds in a hospital increases the utility obtained when matching with a hospital. The hospital prefers doctors from a public university, which is more competitive to enter.
We find positive coefficients for the Tokyo dummy and the urban dummy on the doctor's side, and a positive coefficient for the urban dummy on the hospital's side.
\footnote{If there is an implicit tax on the prefectures in the urban areas, this coefficient might be underestimated. In Appendix \ref{sec:appendix_AE_EAE}, we exploit the tightening regional caps to check this existence and conclude that the current market does not face such implicit taxation.}


\begin{table}[tbp]\centering
\def\sym#1{\ifmmode^{#1}\else\(^{#1}\)\fi}
\captionsetup{justification=centering}
\caption{Hospital Preference Parameters (Unit: Million Yen)\\
Degree of Polynomials = 3}
\label{tab:monetary_v_3}
\scalebox{0.9}{
\begin{tabular}{l*{3}{c}}
\hline\hline
            &\multicolumn{1}{c}{(1)}&\multicolumn{1}{c}{(2)}&\multicolumn{1}{c}{(3)}\\
Coefficient of Salary =           &\multicolumn{1}{c}{$1.780$} &\multicolumn{1}{c}{$1.579$} &\multicolumn{1}{c}{$1.810$}\\[1em]
\hline
log(Distance)&      -0.209\sym{*}  &      -0.236\sym{*}  &      -0.206\sym{*}  \\
            &      (0.10)         &      (0.11)         &      (0.10)         \\[1em]
log(Previous Match)&       0.877\sym{*}  &       0.989\sym{*}  &       0.862\sym{*}  \\
            &      (0.39)         &      (0.44)         &      (0.38)         \\[1em]
Affiliation &      -1.235\sym{*}  &      -1.383\sym{*}  &      -1.215\sym{*}  \\
            &      (0.52)         &      (0.59)         &      (0.50)         \\[1em]
Public University&       0.163\sym{*}  &       0.188\sym{*}  &       0.160\sym{*}  \\
            &      (0.08)         &      (0.09)         &      (0.07)         \\[1em]
Prestige    &      -1.558\sym{*}  &      -1.764\sym{*}  &      -1.534\sym{*}  \\
            &      (0.66)         &      (0.74)         &      (0.64)         \\[1em]
\hline
\(N\)       &        2627         &        2627         &        2627         \\
Urban $\times$ Year &  & $\surd$ & $\surd$  \\
Tokyo $\times$ Year &  &  & $\surd$  \\
\hline\hline
\multicolumn{4}{l}{\footnotesize Standard errors in parentheses. * p<0.1, ** p<0.05, *** p<0.01}\\
\end{tabular}}
\end{table}

Next, we evaluate the marginal effects in monetary units. For this purpose, we calculate the fractions of the estimated coefficients of the covariates relative to the coefficients of salary.
The fractions and the standard errors for the doctor side are shown in Table \ref{tab:monetary_u_3}, and the same ones for the hospital side are shown in Table \ref{tab:monetary_v_3}.
In both tables, we present results for three specifications that differ in whether they include interaction terms between year dummies and the urban and Tokyo dummy variables. The results remain qualitatively robust across these specifications.

For the doctor side, the estimates tell that the match with a hospital that is $10\%$ far away decreases the utility by from $0.017$ to $0.018$ million yen: this is about $\$108$. The number of previous matches, the affiliation relationship and the number of beds of a hospital play the positive influences: $10\%$ increase in the number of previous matches improves the utility by from $0.049$ to $0.051$ million yen, which is about $\$319$, the hiring by affiliated hospitals increases the utility by from $0.182$ to $0.184$ million yen, which is about $\$1,169$, and $10\%$ increase in the number of beds improves the utility by from $0.032$ to $0.034$ million yen, which is about $\$210$.

For hospital side, distance and the previous number of matches play the similar roles: the doctors from $10\%$ faraway university decrease the utility of hospital by from $0.021$ to $0.024$ million yen, which is about $\$146$, and $10\%$ increase in the number of previous matches improves the utility of hospital by from $0.086$ to $0.099$ million yes, which is about from $\$549$ to $\$632$. 
The indicator of public university also has a positive impact as expected: the premium of graduating from a public university is from $0.160$ to $0.188$ million yen, which is about from $\$1,022$ to $\$1,201$.

These marginal effects of the covariates are significant even compared to the aggregate-level utilities. In Appendix \ref{sec:appendix_relative_impacts}, we calculate the ratio of the estimated marginal effects of the covariates to the aggregate-level utility. For many covariates, a $10\%$ change in the covariates corresponds to approximately $1$--$5\%$ of the aggregate-level utility on both sides.
\section{Counterfactual Simulations}
\label{sec:cf_simulations}

\subsection{Decomposing the Welfare Loss of Cap-based Policy}
\label{sec:cf_simulation_agg}

We use our empirical estimates to conduct counterfactual simulations designed to evaluate the current regulatory framework against two alternative policies. This analysis enables us to quantify the welfare losses associated with the existing cap-based system and to evaluate the potential benefits of a more efficient, subsidy-based approach. By comparing these scenarios, we can decompose the sources of inefficiency and distinguish the welfare cost of the distributional goal itself from the cost of the policy instrument used to achieve it.

The first policy, \emph{Artificial Caps} (\textsf{AC}), simulates the status quo implemented by the Japan Residency Matching Program (JRMP). We compute the aggregate equilibrium under the set of residency positions that were actually available in the JRMP, subject to regional caps.

The second, \emph{No Caps} (\textsf{NC}), represents an unregulated welfare benchmark. This counterfactual assumes all residency positions are made available at each hospital's true capacity. 
To simulate this policy, we must infer the true capacity of each hospital, as this is unobservable under the current regulated regime. We define a hospital's true capacity as the maximum number of positions it reported to the JRMP between 2015 and 2023.\footnote{This definition provides a conservative estimate of each hospital's true capacity. As a result, our calculation of the total welfare loss from the cap system should be interpreted as a lower bound.} The welfare difference between the \textsf{NC} and \textsf{AC} scenarios measures the total efficiency loss created by the regional cap system.

The third policy, \emph{Optimal Subsidy} (\textsf{OS}), provides a constrained-optimal benchmark. In this scenario, we remove all artificial caps and instead introduce subsidies designed to maximize social surplus subject to regional constraints. These subsidies are set to achieve a specific distributional goal: ensuring that designated rural prefectures receive at least as many residents as they did under the AC policy.\footnote{By setting the floor for the number of residents equal to the outcome of the current policy, we adopt a conservative distributional objective. A more ambitious policy goal with higher floors would necessarily require a larger total subsidy. Thus, our estimate of the fiscal cost should be interpreted as a lower bound.} The optimal set of subsidies is then derived by solving the social surplus maximization problem \optprobref{opt:natural_agg_eq_primal} subject to these regional floor constraints.\footnote{We define rural prefectures as the 15 prefectures with the lowest physician match rates between 2017 and 2019, calculated as the total matches over the period divided by the sum of reported capacities.}


By construction, the three policies have a clear welfare ranking: the unconstrained \textsf{NC} policy is first-best, the \textsf{OS} policy is second-best, and the \textsf{AC} policy yields the lowest surplus.
This framework allows for a clear decomposition of the total inefficiency.
The comparison between the \textsf{OS} and \textsf{AC} outcomes isolates the welfare loss attributable to the choice of regulatory instrument (i.e., caps instead of subsidies). The remaining welfare difference, observed between the \textsf{NC} and \textsf{OS} outcomes, quantifies the inherent cost of the distributional goal itself.

\begin{table}[tbp]
\centering
\small
\caption{Comparison of Three Policies in 2017}
\label{tab:welfare_AE_2017}
\begin{threeparttable}
\begin{tabularx}{\textwidth}{l @{} RRR}
\toprule 
Policy & \multicolumn{1}{c}{\textsf{AC (Artificial Caps)}} & \multicolumn{1}{c}{\textsf{NC (No Caps)}} & \multicolumn{1}{c}{\textsf{OS (Optimal Subsidy)}} \\

\midrule
Artificial caps & Yes & No & No \\
Floor constraints & Yes & No & Yes \\
Subsidies & No & No & Yes \\
\midrule
Match rate & 0.868 & 0.912 & 0.912 \\
Doctors' welfare & 32795.8 & 33441.7 & 33444.5 \\
Hospitals' welfare & 29586.6 & 31578.1 & 31579.4 \\
Government's revenue & 0.0 & 0.0 & $[-10.5, -7.4]$ \\
Social welfare & 62382.3 & 65019.8 & $[65013.3, 65016.5]$ \\
\#(subsidized regions) & 0 & 0 & 3 \\
Average subsidy & 0.000 & 0.000 & -0.040 \\
\#(constraint violations) & 0 & 3 & 0 \\
\bottomrule
\end{tabularx}
\begin{tablenotes}
\footnotesize
\item[*] All welfare and revenue figures are expressed in units of 1 million JPY per month. Government revenue is positive when taxes are imposed on doctors and hospitals and negative when subsidies are provided. Doctors' and hospitals' welfare are scaled according to specification (1) in \Cref{tab:monetary_u_3} and \Cref{tab:monetary_v_3}. We present the bounds of the government's net revenue, scaled by the coefficients on the doctor side and the hospital side, respectively. The social welfare is the sum of doctors' welfare, hospitals' welfare, and the government's revenue. \#(constraint violations) counts the number of prefectures violating the lower bounds (among the 15 rural regions).
\end{tablenotes}
\end{threeparttable}
\end{table}

Table~\ref{tab:welfare_AE_2017} presents the simulation results for 2017. The findings for other years are qualitatively similar and are available in Appendix~\ref{sec:appendix_cf_other_years}. Unless otherwise specified, monetary values are in millions of JPY per month, calculated using the salary coefficients from specification (1) in Tables~\ref{tab:monetary_u_3} and~\ref{tab:monetary_v_3}. For the \textsf{OS} policy, we report bounds for government revenue and total welfare, reflecting different possible incidences of taxes and subsidies between doctors and hospitals.

Comparing the \textsf{AC} and \textsf{NC} scenarios reveals the direct impact of the caps. Removing them increases social surplus by approximately 2.6 billion JPY and raises the match rate by 4.4 percentage points. This welfare gain, however, comes at the expense of three rural prefectures receiving fewer residents, confirming that the caps are binding and effective at redirecting physicians, albeit inefficiently.

The \textsf{OS} policy demonstrates that these distributional goals can be met with minimal welfare loss. The surplus under \textsf{OS} is substantially higher than under \textsf{AC} and is nearly identical to the unconstrained \textsf{NC} benchmark. This finding suggests that the vast majority of the welfare loss from the current policy originates from the choice of regulatory instrument---the caps themselves---rather than from the underlying distributional constraint.

Figure~\ref{fig:welfare_diff} illustrates the prefectural-level welfare differences between the \textsf{OS} and \textsf{AC} policies. The gains from adopting subsidies are positive for all prefectures and are largest in urban areas that face the tightest caps under the JRMP. This reinforces the conclusion that the inefficiency originates primarily from restricting matches in high-demand urban centers.

\begin{figure}[tb]
    \centering
    \includegraphics[width=0.8\linewidth]{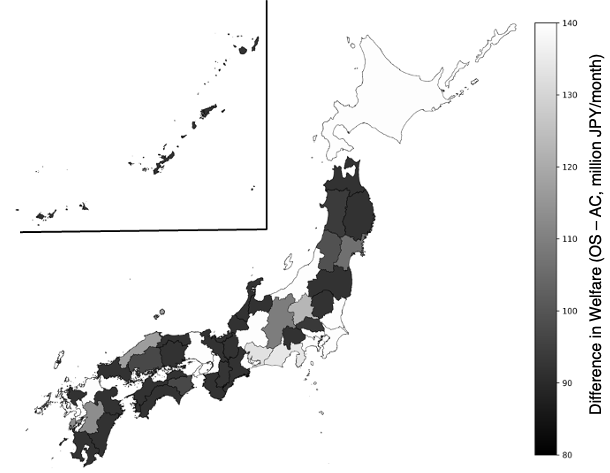}
    \caption{Difference in Welfare between \textsf{OS} and \textsf{AC} by Prefecture}
    \label{fig:welfare_diff}
\end{figure}

Finally, the \textsf{OS} policy appears fiscally practical. The required subsidy to achieve the distributional floor is approximately \$400 per month for each matched pair in the targeted rural prefectures, which represents 10–20\% of a typical resident's salary. The total national cost of this subsidy program would not exceed \$100,000 per month, suggesting that a shift to a more efficient subsidy-based policy is feasible.

\paragraph{Changes in matching patterns}
\begin{table}[t]
\centering
\small
\begin{threeparttable}
\caption{Coefficient Size Tests for Hospital Side}
\label{tab:coef_test_hospital_panels_pair}
\setlength{\tabcolsep}{6pt}

\begin{tabular}{lcccc}
\toprule
\textbf{Panel A: \textsf{NC} vs \textsf{AC}} & Public & Difficulty & Distance & Unmatch~Rate \\
\midrule
Urban       & $<^{***}$ & $<^{***}$ & $>^{***}$ & $<^{***}$ \\
\midrule\midrule
\textbf{Panel B: \textsf{OS} vs \textsf{NC}} & Public & Difficulty & Distance & Unmatch~Rate \\
\midrule
Urban       & $>^{***}$ & $>^{**}$ & $<^{***}$ & $>^{***}$ \\
\bottomrule
\end{tabular}

\begin{tablenotes}
\footnotesize
\item[*] Each cell compares coefficients across the two regimes indicated in the panel header.
``$>$'' (resp.\ ``$<$'') means the coefficient is larger (resp.\ smaller) under the first regime.
Significance: $^{***}p<0.01$, $^{**}p<0.05$, $^{*}p<0.10$ (two‐sided test).
\end{tablenotes}
\end{threeparttable}
\end{table}


\begin{table}[t]
\centering
\small
\begin{threeparttable}
\caption{Coefficient Size Tests for School Side}
\label{tab:coef_test_school_panels_pair}
\setlength{\tabcolsep}{6pt}

\begin{tabular}{lccc}
\toprule
\textbf{Panel A: \textsf{NC} vs \textsf{AC}} & Urban~Hosp. & Distance & Unmatch~Rate \\
\midrule
Public~University & $<$ & $<$ & $<$ \\
Difficulty        & $<$ & $<$ & $<$ \\
\midrule\midrule
\textbf{Panel B: \textsf{OS} vs \textsf{NC}} & Urban~Hosp. & Distance & Unmatch~Rate \\
\midrule
Public~University & $>$ & $<$ & $<$ \\
Difficulty        & $>^{*}$ & $<$ & $>$ \\
\bottomrule
\end{tabular}

\begin{tablenotes}
\footnotesize
\item[*] Each cell compares coefficients across the two regimes indicated in the panel header. ``$>$'' (resp.\ ``$<$'') means the coefficient is larger (resp.\ smaller) under the first regime. Significance: $^{***}p<0.01$, $^{**}p<0.05$, $^{*}p<0.10$ (two‐sided test).
\end{tablenotes}
\end{threeparttable}
\end{table}



We compare the equilibrium matching patterns between \textsf{AC}, \textsf{NC}, and \textsf{OS}. Specifically, we replicate the regression analysis from Section~\ref{sec:empirical_pattern} for every scenario and statistically test whether the coefficients on each covariate differ significantly. 
For these tests, we compute standard deviations of the estimates using cluster-robust standard errors at the unit level: the hospital level for the hospital perspective and the school level for the school perspective.
The results from the hospital perspective are presented in Table \ref{tab:coef_test_hospital_panels_pair}, and those from the hospital perspective are shown in Table \ref{tab:coef_test_school_panels_pair}. In both tables, Panel A compares \textsf{NC}, where regional caps in urban areas are removed, with \textsf{AC}, which represents the current market equilibrium. Panel B compares \textsf{OS}, which allows subsidies to meet floor constraints in rural areas, with \textsf{NC}.

We begin by examining Panel A, which captures the changes in matching patterns that occur when the regional caps in urban areas are removed. According to Table~\ref{tab:coef_test_hospital_panels_pair}, the newly available slots in urban hospitals are filled primarily by graduates from private universities, particularly those with lower entrance difficulty and located in more distant regions. While the unmatch rate in urban areas is already low, it decreases even further under the \textsf{NC} policy. From the school perspective shown in Table~\ref{tab:coef_test_school_panels_pair}, although the differences are not statistically significant, the first row displays a pattern consistent with this observation. The negative correlation between distance and both public status and difficulty also suggests that public universities, especially those located in rural areas, have become even more prominent sources of residents for rural hospitals under this scenario.

Next, we turn to Panel B, which illustrates changes in matching patterns resulting from the introduction of subsidies in rural areas. As shown in Table~\ref{tab:coef_test_hospital_panels_pair}, the changes in matching patterns from \textsf{NC} to \textsf{OS} (Panel B) are the exact opposite of those from \textsf{AC} to \textsf{NC} (Panel A) across all variables. This suggests that the subsidy effectively counteracts the influx of students into urban areas resulting from the removal of regional caps, encouraging them to remain in rural areas instead. The results from the school perspective in Table~\ref{tab:coef_test_school_panels_pair} are also consistent with this interpretation.

These changes in matching patterns help explain why the optimal subsidy required in our analysis is relatively small. The subsidy's primary role is to retain students from rural universities who would otherwise move to urban hospitals after the removal of regional caps. Because distance is a significant factor in student preferences, matching with a local hospital is not a prohibitively costly alternative for these students. Therefore, a modest subsidy is sufficient to persuade enough students to choose a local hospital, thereby meeting the distributional goal at a low fiscal cost.

\subsection{Inherent Inefficiency of Caps: An Evaluation of the Flexible DA Algorithm}
\label{sec:cf_simulation_indiv}


To further assess the limitations of cap-based regulations, we conduct a second set of counterfactual simulations evaluating a state-of-the-art matching algorithm from the literature. The \emph{Flexible Deferred Acceptance (FDA) algorithm}, a variant of the standard DA algorithm, adjusts hospital capacities to meet regional constraints dynamically, rather than relying on the pre-determined, fixed capacity reductions imposed by the current JRMP protocol. To distinguish whether the inefficiency of the current policy stems from a naive implementation or an inherent limitation of caps, this exercise compares the outcomes of the FDA algorithm with the three policies considered previously: Artificial Caps (\textsf{AC}), No Caps (\textsf{NC}), and Optimal Subsidy (\textsf{OS}).

This simulation exercise more closely mimics the institutional details of the JRMP. We generate individual-level preferences from which agents construct preference rankings, and then use a centralized algorithm to determine the final matches. Unlike the previous analysis, this approach computes individual-level outcomes without relying on the large-market approximation, and it determines matches using the DA algorithm rather than as a transferable-utility stable outcome. This allows us to check the robustness of our baseline findings from Section~\ref{sec:cf_simulation_agg}.

For the \textsf{AC} policy, we first estimate the aggregate-level utilities by solving the optimization problem \optprobref{opt:natural_agg_eq_primal} under the actual JRMP caps. We then generate individual doctor and hospital utilities by adding a random logit error to these aggregate utilities. Based on the resulting individual preferences, we construct preference lists and run the standard DA algorithm to determine the final matching under the JRMP caps. The \textsf{NC} policy simulation is analogous, except that we use each hospital's true capacity to compute the aggregate-level utilities.

For the \textsf{OS} policy, we set floor constraints requiring each designated rural prefecture to receive at least as many residents as it did under the AC policy. We then solve the optimization problem \optprobref{opt:natural_agg_eq_primal} to compute the optimal subsidy and the corresponding aggregate-level utilities. Individual preferences are generated from these utilities, and the standard DA algorithm determines the match outcome. 

Finally, for the \textsf{FDA} policy, individual utilities are identical to those in the \textsf{AC} scenario.\footnote{Since hospital capacities are computed during the execution of the FDA algorithm, no definitive method exists for calculating the aggregate-level utilities prior to running the algorithm. We therefore use the utilities from the \textsf{AC} scenario.} The matching is then determined by the FDA algorithm, which operates under regional caps calibrated to satisfy the floor constraints as in the \textsf{OS} scenario. The floor constraints requires each rural prefecture to receive at least as many residents as the number of matches under DA in the \textsf{AC} scenario.
\footnote{Specifically, a regional cap is initially set for each non-rural prefecture equal to the sum of the true capacities of the hospitals within it. We then find the smallest $\alpha$ such that if the regional caps of all prefectures are reduced uniformly by $\alpha$\%, the final matching produced by the FDA algorithm satisfies the floor constraints.}

\begin{table}[tbp]
\centering
\small
\caption{Welfare Comparison of the Simulated Matchings in 2017}
\label{tab:welfare_indiv}
\begin{threeparttable}
\resizebox{\linewidth}{!}{%
\begin{tabular*}{\textwidth}{l @{\extracolsep{\fill}} rrrr}
\toprule
Scenario & \textsf{AC} & \textsf{NC} & \textsf{OS} & \textsf{FDA} \\
\midrule
Algorithm & DA & DA & DA & FDA \\
Artificial Caps & Yes & No & No & No \\
Floor Constraints & Yes & No & Yes & Yes \\
Subsidies & No & No & Yes & No \\
\midrule
Match rate & 0.762 & 0.810 & 0.810 & 0.771 \\
Doctors' welfare & 35850.2 & 37176.3 & 37192.4 & 35979.4 \\
Hospitals' welfare & 29748.7 & 29174.1 & 29176.6 & 29643.0 \\
Government's revenue & 0.0 & 0.0 & $[-17.0, -12.0]$ & 0.0 \\ 
Social welfare & 65598.9 & 66350.5 & $[66352.1, 66357.1]$ & 65622.4 \\ \midrule
{\footnotesize \#regions violating lower bounds} & 0 & 4 & 3 & 0 \\
{\footnotesize \#doctors required to meet lower bounds} & 0 & 15 & 12 & 0 \\ \midrule
{\footnotesize \#doctors matched with urban hospitals} & 3103 & 3196 & 3192 & 3016 \\
{\footnotesize \#doctors matched with rural hospitals} & 1071 & 1170 & 1175 & 1209 \\
Urban hospitals' welfare & 6143.2 & 5758.4 & 5753.2 & 5954.1 \\
Rural hospitals' welfare & 2281.8 & 2227.0 & 2243.8 & 2447.8 \\
\bottomrule
\end{tabular*}
}
\begin{tablenotes}
\footnotesize
\item[*] All welfare and revenue figures are expressed in units of 1 million JPY per month. Government revenue is positive when taxes are imposed on doctors and hospitals, and negative when subsidies are provided. Doctors' and hospitals' welfare are scaled according to specification (1) in \Cref{tab:monetary_u_3} and \Cref{tab:monetary_v_3}. We present the bounds of the government's net revenue, scaled by the coefficients on the doctor side and the hospital side, respectively. The social welfare is the sum of doctors' welfare, hospitals' welfare, and the government's revenue. ``\#regions violating lower bounds'' indicates the number of prefectures whose matched doctor count is less than that of \textsf{AC}. ``\#doctors required to meet lower bounds'' indicates how many additional doctors must be matched with rural hospitals so that these prefectures exceed the matched doctor count in \textsf{AC}.
\end{tablenotes}
\end{threeparttable}
\end{table}

Table~\ref{tab:welfare_indiv} reports the simulation results. Under the \textsf{FDA} policy, capacities in non-rural areas are reduced by 30\% to achieve the floor constraints. Despite this substantial reduction, the FDA algorithm improves upon the current \textsf{AC} policy, increasing the total match rate by 0.9 percentage points and yielding higher total welfare. However, the welfare under \textsf{FDA} remains significantly lower than that under the unconstrained \textsf{NC} benchmark. In contrast, the \textsf{OS} policy achieves a total welfare level nearly identical to that of the \textsf{NC} policy. This result strengthens our earlier conclusion that a fundamental inefficiency is inherent in cap-based policies, at least in this application, and that monetary interventions offer a powerful alternative.

Furthermore, the results for the \textsf{AC}, \textsf{NC}, and \textsf{OS} policies are consistent with the findings in Section~\ref{sec:cf_simulation_agg}. Social welfare increases when the artificial caps are removed (\textsf{NC} vs. \textsf{AC}), and the \textsf{OS} policy nearly restores the first-best welfare level. This consistency across different simulation methodologies validates our baseline approach and reinforces the central findings of our analysis.

\section{Conclusion}
\label{sec:discussion_conclusion}

This study develops and implements a framework to evaluate the efficiency of regulations in matching markets with distributional disparities, focusing on the trade-offs between cap-based policies and monetary interventions. We build a transferable utility (TU) matching model that incorporates regional constraints, such as caps and floors, on the number of matches. By embedding this model within the aggregate matching framework of Galichon and Salanié (2021a), we provide a method to identify preference structures and conduct policy simulations using only aggregate-level matching data and an observable measure of transfers, such as salary. 

Applying this framework to the Japan Residency Matching Program (JRMP), our empirical analysis reveals several key findings. First, the estimation results show that preferences are horizontally differentiated on both sides of the market; factors such as geographic distance and the history of previous matches are significant determinants of preferences for both doctors and hospitals. Second, our counterfactual simulations demonstrate that the current cap-based regulations, while redirecting some physicians to underserved areas, generate substantial welfare losses. We find that a small, targeted subsidy policy can achieve the same distributional outcome as the existing caps but at a significantly higher level of social welfare.  This indicates that the inefficiency stems not from the distributional goal itself, but from the choice of regulatory instrument. Finally, we demonstrate that even a more sophisticated, flexible Deferred Acceptance mechanism still yields considerable welfare losses compared to the optimal subsidy, underscoring the inherent limitations of employing quantity restrictions in this context.

The findings offer important policy implications. For markets like the JRMP, where geographic imbalances are a primary concern, monetary interventions appear to be a more efficient and fiscally practical tool than rigid caps. Our framework offers a practical path to designing such interventions, as its reliance on aggregate-level matching data and salaries is less demanding than approaches that require individual-level preference data. More broadly, the approach developed in this paper can be applied to other matching markets where transfers are present but granular data are scarce. Future research could extend this approach to analyze more complex distributional constraints or explore the long-run effects of such policies on physician career paths and settlement patterns.

\bibliographystyle{aea}
\bibliography{references}

\newpage
\appendix

\renewcommand{\thesection}{\Alph{section}} 
\counterwithin{theorem}{section} 
\counterwithin{lemma}{section}
\counterwithin{prop}{section}
\counterwithin{cor}{section}

\section{Omitted Proofs}

\subsection{Supporting Stable Outcomes as an Equilibrium}
\label{app:JRMP_game}

In this section, we show that any stable outcome (Definition~\ref{def:indiv_eqm}) can be supported as an equilibrium of a stylized game that mirrors the JRMP's matching process. For expository simplicity, we assume that taxation policy is inactive, i.e., $w_z \equiv 0$.\footnote{If not, we can consider $\Phi_{ij}$ below as the gross joint surplus for any given $w$.}

The game consists of three periods. In Period 1, each slot simultaneously announces a profile of transfers $t = (t_{ij})_{i,j}$, where $t_j = (t_{ij})_i$ is the transfer profile chosen by slot $j$. In Period 2, observing all the wage offers, doctors and slots submit their preference lists. In Period 3, the doctor-proposing deferred acceptance (DA) algorithm determines the final matching.
The utility of doctor $i$ being matched with a slot $j$ is defined as follows: the \emph{base utility}, which is the utility felt by $i$ net of transfer, is denoted by $U_{ij}$. Given transfer offer $t$, when matched with slot $j$, doctor $i$'s payoff is $U_{ij} + t_{ij}$. Similarly, the base utility of slot $j$ when matched with doctor $i$ is denoted by $V_{ij}$, and the payoff of $j$ when matched with $i$ given $w$ is $V_{ij} - t_{ij}$. Note that $\Phi_{ij} = U_{ij} + V_{ij}$. We refer to this game as the \emph{JRMP game}.

We assume that all agents submit the preference ranking truthfully given $t$ in Period 2.\footnote{Since the DA algorithm is strategy-proof for the proposing side, this means that we assume the truthful report from the other side. This is a common assumption when studying DA algorithms.}
Given this assumption, the analysis of this game boils down to that of one-shot game played by the slots in Period 1.

The following proposition establishes the link between the stable outcome and the equilibrium of this game.

\begin{prop}
\label{prop:JRMP_game}
    For any stable outcome $(d^*,u^*,v^*)$ of the  matching market $(I, J, \Phi)$, there exists a Nash equilibrium of the JRMP game that results in the matching $d^*$ and payoffs $(u^*,v^*)$.
\end{prop}

\begin{proof}
    Consider the following transfer offer profile:
    \begin{align}
        t_{ij}^* \coloneqq \begin{cases}
            u_{i}^* - U_{ij} & (d_{ij}^* = 1) \\
            M_{ij} & (\text{o.w.})
        \end{cases},
    \end{align}
    where $M_{ij}$ is sufficiently low to be unattractive, i.e., $M_{ij} < u_i^* - U_{ij}$.
    
    Suppose that this transfer profile $t^*$ is chosen in Period 1.
    If doctor $i$ is matched with some slot $j$ under the stable outcome (i.e, $d_{ij}^*=1$), then the preference list submitted by doctor $i$ is\footnote{If there is a tie between slot $j$ and the outside option, the tie-breaking rule does not matter because it won't affect the equilibrium payoff. For simplicity, I assume that doctor $i$ always places slot $j$ at the top.}
    \[
    i \colon \ j \succ \emptyset \succ \cdots.
    \]
    Since $d_{ij}^* = 1$, in the slot $j$'s preference list, doctor $i$ is placed above the outside option. Therefore, doctor $i$ is not rejected in the first round of the doctor-proposing DA algorithm.
    If doctor $i$ is unmatched under the stable outcome, he lists the outside option at the top of his preference list, and proposes to no slot in all rounds.
    Since no rejection occurs, the algorithm terminates in the first round, yielding the matching $d^*$. By construction of $t^*$, each agent enjoys the same payoff as in the stable outcome.

    We now verify that no hospital has a profitable unilateral deviation in Period 1. Suppose toward contradiction that there exists slot $j$ that can be strictly better off by matching with $i'$ by setting $(\tilde{t}_{ij})_i$, where $d_{i'j}^* = 0$.
    For this to be a successful and profitable deviation, two conditions must be met:
    \begin{align}
        v_j^* &< V_{i' j} - \tilde{t}_{i'j}, \\
        u_{i'}^* &\leq U_{i'j} + \tilde{t}_{i'j}.
    \end{align}
    This implies that
    \[
    v_{j}^* + u_{i'}^* < \Phi_{i'j},
    \]
    which violates the stability of $(d^*, u^*, v^*)$. A contradiction.
\end{proof}

\subsection{Results for Section~\ref{sec:limit_cap_based}}
\label{app:taxation_possibility}

Below, without loss of generality, we normalize the outside option values to zero, i.e., $\Phi_{i, j_0} = \Phi_{i_0, j} \equiv 0$ for all $i \in I$ and $j \in J$.\footnote{This is possible because the optimal solution to $(\mathrm{P_0})$ remains the same when we redefine $\Phi_{ij}$ as $\Phi_{ij} - \Phi_{i,j_0} - \Phi_{i_0, j}$ and ignore the terms in the objective function that are independent of $d$.} The goal of this subsection is to prove the following result:
\begin{theorem}
\label{thm:IE_optimal_taxation}
    Fix any market $(I, J, Z, z, \Phi)$ and regional constraints $(\uo_z, \bo_z)_z$. There exists taxation policy $w^*$ such that the matching under the stable outcome under $w^*$ is $d^*$.
\end{theorem}

We first define an auxiliary problem in which we relax the integral constraint with respect to $d$ so that $d_{ij}$ can take any nonnegative real value:
\begin{align}
    &(\mathrm{P_0'}) \left[
    \begin{matrix}
    \displaystyle{\max_{d \geq 0}} & \displaystyle{\sum_{i,j} \Phi_{ij} d_{ij}} \\[12pt]
    \mathrm{s.t.} & \sum_j d_{ij} \leq 1 & (i \in I) \\
    & \sum_i d_{ij} \leq 1 & (j \in J) \\
    & \sum_{j \in z}\sum_i d_{ij} \leq \bar o_z & (z \in Z) \\
    & \sum_{j \in z}\sum_i d_{ij} \geq \underline o_z & (z \in Z)
    \end{matrix}
    \right.
\end{align}

For $(\mathrm{P_0'})$ to be useful, it must yield a matching, or an integer solution. This is known to be the case without regional constraints. The following proposition states that it remains true even with regional constraints.
\begin{prop}
\label{prop:IE_integral_solution}
Fix any matching market $(I, J, Z, z, \Phi)$ and regional constraints $(\uo_z, \bo_z)_z$. Assume that $\mathrm{(P_0')}$ has at least one feasible solution. Then, $\mathrm{(P_0')}$ has an integer optimal solution.
\end{prop}

\subsubsection*{Proof of Proposition~\ref{prop:IE_integral_solution}}
Our proof utilizes the known results regarding total unimodularity:
\begin{definition}[Total unimodularity]
Let $A$ be an integer matrix. $A$ is \emph{totally unimodular} if
any minor principal is either -1, 0, or 1.
\end{definition}

First, we show the following lemma:

\begin{lemma}\label{prop:TU_regional}
The matrix that represents the set of constraints of $(\mathrm{P}'')$ is totally unimodular.
\end{lemma}

Our proof of Lemma~\ref{prop:TU_regional} relies on the following fact about the total unimodularity.

\begin{lemma}[\cite{ghouila1962caracterisation}]\label{lem:ghouila-Houri}
An $m \times n$ integer matrix $A$ is totally unimodular iff for each subset of rows $R \subseteq [m]$, there exists a partition $R_1$ and $R_2$ of $R$ such that
\[
\forall k \in [n],\quad \sum_{i \in R_1} a_{ik} - \sum_{i \in R_2} a_{ik} \in \{-1,0,1\}
\]
\end{lemma}
\begin{proof}[Proof of Lemma~\ref{prop:TU_regional}]
First, observe that the feasibility constraints (i.e., all the constraints except $x \geq 0$) of an instance of TU matching with regional constraints can be represented by $(|I|+|J|+2|Z|)\times (|I|\times|J|)$ matrix $A$ (see also Example~\ref{example:tu_regional}) such that
\begin{easylist}[itemize]
@ Each row corresponds to either (1) agent $i \in I$, (2) agent $j \in J$, (3) an upper bound $\bar o_z$ for region $z \in Z$, or (4) a lower bound $\underline o_z$ for region $z \in Z$.
@ Each column corresponds to an $(i,j) \in I \times J$ pair.
@ The component of row $i \in I$ is 1 for column $(i, j')$ for any $j' \in J$; the component is 0 otherwise.
@ The component of row $j \in J$ is 1 for column $(i', j)$ for any $i' \in I$; the component is 0 otherwise.
@ Each region $z$ has two corresponding rows: one is for upper bound $\bar z$ and another one is for lower bound $\underline z$.
@@ The component of row $\bar z$ is 1 for column $(i,j)$ such that $z(j) = z$; the component is 0 otherwise.
@@ The component of row $\underline z$ is $-1$ for column $(i,j)$ such that $z(j) = z$; the component is 0 otherwise.
\end{easylist}
We apply Lemma~\ref{lem:ghouila-Houri} to prove that $A$ is totally unimodular.
Fix any subsets of rows $R$. Let $I_R \subseteq I$ be the rows corresponding to $i \in I$ contained in $R$. $J_R$, $\bar Z_R$, and $\underline Z_R$ are defined analogously.

We classify the rows in $R$ by the following algorithm:
\begin{enumerate}
    \item
    Let $R_1 \coloneqq \bar Z_R$ and $R_2 \coloneqq \emptyset$.
    \item 
    For each $z \in \underline Z_R$, if $z \in \bar Z_R$, then update $R_1 \leftarrow R_1 \cup \{z\}$; otherwise, update $R_2 \leftarrow R_2 \cup \{z\}$.
    \item Let $J_R \coloneqq \{j(1), \dots, j(T_J)\}$, where $T_J \coloneqq |J_R|$,
    $R_1(1) \coloneqq R_1$, and $R_2(1) \coloneqq R_2$.
    For each $t \in [T_J]$,
    \begin{enumerate}
        \item Define a row vector
        \[
        r(t) \coloneqq \left(\sum_{i \in R_1(t)} a_{ik} - \sum_{i \in R_2(t)} a_{ik} \right)_{k}.
        \]
        \item If $r(t)_{i, j(t)} = 1$ for some $i$,
        then $R_1(t+1) \coloneqq R_1(t) $ and $R_2(t+1) \coloneqq R_2(t) \cup \{j(t)\}$.
        \item Otherwise, $R_1(t+1) \coloneqq R_1(t) \cup \{j(t)\}$ and $R_2(t+1) \coloneqq R_2(t)$.
    \end{enumerate}
    \item Let $R_1 \coloneqq R_1(T_J)$ and $R_2 \coloneqq R_2(T_J) \cup I_R$. Return $R_1$ and $R_2$.
\end{enumerate}
We will show why the algorithm above works.
First, under the hierarchical regional constraints, each component of $r(1)$ is either 0 or 1.
Next, let
    \[
    r \coloneqq \left(\sum_{i \in R_1(T_J)} a_{ik} - \sum_{i \in R_2(T_J)} a_{ik} \right)_{k}.
    \]
Note that, under the regional constraints, if $z(j) = z$ and $a_{j,k}=1$, then $a_{\bar z, k}=1$ and $a_{\underline z, k}=-1$.
Thus, by the construction of Step 4, all the components of $r(t)$ are either $0$ or $1$ for each $t$. 
Lastly, by construction, for any column $k$, $\sum_{i \in I_R} a_{ik} \in \{0,1\}$. Therefore, we have
    \[
\forall k \in [n],\quad \sum_{i \in R_1} a_{ik} - \sum_{i \in R_2} a_{ik} \in \{-1,0,1\}.
\]
By Lemma~\ref{lem:ghouila-Houri}, this implies that $A$ is totally unimodular.
\end{proof}

By the following well-known result, we can conclude that $(\mathrm{P_0'})$ has an integer optimal solution:
\begin{lemma}[\cite{hoffman2010integral}]
$A$ is totally unimodular iff, for any $b \in \Z^m$, $P \coloneqq \{x \in \R^n \colon Ax \leq b, x \geq 0\}$ is an integral polyhedra, i.e., all the faces includes an integer vector.
If $P$ is bounded, this is equivalent to that the components of all vertices of $P$ are integers.
\end{lemma}

This completes the proof of Proposition~\ref{prop:IE_integral_solution}.
\qed

\begin{example}[TU matching with regional constraints]
\label{example:tu_regional}
Let $I \coloneqq \{i_1, i_2\}$, $J \coloneqq \{j_1, j_2, j_3\}$, $Z \coloneqq \{z_1, z_2\}$, $z(j_1) = z(j_2) = z_1$, and $z(j_3)=z_2$.
The set of constraints can be written as $A_0 x \leq b$ by defining $A_0$, $x$, and $b$ as follows:
\begin{align}
A_0 &\coloneqq
\kbordermatrix{
& (i_1, j_1) & (i_1, j_2) & (i_1, j_3) & (i_2, j_1) & (i_2, j_2) & (i_2, j_3) \\
i_1 & 1&1&1&0&0&0\\
i_2 & 0&0&0&1&1&1\\
j_1 & 1&0&0&1&0&0\\
j_2 & 0&1&0&0&1&0\\
j_3 & 0&0&1&0&0&1\\
\bar z_1 & 1&1&0&1&1&0\\
\bar z_2 & 0&0&1&0&0&1\\
\underline z_1 & -1&-1&0&-1&-1&0\\
\underline z_2 & 0&0&-1&0&0&-1\\
(i_1, j_1) & -1&\\
(i_1, j_2) & &-1\\
(i_1, j_3) & &&-1\\
(i_2, j_1) & &&&-1\\
(i_2, j_2) & &&&&-1\\
(i_2, j_3) & &&&&&-1
} \\
x &\coloneqq (x_{i_1 j_1}, x_{i_1 j_2}, x_{i_1 j_3}; x_{i_2 j_1}, x_{i_2 j_2}, x_{i_2 j_3})^\top \\
b &\coloneqq (1,1,1,1,1,\bo_{z_1}, \bo_{z_2},\uo_{z_1}, \uo_{z_2}, 0, 0, 0, 0, 0, 0)^\top
\end{align}
The last $|I| \times |J|$ rows corresponds to non-negativity constraints $d_{ij} \geq 0$. Note that $[B \ I]$ is totally unimodular if $B$ is totally unimodular.
Thus, to show $A_0$ is totally unimodular, it suffices to show that
\[
A \coloneqq
\kbordermatrix{
& (i_1, j_1) & (i_1, j_2) & (i_1, j_3) & (i_2, j_1) & (i_2, j_2) & (i_2, j_3) \\
i_1 & 1&1&1&0&0&0\\
i_2 & 0&0&0&1&1&1\\
j_1 & 1&0&0&1&0&0\\
j_2 & 0&1&0&0&1&0\\
j_3 & 0&0&1&0&0&1\\
\bar z_1 & 1&1&0&1&1&0\\
\bar z_2 & 0&0&1&0&0&1\\
\underline z_1 & -1&-1&0&-1&-1&0\\
\underline z_2 & 0&0&-1&0&0&-1
}
\]
is totally unimodular.
\end{example}

\subsubsection*{Proof of Theorem~\ref{thm:IE_optimal_taxation}}
Consider the following dual problem of $(\mathrm{P_0'})$:

\begin{align}
    &(\mathrm{D_0'}) \left[
    \begin{matrix}
    \displaystyle{\max_{u,v,\bar w, \underline w \geq 0 }} &
    \displaystyle{\sum_{i} u_i + \sum_j v_j + \sum_{z \in Z} \bar o_z \bar w_z - \sum_{z \in Z_L} \underline o_z \underline w_z} \\[12pt]
    \mathrm{s.t.} & u_i + v_j \geq \Phi_{ij} - \bar w_{z(j)} + \underline w_{z(j)} & (i \in I, j \in J)
    \end{matrix}
    \right.
\end{align}

Let $d^*$ be an integer optimal solution to $(\mathrm{P_0}')$ (NB: this $d^*$ is also the optimal solution to $(\mathrm{P_0})$) and let $(u, v, \bar w, \underline w)$ be a solution to $(\mathrm{D_0'})$.
By a similar argument as in the standatd TU matching model that characterizes an stable outcome as a solution to the social welfare maximization problem and its dual problem, we can show that $(d^*, (u,v))$ is a stable outcome given $\Phi$ and $w^*$, where
$w^*_z \coloneqq \bar w_z \mathbbm{1}\{\bar w_z > 0\} - \underline w_z \mathbbm{1}\{\underline w_z > 0\}$. Note that $\bar w_z > 0$ and $\underline w_z > 0$ cannot happen simultaneously due to the complementary slackness condition.

\begin{cor}
\label{cor:limit_cap_based}
    Fix any cap-based policy characterized by $J' \subseteq J$ such that the resulting matching satisfies the regional constraints. The social welfare achieved under this policy is weakly less than the welfare attained under the optimal taxation policy $w^*$.
\end{cor}
\begin{proof}
    Fix any $J'$. The social welfare with $J'$ cannot be better than the optimal value of $(\mathrm{P_0})$) with additional constraints $d_{ij} = 0$ for any $i \in I$ and $j \in J'$.
\end{proof}

\subsection{Proof of Lemma~\ref{lemma:tu_discrete_choice}}
\label{app:tu_discrete_choice}
First, we will show the following lemma:
\footnote{
The following proof of Lemma~\ref{lem:ineq_systematic_surplus} is almost identical to the proof of Proposition 1 of \cite{galichon_cupids_2021}.
}
\begin{lemma}\label{lem:ineq_systematic_surplus}
Suppose that $U$ and $V$ are aggregate-level utilities given $(u,v)$ (i.e., \eqref{eq:agg_level_util} holds).
Suppose that $(d,(u,v))$ is a stable outcome. Then, we have $U_{xy} + V_{xy} \geq \Phi_{xy} + w_{z(y)}$ with equality when $\mu_{xy}>0$ for each $(x,y) \in T$.
\end{lemma}
\begin{proof}
Fix any $(x,y) \in T$.
First, we show $\Phi_{xy} - w_{z(y)} \leq U_{xy} + V_{xy}$. Suppose that $i \in x$ is matched with hospital $j \in y$.
We have
\begin{align}
    u_i &= \max_{j \in J} \{\tilde \Phi_{ij} - v_j\} \\
    &=
     \max_{y \in Y} \max_{j \in Y} \{\Phi_{ij} - w_{z(y)} - v_j\} \\
    &=
     \max_{y \in Y} \max_{j \in Y} \{\Phi_{xy} - w_{z(y)} + \epsilon_{iy} + \eta_{xj} - v_j\} \\
    &=
     \max_{y \in Y}
     \{
     \Phi_{xy} - w_{z(y)} + \epsilon_{iy}
     + \max_{j \in y} \{\eta_{xj} - v_j\}
     \} \\
    &=
     \max_{y \in Y}
     \{
     \Phi_{xy} - w_{z(y)} + \epsilon_{iy}
     - V_{xy}
     \} 
\end{align}
Thus, for any $i \in x$, we have
\begin{align}
    u_i &=
    \max \left\{
    \max_{y \in Y}
     \{
     \Phi_{xy} - w_{z(y)} + \epsilon_{iy}
     - V_{xy}
     \}, \epsilon_{i, y_0}
    \right\} \\
    &=
     \max_{y \in Y_0}
     \{
     \Phi_{xy} - w_{z(y)} + \epsilon_{iy}
     - V_{xy}
     \}.
\end{align}
Hence,
\[
\Phi_{xy} - w_{z(y)} \leq u_i - \epsilon_{iy} + V_{xy}.
\]
By taking the infimum over $i \in x$, we have
\[
\Phi_{xy} - w_{z(y)} \leq U_{xy} + V_{xy},
\]
for each $x \in X$ and $y \in Y$.

Next, suppose that $\mu_{xy}>0$. This implies that there exist $i \in x$ and $j \in y$ such that $d_{ij}=1$. For this pair, we have
$u_i + v_j = \Phi_{ij} - w_{z(y)}$. Suppose toward contradiction that $U_{xy}+V_{xy} > \Phi_{xy}$. By \eqref{eq:agg_level_util}, we have
$
\Phi_{xy} - w_{z(y)}<u_i-\epsilon_{iy}+v_j-\eta_{xj}$, and thus
$\Phi_{ij} - w_{z(y)} < u_i+v_j$. A contradiction.


\end{proof}

\begin{proof}[Proof of Lemma~\ref{lemma:tu_discrete_choice}]
Fix any type $x \in X$ and doctor $i \in x$.
By definition of $U_{xy}$, we have 
\begin{align}
U_{xy} &\leq u_i - \epsilon_{iy}, \ \forall y \in Y_0\\
\iff u_i &\geq U_{xy} + \epsilon_{iy}, \ \forall y \in Y_0\\
\iff u_i &\geq \max_{y \in Y_0} \{ U_{xy} + \epsilon_{iy} \}.
\end{align}
Similarly, for any type $y \in Y$ and doctor $j \in J$ with type $y$, we have $v_j \geq \max_{x \in X_0} \{ V_{xy} + \eta_{xj}\}$.

We want to claim that $u_i \leq \max_{y \in Y_0} \{ U_{xy} + \epsilon_{iy} \}$.
Suppose toward contradiction that there exists type $x \in X$ and doctor $i \in x$ such that
\begin{align}
u_i > \max_{y \in Y_0} \{ U_{xy} + \epsilon_{iy}\}.
\end{align}

First, consider the case where $i$ is matched with some hospital $j \in y$.
Then 
\begin{align}
\Phi_{ij} - w_{z(y)} &= u_i + v_j \\
&> \Bigl(\max_{y^\prime \in Y_0} \ U_{xy^\prime} + \epsilon_{iy^\prime} \Bigr) + \Bigl(\max_{x^\prime \in X_0} \ V_{x^\prime y} + \eta_{x^\prime j}\Bigr) \\
&\geq U_{xy(j)} + \epsilon_{iy(j)} + V_{xy(j)} + \eta_{xj} \\
&\geq \Phi_{xy} - w_{z(y)} + \epsilon_{i y(j)} + \eta_{x j} \quad (\because \ \text{Lemma~\ref{lem:ineq_systematic_surplus}})\\
&= \Phi_{ij} - w_{z(y)}.
\end{align}
A contradiction.
Next, consider the case where $i$ is unmatched. Then
\begin{align}
    u_i = \Phi_{i, y_0} =  \epsilon_{i, y_0}  
    > \max_{y \in Y_0} \{U_{xy} + \epsilon_{iy}\} 
    \geq
     \epsilon_{i, y_0}.
\end{align}
A contradiction. Therefore, we have $u_i \leq \max_{y \in Y_0} \{ U_{xy} + \epsilon_{iy} \}$ and hence $u_i = \max_{y \in Y_0} \{ U_{xy} + \epsilon_{iy} \}$. We can show $v_j = \max_{x \in X_0} \{ V_{xy} + \eta_{xj}\}$ in a similar manner.
\end{proof}




\subsection{Proof of Theorem~\ref{theorem:EAE_opt}}
\label{app:proof_thm_eae}

First, we show the strict concavity of $G$ and $H$.
\begin{lemma}
\label{lemma:tu_GH_strict_increasing_convex}
Under Assumptions 1-4, $G$ and $H$ are strictly increasing and strictly convex.
\end{lemma}
\begin{proof}
\noindent
\underline{$G$ is strictly increasing.}
Take any $U^1, U^2 \in \R^{N \times M}$ such that $U^1 \geq U^2$ and $U^1 \neq U^2$. Then $G(U^1) \geq G(U^2)$ by definition. In addition, note that $U^1_{xy} > U^2_{xy}$ holds for some $x \in X$ and $y \in Y$. Since $P_x$ has full support, we have
\begin{align}
\prob_{\epsilon_i}(u_i = U^1_{xy} + \epsilon_{iy}) \geq \prob_{\epsilon_i}(u_i = U^2_{xy} + \epsilon_{iy}) > 0.
\end{align}
Because $\E_{\epsilon_i} \left[ u_i \mid u_i = U_{xy} + \epsilon_{iy} \right]$ is strictly increasing in $U_{xy}$, we have 
\begin{align}
\E_{\epsilon_i} \left[ u_i \mid u_i = U^1_{xy} + \epsilon_{iy} \right] &\cdot \prob_{\epsilon_i}(u_i = U^1_{xy} + \epsilon_{iy}) \\
&> \E_{\epsilon_i} \left[ u_i \mid u_i = U^2_{xy} + \epsilon_{iy} \right] \cdot \prob_{\eta_j}(u_i = U^2_{xy} + \epsilon_{iy}),
\end{align}
and thus $G(U^1) > G(U^2)$ holds.

\noindent
\underline{$G$ is strictly convex.}
Take any $U^1, U^2 \in \R^{N \times M}$ and $s \in [0, 1]$. Since
\begin{align}
&s G(U^1) + (1-s) G(U^2) \\
&= \sum_x n_x \E\left[ 
    \Bigl( \max_y \ s( U^1_{xy} + \epsilon_{iy}) \Bigr) +
    \Bigl( \max_y \ (1-s)( U^2_{xy} + \epsilon_{iy}) \Bigr)
\right] \label{eq:strict_convex_1} \\
&\geq \sum_x n_x \E\left[ 
    \max_y \ s U^1_{xy} + (1-s) U^2_{xy} + \epsilon_{iy}
\right] \label{eq:strict_convex_2} \\
&= G \Bigl( s U^1 + (1-s) U^2 \Bigr)
\end{align}
holds, $G$ is a convex function.

Now suppose $U^1 \neq U^2$. Then $U^1_{xy} \neq U^2_{xy}$ holds for some $x \in X$, $y \in Y$. Without loss of generality, assume $U^1_{xy} > U^2_{xy}$. Since $P_x$ is of full support, 
\begin{align}
\prob\Biggl( \Bigl\{ \epsilon_i \colon 
    U^1_{xy} + \epsilon_{iy} > \max_{y' \neq y} U^1_{xy'} + \epsilon_{iy'} 
    \hspace{1em}\land\hspace{1em}
    \max_{y^\prime \neq y} U^2_{xy^\prime} + \epsilon_{iy^\prime} > U^2_{xy} + \epsilon_{iy}
\Bigr\} \Biggr) > 0
\end{align}
holds.
This implies that
\[
\Bigl( \max_y \ s( U^1_{xy} + \epsilon_{iy}) \Bigr) +
    \Bigl( \max_y \ (1-s)( U^2_{xy} + \epsilon_{iy}) \Bigr)
    >
    \max_y \ s U^1_{xy} + (1-s) U^2_{xy} + \epsilon_{iy}
\]
occurs with strictly positive probability, and thus $\eqref{eq:strict_convex_1} > \eqref{eq:strict_convex_2}$ holds.
Therefore, for any $s \in (0, 1)$, we have
\begin{align}
s G(U^1) + (1-s) G(U^2) > G \Bigl( s U^1 + (1-s) U^2 \Bigr),
\end{align}
which implies $G$ is strictly convex.
Similarly, we can show $H$ is also strictly increasing and strictly convex.
\end{proof}

We will show the following, which corresponds to the first part of the theorem:
\begin{lemma}
\label{lem:ea_dual_unique_sol}
If $G$ and $H$ are strictly convex and differentiable,\footnote{This holds under Assumptions~\ref{ass:indep_error}-\ref{ass:full_supp}.} \optprobref{opt:natural_agg_eq_primal} and \optprobref{opt:natural_agg_eq_dual} have unique solutions for any $\Phi$.
\end{lemma}
\begin{proof}
For \optprobref{opt:natural_agg_eq_primal}, since $G$ and $H$ are differentiable (WDZ theorem), $G^*$ and $H^*$ are strictly convex (Proposition D.14 of \cite{galichon2018optimal}). Therefore, the objective function of \optprobref{opt:natural_agg_eq_primal} is strictly concave in $\mu$, which implies the uniqueness of the solution.

For \optprobref{opt:natural_agg_eq_dual}, suppose toward a contradiction that there are two different optimal solutions $\alpha \coloneqq (U, V, \bw, \uw)$ and $\beta \coloneqq (U', V', \bw', \uw')$. Note that $\gamma \coloneqq \frac{1}{2}(\alpha + \beta)$ is also feasible. If either $U \neq U'$ or $V \neq V'$, then $\gamma$ gives a strictly lower value due to the strict convexity of $G$ and $H$, which contradicts the optimality of $\alpha$ and $\beta$.

Suppose that $U = U$ and $V= V'$. We must have $(\bw, \uw) \neq (\bw', \uw')$. Since $G$ is strictly convex, we have $\frac{\partial G}{\partial U_{xy}}>0$ for each $(x,y)$. By the complementary slackness condition with respect to $U_{xy}$, we have $U_{xy} + V_{xy} = \Phi_{xy} +\bw_{z(y)} -\uw_{z(y)}$ and $U'_{xy} + V'_{xy} = \Phi_{xy} +\bw_{z(y)}' -\uw_{z(y)}'$ for each $(x,y)$. Since $U=U'$ and $V=V'$, this implies that
\begin{align}
    \label{eq:w_equality}
    \bw_{z} -\uw_{z} = \bw_{z}' -\uw_{z}'.
\end{align}
for each $z$. Since $\bar o_z > \underline o_z$, we must have $\bw_z \uw_z = 0$; otherwise, there exists $\epsilon > 0$ such that $(U, V, \tilde \bw, \tilde \uw)$, where $\tilde \bw_z \coloneqq \bw_z - \epsilon$ and $\tilde \uw_z \coloneqq \uw_z - \epsilon$ attains a strictly lower value. Similarly, we have $\bw_z' \uw_z' = 0$. However, these combined with \eqref{eq:w_equality} imply that $\bw = \bw'$ and $\uw = \uw'$: if $\bw_z > 0$, then $\uw_z =0$, $\uw_z'=0$, and $\bw_z = \bw_z'$. If $\uw_z > 0$, then $\bw_z =0$, $\bw_z'=0$, and $\uw_z = \uw_z'$. If $\bw_z = \uw_z = 0$, then $\bw_z'=\uw_z'=0$.
A contradiction.
\end{proof}


By the complementary slackness condition, for each $z$, only one of the following holds: $\bar{w}_z^* > 0$ and $\underline{w}_z^* = 0$, $\bar{w}_z^* = 0$ and $\underline{w}_z^* > 0$, or $\bar{w}_z^* = \underline{w}_z^* = 0$. If we define
\[
w_z^* \coloneqq \1\{\bar{w}_z^* >0\}\bar{w}_z^* -  \1\{\underline{w}_z^* >0\}\underline{w}_z^*,
\]
then $w^* = (w_z^*)_z$ is the taxation policy, and we have
\[
\Phi_{xy} - \bar{w}_{z(y)}^* +\underline{w}_{z(y)}^* = \Phi_{xy} - w_{z(y)}^*,
\]
which corresponds to the stability condition with the gross surplus under taxation policy $w^*$.
Thus, the primal and dual problems jointly compute the equilibrium matching and corresponding social welfare under the taxation policy $ w^*$. The triple $(\mu^*, U^*, V^*)$ then constitutes the corresponding aggregate-level matching and utilities.
\qed

\subsection{%
\texorpdfstring{%
Derivation of two-way fixed effects Poisson regression of \cite{Galichon_poisson}
}{%
Derivation of two-way fixed effects Poisson regression of Galichon and Salani\'e (2021b)%
}%
}
\label{app:poisson_derivation}
We derive the estimator based on two-way fixed effects Poisson regression, which is proposed in \cite{Galichon_poisson}. The derivation is for the estimation of $\Phi$ and the same argument can be applied to the estimation of aggregate level utility, $U$ and $V$, which is transformed into a single-side fixed effect Poisson regression.

Consider a Poisson regression with two-way FE. The observation is matching between type $x$ and $y$: $\mu_{xy}$. We assume that $\mu_{xy} \sim \mathrm{Po}(\theta_{xy})$ where $\theta_{xy} = e^{\frac{\sum_{k} \lambda_k \phi_{xy}^{k} - u_x - v_y}{2}}$, $\mu_{x0}\sim \mathrm{Po}(e^{-u_x})$, and $\mu_{0y} \sim \mathrm{Po}(e^{-v_y})$.

The likelihood is as follows: note that we double-count the matchings when we construct the likelihood.
\begin{align}
    &\prod_{xy} \frac{ e^{ \mu_{xy} \frac{\sum_{k} \lambda_k \phi_{xy}^{k} - u_x - v_y}{2}} e^{-e^{\frac{\sum_{k} \lambda_k \phi_{xy}^{k} - u_x - v_y}{2}}}}{\mu_{xy}!} \prod_{xy} \frac{ e^{ \mu_{xy} \frac{\sum_{k} \lambda_k \phi_{xy}^{k} - u_x - v_y}{2}} e^{-e^{\frac{\sum_{k} \lambda_k \phi_{xy}^{k} - u_x - v_y}{2}}}}{\mu_{xy}!} \\
    &\prod_{x}\frac{e^{\mu_{x0} u_x} e^{-e^{-u_x}}}{\mu_{x0}!} \prod_{y}\frac{e^{\mu_{0y} v_y} e^{-e^{-v_y}}}{\mu_{0y}!}
\end{align}

Then the log-likelihood function is written as follows:
\begin{align}
    G(\lambda, u, v) &\equiv 2\sum_{xy} \ln e^{ \mu_{xy}\frac{\sum_{k} \lambda_k \phi_{xy}^{k} - u_x - v_y}{2}} e^{-e^{\frac{\sum_{k} \lambda_k \phi_{xy}^{k} - u_x - v_y}{2}}} + \sum_{x} \ln e^{-\mu_{x0} u_x} e^{-e^{-u_x}} + \sum_{y}\ln e^{-\mu_{0y} v_y} e^{-e^{-v_y}}\\
    &= 2\sum_{xy} \mu_{xy}\frac{\sum_{k} \lambda_k \phi_{xy}^{k} - u_x - v_y}{2} -e^{\frac{\sum_{k} \lambda_k \phi_{xy}^{k} - u_x - v_y}{2}} + \sum_{x} -\mu_{x0} u_x - e^{-u_x} + \sum_{y} -\mu_{0y}v_y - e^{-v_y}\\
    &= \sum_{xy} \mu_{xy}\left( \sum_{k} \lambda_k \phi_{xy}^{k} - u_x - v_y\right) - 2e^{ \frac{\sum_{k} \lambda_k \phi_{xy}^{k} - u_x - v_y}{2}}\\
    &\quad - \sum_{x} \mu_{x0} u_x - \sum_{x} e^{-u_x} - \sum_{y} \mu_{0y}v_y - \sum_{y}e^{-v_y}
\end{align}

The minimization objective function is
\begin{align}
    F(\lambda, u, v) &= \sum_{x} e^{-u_x} + \sum_{y} e^{-v_y} + 2\sum_{xy}e^{\frac{\sum_{k} \lambda_k \phi_{xy}^{k} - u_x - v_y}{2}} - \sum_{xy} \mu_{xy} \left(\sum_{k} \lambda_k \phi_{xy}^{k} - u_x -v_y\right)\\
    &+ \sum_{x}\mu_{x0} u_x + \sum_{y} \mu_{0y} v_y
\end{align}
As shown in \cite{Galichon_poisson}, minimizing $F$ is equivalent to the moment matching estimator.

\section{Comparison between Aggregate Equilibria, Stable Outcomes, and Deferred Acceptance Outcomes}
\label{sec:comparison}

In this section, we clarify the relationships among aggregate equilibrium (AE), stable outcome, and the outcome under the Deferred Acceptance (DA) mechanism. In our model, agents face nonzero outside options: doctor \(i\) obtains utility \(\epsilon_{i0}\) when unmatched, and slot \(j\) obtains utility \(\eta_{0j}\) when unmatched.

By construction, the social welfare achieved under the DA mechanism is lower than that under the stable outcome. Furthermore, we demonstrate that, in a suitably defined large market limit, the per-capita social welfare under the stable outcome and AE converges. These results are established through both theoretical analysis (Section~\ref{app:theory_AEIEDA}) and simulation (Section~\ref{app:simulation_AEIEDA}).

\subsection{Theory}
\label{app:theory_AEIEDA}
Below, we claim that social welfare \textit{per capita} (i.e., total surplus divided by the number of doctors and positions) under the stable outcome almost surely converges to the (modified version of) optimal value of AE. For notational simplicity, we assume $w_z \equiv 0$ below. For general taxation policy $z$, the same argument goes through with $\Phi_{ij}$ being replaced by the gross surplus.

Let $L \coloneqq |I| + |J|$. We consider the large market limit in which $n_x$ and $m_y$ tends to $\infty$ for all $x,y$ with
\begin{align}
\label{eq:limit_fraction}
    \frac{n_x}{L} \to p_x, \quad \frac{m_y}{L} \to p_y \quad (n_x, m_y \to \infty).
\end{align}

Let $(u_i)_i, (v_j)_j$ denote the equilibrium payoff profiles under the stable outcome with non-zero-valued outside option values, which is the solution to the following dual problem $(D_S)$:
\begin{align}
    (D_S)
    \left[
    \begin{matrix}
        \displaystyle{\min_{u, v}} & \displaystyle{\sum_{i \in I} u_i + \sum_{j \in J} v_j} \\
        s.t. & u_i + v_j \geq \Phi_{ij} \\
         & u_i \geq \epsilon_{i0} \\
         & v_j \geq \eta_{0j}
    \end{matrix}
    \right.
\end{align}
The corresponding primal problem $P_S$ is
\begin{align}
    (P_S)
    \left[
    \begin{matrix}
        \displaystyle{\max_{\mu}} & \displaystyle{\sum_{ij \in IJ} \mu_{ij} \left(\Phi_{ij} - \epsilon_{i0} - \eta_{0j}\right) + \sum_{i \in I} \epsilon_{i0} + \sum_{j \in J} \eta_{0j}} \\
        s.t. & \text{$\mu$ is feasible}
    \end{matrix}
    \right.
\end{align}
Observe that the value of the objective function of $(P_S)$ corresponds to social welfare under the stable outcome with non-zero-valued outside options since $(u,v)$ satisfies the corresponding stability condition (the constraints in $(D_S)$) and
\begin{align}
    &\sum_{ij \in IJ} \mu_{ij} \left(\Phi_{ij} - \epsilon_{i0} - \eta_{0j}\right) + \sum_{i \in I} \epsilon_{i0} + \sum_{j \in J} \eta_{0j} \\
&=
\sum_{ij \in IJ} \mu_{ij} \Phi_{ij} + \sum_{i \in I} \left(1 - \sum_j \mu_{ij} \right) \epsilon_{i0} + \sum_{j \in J} \left(1 - \sum_i \mu_{ij} \right) \eta_{0j},
\end{align}
which means that the matched pair with $\mu_{ij}=1$ yields joint surplus $\Phi_{ij}$ and the unmatched agents yields surplus $\epsilon_{i0}$ or $\eta_{0j}$.

We consider a version of the dual for AE:
\begin{align}
\label{eq:AE_normalized}
\bar V(AE) \coloneqq 
\min_{U, V}  \left\{
\sum_{x \in X} p_x \E \left[ \max_{y \in Y_0} \left\{U_{xy} + \epsilon_{iy} \right\} \right]
+
\sum_{y \in Y} q_y \E \left[ \max_{y \in Y_0} \left\{V_{xy} + \epsilon_{xj} \right\} \right]
\right\}.    
\end{align}
Note that the weights $p_x$ and $q_y$ are used instead of $n_x$ and $m_y$ in the current main text.

The social welfare per capita under the stable outcome with non-zero-valued outside options is denoted by
\begin{align}
\bar V(S) &\coloneqq 
    \frac{1}{L}\left[\sum_{x \in X}\sum_{i \in x} u_i + \sum_{y \in Y} \sum_{j \in y} v_j \right] \\
    &=
    \sum_{x \in X}\frac{n_x}{L} \frac{1}{n_x}\sum_{i \in x} u_i
    +
    \sum_{y \in Y} \frac{m_y}{L} \frac{1}{m_y}\sum_{j \in y} v_j
    .
    \label{eq:sw_IE}
\end{align}
By Lemma~\ref{lemma:tu_discrete_choice}, for any $i \in x$ and $j \in y$, we have
\begin{equation}
    u_i = \max_{y \in Y_0} \left\{U_{xy} + \epsilon_{iy} \right\},
\quad
v_j = \max_{x \in X_0} \left\{V_{xy} + \epsilon_{xj} \right\}.
\label{eq:disc_choice}
\end{equation}
By the SLLN and \eqref{eq:disc_choice}, we have
\[
\frac{1}{n_x}\sum_{i \in x} u_i \to \E \left[ \max_{y \in Y_0} \left\{U_{xy} + \epsilon_{iy} \right\} \right], \ a.s. \quad (n_x \to \infty),
\]
and
\[
\frac{1}{m_y}\sum_{j \in y} v_j \to \E \left[ \max_{y \in Y_0} \left\{V_{xy} + \epsilon_{xj} \right\} \right], \ a.s. \quad (m_y \to \infty).
\]
By \eqref{eq:limit_fraction}, we have
\begin{align}
    \bar V(S) \to \bar V(AE), \ a.s. \quad (n_x, m_y \to \infty).
\end{align}
\qed

We compute $(D_S)$ for the stable outcome in our simulation, and compare its social welfare per capita (i.e., the optimal value of $(D_S)$ divided by $L$) with the social welfare per capita in AE computed by \eqref{eq:AE_normalized}.

\subsection{Simulation}
\label{app:simulation_AEIEDA}
We consider a matching market with four medical schools, denoted by 
$
X = \{ x_1, x_2, x_3, x_{4} \},
$
and three hospitals, denoted by 
$
Y = \{ y_1, y_2, y_3 \}.
$
Each medical school admits \(n\) students, and each hospital offers \(m\) available slots. The aggregate-level joint surplus generated by a match between a university \(x\) and a hospital \(y\) is given by the following matrix:
\[
(\Phi_{xy})_{xy} = \begin{pmatrix}
5 & 3 & 2\\
0 & 4 & 1\\
4 & 2 & 3\\
-1 & 2 & 2
\end{pmatrix}.
\]

Given the matrix \((\Phi_{xy})_{xy}\) and the population of doctors and hospital slots, we first compute the AE for this market. This computation yields the number of matches across all types as well as the aggregate-level utilities \(U_{xy}\) and \(V_{xy}\). In addition, we obtain the welfare measures \(G\) and \(H\) for both sides.

At the individual level, the utility for a doctor \(i\) of type \(x\) and for a slot \(j\) of type \(y\) is given by
\begin{align}
    U_{xy} + \epsilon_{iy} \quad \text{and} \quad V_{xy} + \eta_{xj},
\end{align}
respectively. Summing these, we define the individual-level social surplus as
\begin{align}
    \Phi_{ij} = U_{xy} + \epsilon_{iy} + V_{xy} + \eta_{xj}.
\end{align}
Note that this individual-level utility is defined only for doctor–slot pairs; in cases of unmatch, a doctor \(i\) obtains utility \(\epsilon_{i0}\) and a slot \(j\) obtains utility \(\eta_{0j}\). We assume all unobserved terms \((\epsilon_{iy})_{y\in Y_{0}}\) and \((\eta_{xj})_{x\in X_0}\) follow a Gumbel distribution.

Using these simulated values of utility and surplus, we compute two types of market outcomes: the stable outcome and the outcome from the deferred acceptance algorithm. For the stable outcome, we solve \((D_S)\) to obtain the equilibrium. For the DA outcome, we first construct preference lists based on the computed utilities and then implement the standard DA algorithm to determine the matching.

We examine four market settings in which \(n\) and \(m\) grow, with the pairs given by 
\[
(n, m) \in \{ (20,40),\ (40,60),\ (60,80),\ (80,100) \}.
\]
For each market setting, we compute the equilibrium outcome 50 times and report the average results.  
Figure~\ref{fig:ae_ss_comparison} compares the three market outcomes. First, AE closely approximates both the average welfare on each side and the matching outcomes observed in the stable outcome, even when the numbers of doctors and slots are relatively small. Second, as expected, the average social welfare in the DA outcome is lower than that in the stable outcome. Nevertheless, as the market size increases, the welfare measures and matching outcomes under DA exhibit trends similar to those observed in AE and the stable outcome.

\begin{figure}
    \centering
    \includegraphics[width=0.9\linewidth]{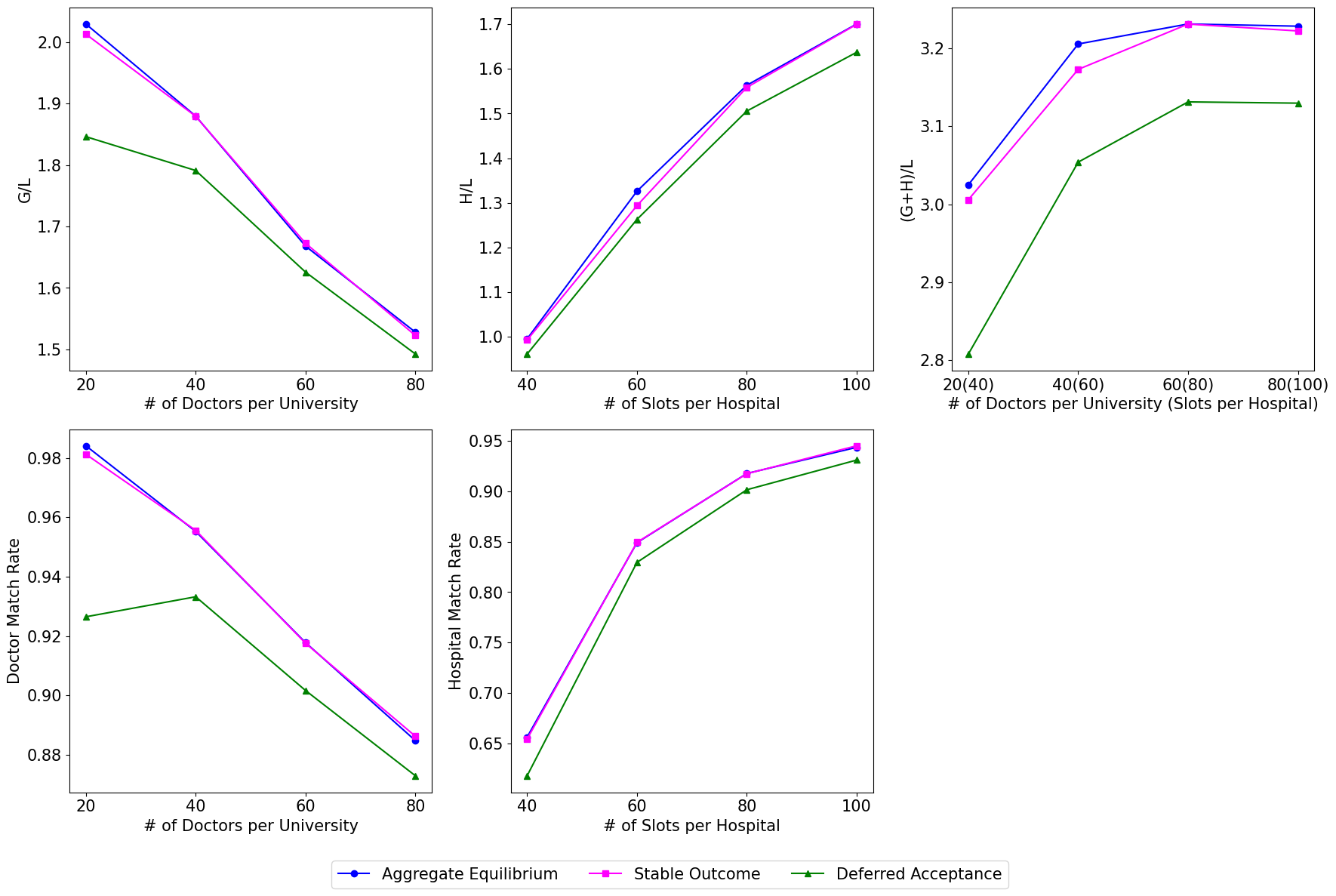}
    \caption{AE and individual market outcomes}
    \label{fig:ae_ss_comparison}
    \vskip 10pt 

    \begin{footnotesize}
    \begin{minipage}{15cm}
	\textit{Note}: The upper panels present the welfare measures. The leftmost panel compares the average welfare on the doctors' side, the center panel compares the average welfare on the hospitals' side, and the rightmost panel compares the total social welfare, which is the sum of the values in the first two panels. The lower two panels compare the matching outcomes. The left panel shows the match rate on the doctors' side, defined as the ratio of matched doctors to the total number of doctors, and the right panel shows the match rate on the hospitals' side, defined as the ratio of filled slots to the total number of slots. The blue line represents AE, the magenta line represents the stable outcome, and the green line represents the DA outcome.
	\end{minipage} 			
    \end{footnotesize}
\end{figure}

\section{Monte Carlo Simulation}
\label{sec:MC}
We start by describing the overall setting of the Monte Carlo simulation. All the detailed parameter values are left to Appendix \ref{sec:appendix_MC}.
There are $10$ prefectures, numbered from $0$ to $9$, grouped into three regions: $\{0,1\}\in R_0$, $\{2,3,4,5\}\in R_1$, and $\{6,7,8,9\}\in R_2$. 
$R_0$ represents an urban area, and $R_1$ and $R_2$ are rural areas. The government is concerned about the inefficient supply of medical services in $R_2$ and attempts to meet a lower bound in terms of the number of matches in the region.

We have a total of $20$ hospitals. Each hospital is placed in one of the prefectures based on a multinomial distribution. Hospital characteristics are modeled dynamically. When we denote each hospital by $h$, each hospital's capacity, denoted by $c_{ht}$, starts with a Poisson distribution at time $t = 0$ and evolves over time through a stochastic process involving increments and decrements modeled by independent Poisson distributions. We use $j$ to denote each slot in a hospital.
Other hospital-specific characteristics, like the number of beds, are captured by a variable $z_{ht}$, which follows a normal distribution.

We have $200$ doctors and they are distributed among the prefectures in a similar way to the hospitals. Each doctor belongs to one of $20$ medical schools, and the schools themselves are distributed among prefectures, also based on a multinomial distribution. Each school has an equal split of the doctors in the same prefecture.
The schools have characteristics such as average ability measures that follow a normal distribution. We use $s$ to denote the school and $i$ to denote a doctor.

We define the net joint surplus generated by a matching between a slot $j$ and a doctor $i$ at time $t$ in the following way:
\begin{align}
    \Phi_{ijt} = \Phi_{sht} + \xi_{ijt},
\end{align}
where
\begin{align}
    \Phi_{sht} = U^{\mathrm{base}}_{sht} + V^{\mathrm{base}}_{sht}, \ \xi_{ijt} = \epsilon_{iht} + \eta_{sjt},
\end{align}
and
\begin{equation}
\begin{aligned}
    &U^{\mathrm{base}}_{sht} = \beta_{1,1} w_{1, ht} + \beta_{1,2} w_{2,ht} + \beta_2 \left| l_{s} - l_{h} \right| + \beta_3 1\{h\in R_1\ \text{or}\ h\in R_2\},\\
    &V^{\mathrm{base}}_{sht} = \gamma_1 x_{1,st} + \gamma_2 x_{2,st},\\
    &\epsilon_{iht} \sim Ex1, \ \eta_{sjt}\sim Ex1.
\end{aligned}
\label{eq:MC_specification}
\end{equation}
$\left|l_s - l_h\right|$ is a measure of the distance between school $s$ and hospital $h$: in this simulation, we define this as the absolute value of the gap between the prefecture index. And the last term in $U^{\mathrm{base}}_{sht}$ captures the negative impact on the utility from living in rural areas. Note that these rural areas include $R_1$, which is not the target of the subsidy to ensure the lower bound on the matching outcomes. 
We also use $U_{ijt} = U^{\mathrm{base}}_{sh(j)t} + \epsilon_{ih(j)t}$ and $V_{ijt} = V^{\mathrm{base}}_{s(i)ht} + \eta_{s(i)jt}$ to denote the individual level preferences.

\begin{figure}
\begin{center}
    \centering
    \includegraphics[width = 15cm]{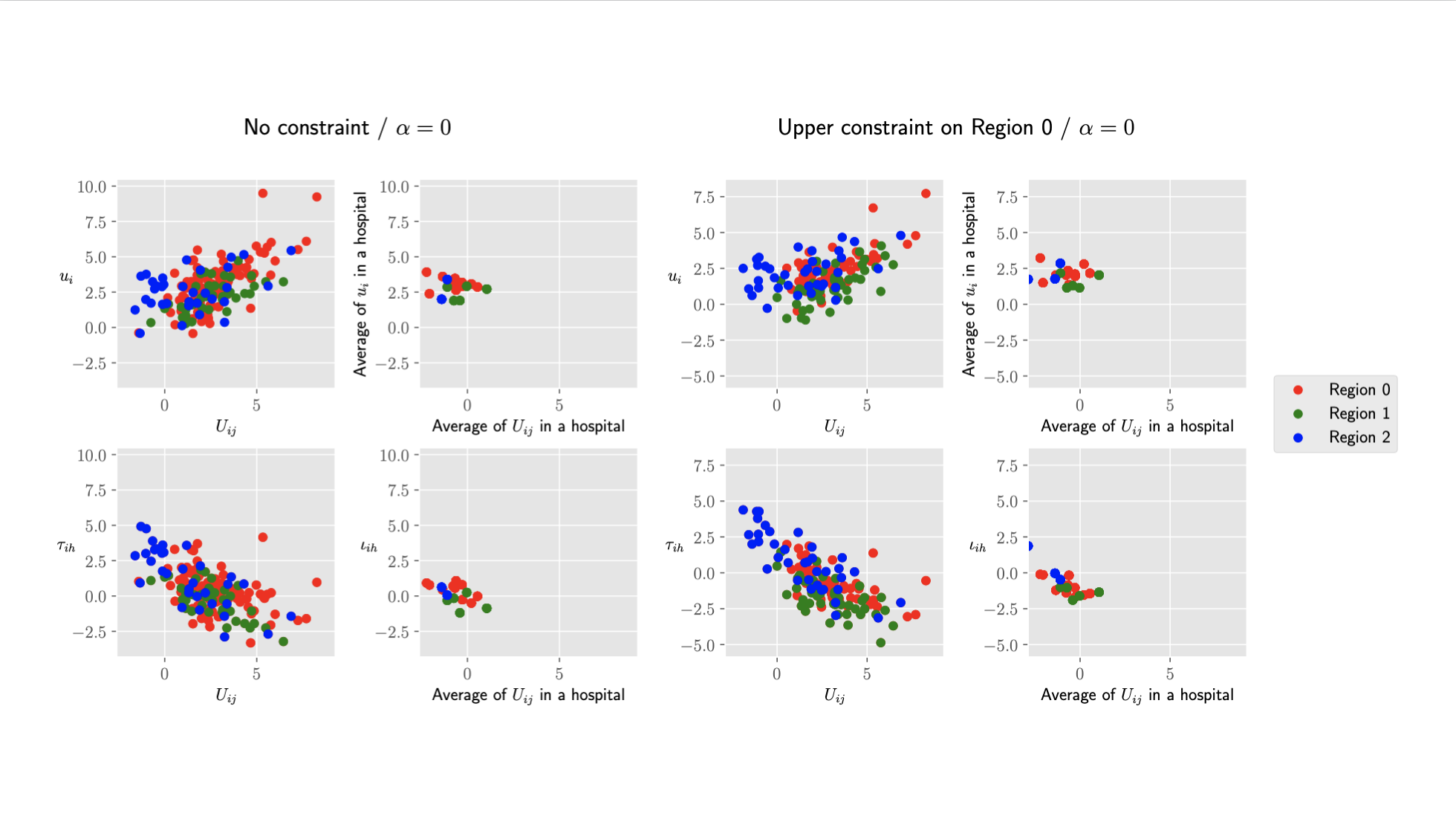}
    \caption{Simulated Stable Outcomes}
    \label{fig:simulation_scatter}
\end{center}
\end{figure}

\subsection{Simulation}
We compute a stable outcome of an instance of the above market at one time period. The number of matches in each region is $94$, $43$, and $33$. The number of unmatched doctors is $30$ and the number of unmatched slots is $46$. 
Imagine that the government sets an upper bound on $R_0$ to increase the number of matches in rural regions. When we set the upper bound on $R_0$ to $60$, the equilibrium numbers of matches are: $60$, $45$, and $37$. The number of unmatched doctors is $58$ and the number of unmatched slots is $74$.
Under this regional constraint, the tax levied on the matchings in $R_0$ is $3.149$.

Figure \ref{fig:simulation_scatter} depicts the scatter plots of several equilibrium objects when we set $\alpha = 0$, which implies that the entire amount of tax is levied on the hospital side.
The left four panels are obtained when we set no regional constraint, and the right four panels are obtained when we set an upper bound on $R_0$ to $60$.
The four panels in each left and right half set of panels depict the same things for the case of no constraint and regional constraint.
The first and third columns are the results in the stable outcome (with optimal tax), where each dot represents a match. The upper panels are the scatter plots of preference of doctor $i$ has for slot $j$, $U_{ij}$, and the utilities attained in a stable outcome, $u_{i}$. The lower panels are the scatter plots of $U^{\mathrm{base}}_{ij}$ and the transfer in a stable outcome, $\tau_{ih}$. 
The second and fourth columns represent aggregate level objects: the upper panels are the scatter plots of the average of $U_{ij}$ and $u_i$ among a matches in a hospital, and the lower panels are the scatter plots of the average of $U_{ij}$ and aggregate-level transfer, $\iota_{ih}$, from a hospital. In all scatter plots, a red marker represents a match or a hospital in $R_0$, a green marker for $R_1$, and a blue marker for $R_2$.

\begin{figure}
\begin{center}
    \centering
    \includegraphics[width = 15cm]{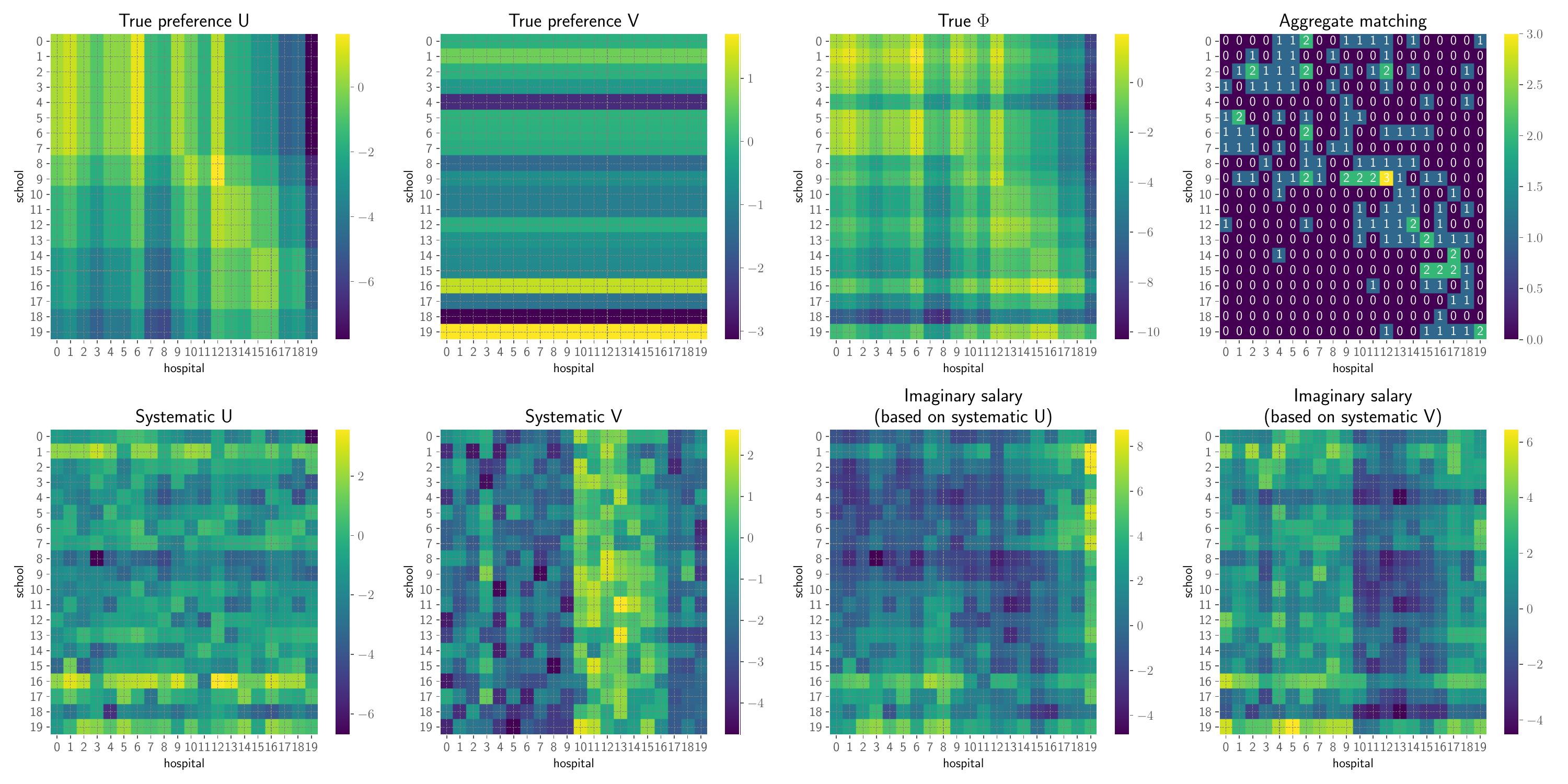}
    \caption{Aggregate Objects}
    \label{fig:aggregate_simulate}
\end{center}
\end{figure}

As expected, the utility attained in a stable outcome is higher when a doctor can be matched with a preferred slot, whereas the transfer decreases. This decrease is also reflected in a decrease in aggregate-level transfers from a hospital: when the average of $U_{ij}$ in a hospital match increases, the aggregate-level transfers from the hospital decrease.
The impact of a regional constraint on the aggregate-level transfers is clear: in the constrained region, $R_0$, they decrease under the constraint compared with the case of no constraint.
This is true in the level sense, and the decrease is larger than the changes in other regions.
Note that the changes in the aggregate-level transfers and their sizes depend on the value of $\alpha$. For example, in the extreme case of $\alpha = 1$, the aggregate-level transfers in $R_0$ increases under the regional constraint. Hence, it is important to estimate the division of tax on the hospital side and the school side.

Hereafter, we set $\alpha = 0.2$.
Figure \ref{fig:aggregate_simulate} summarizes the aggregate-level objects computed based on the simulated stable outcome. 
In all the heatmaps, the horizontal axis represents hospitals and the vertical axis represents schools.
Aggregate matching is depicted in the upper panel in the rightmost column. The number annotated in each cell represents the number of matches between a hospital and a school.
Aggregate-level utilities are computed following the definition stated in \eqref{eq:agg_level_util}. 

For the ease of argument, we name the gap between the aggregate-level utilities and the aggregate-level base utilities by \textit{imaginary salary}: the imaginary salary from school is defined as $\chi_{sh}^{U} \equiv U_{sh} - U^{\mathrm{base}}_{sh}$ and the same one from the hospital side is defined as $\chi_{sh}^{V} \equiv V^{\mathrm{base}}_{sh} - V_{sh}$.\footnote{This name is from the fact that the results of the following two are the same: (1) the agents in one side chooses the agent in the other side by comparing the sum of preference term, imaginary salary, and individual disturbance and (2) Aggregate matching outcome.}
The lower left two panels in Figure \ref{fig:aggregate_simulate} show the imaginary salaries between schools and hospitals. The number of doctors in our simulation is $200$, which is insufficient for approximating the market with an infinite number of doctors. This makes the gap between the two imaginary wages computed based on $U$ and $V$.\footnote{We can show that these two must be equal in the infinite sample case.}

\begin{figure}[t]
\begin{center}
    \centering
    \includegraphics[width = 15cm]{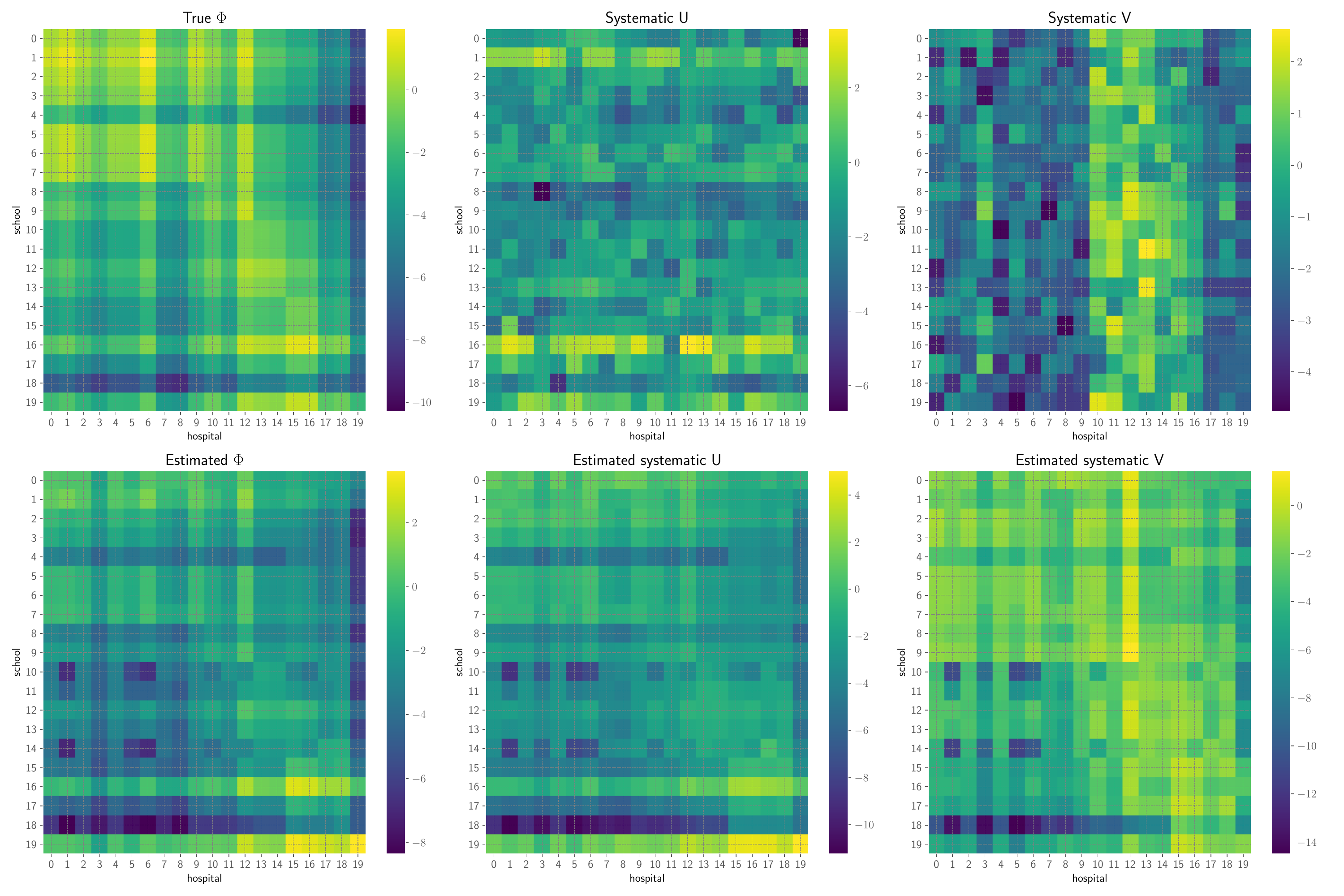}
    \caption{Estimation Results of the First Stage}
    \label{fig:estresults_first_stage_sim}
\end{center}
\end{figure}

\subsection{Estimation}
The estimation results in the first stage are depicted in Figure \ref{fig:estresults_first_stage_sim}. The upper panels are the heatmap of the true values of $\Phi_{sh}$, $U_{sh}$, and $V_{sh}$. They are the estimation targets. The lower panels display the estimation results corresponding to the upper panels. We set the degree of the polynomials to two. The estimated social surplus follows similar patterns to the true social surplus, whereas the estimated aggregate-level utilities show different patterns from the true values. These gaps are due to the incompleteness of polynomial approximations in equation \eqref{eq:systematic_polynomial_approx}. In practice, we handle this problem by including non-linearly transformed base variables when making polynomial series.

\begin{table}[t]
\begin{center}
\centering
\caption{Estimates}
\scalebox{0.9}{
\begin{tabular}{@{}cccccccccc@{}}
\toprule
Parameter           & $\gamma_1$    & $\gamma_2$    & $\beta_{1,1}$ & $\beta_{1,2}$ & $\beta_3$    & $\beta_2$     &  $ \frac{\alpha}{T}\sum_{t}w_{0t}$ & $ \frac{1-\alpha}{T}\sum_{t}w_{0t}$  \\ \midrule
Estimate            & 0.151         & -0.0384       & 0.907         & -0.490        & -0.527       & -0.879       & 0.265             & 1.471            \\
Standard Dev.       & (0.178)       & (0.220)       & (0.0670)      & (0.0496)      & (0.0922)     & (0.130)      & (0.128)           & (0.156)          \\
True Value          & 1.00          & -0.400        & 1.00          & -0.500        & -1.00        & -1.00       & 0.378             & 1.510            \\ \bottomrule
\end{tabular}}
\label{tab:MC_est_table}
\end{center}
\end{table}

For the second-stage estimation, we simulate matching outcomes over two periods. In the first period, the government set the upper bound on region $0$ to $80,$ and in the second period, the upper bound is changed to $60$. 
Because, in this exercise, we assume that the true value of aggregate-level transfers is observable, we use the moment conditions \eqref{eq:moment_conditions} directly to construct a minimum distance estimator. We leave the details of the construction of this estimator in Appendix \ref{sec:appendix_MC}.
In this exercise, we use the time-average version of the moment conditions, and therefore, the tax term is simply identified as the time-average of the levied tax.

Table \ref{tab:MC_est_table} summarizes the estimation results of the second stage.
The first six columns are the structural parameters in equation \eqref{eq:MC_specification}. The last two columns are the average taxes levied on the doctor side and the hospital side.
From these estimation results, the estimate of $\alpha$ is $0.153$, whereas the true value is $0.2$.
Based on these estimates, we can conduct a counterfactual analysis: for example, the taxes in the alternative regional constraints, the matching outcomes, and the salaries are obtained by solving the equilibrium.

\subsection{Minimum distance estimator}
\label{sec:appendix_MC}
We use the time average of both sides to construct the moment conditions \eqref{eq:moment_conditions}. Note that, in this case, it is impossible to identify $\delta_{ht}^{U}$ and $\delta_{ht}^{V}$ for every $t$ because the summation of them with respect to time determines the moment values. Hence, all we can identify is the average tax levied over the time periods.
We define a function $g$ to represent the moment conditions:
\begin{align}
    g(\theta) \equiv \begin{pmatrix}
        \frac{1}{T} \sum_{t}\left( \sum_{s} \omega_{s1t} \left(X_{s1t}^{U\prime}\beta_U - \delta^{U}_{Ht} \right) - \sum_{s}\omega_{s1t} \hat{\tilde{U}}_{s1t} - \iota_{1t}\right)\\
        \vdots\\
        \frac{1}{T} \sum_{t}\left( \sum_{s} \omega_{sHt} \left(X_{sHt}^{U\prime}\beta_U - \delta^{U}_{Ht} \right) - \sum_{s}\omega_{sHt} \hat{\tilde{U}}_{sHt} - \iota_{Ht}\right)\\
        \frac{1}{T}\sum_t \left(\sum_{s} \omega_{s1t} \left(X_{s1t}^{V\prime}\beta_V - \delta^{V}_{1t} \right) - \sum_{s}\omega_{s1t} \hat{\tilde{V}}_{s1t} - \iota_{1t}\right)\\
        \vdots\\
        \frac{1}{T}\sum_t \left(\sum_{s} \omega_{sHt} \left(X_{sHt}^{V\prime}\beta_V - \delta^{V}_{Ht} \right) - \sum_{s}\omega_{sHt} \hat{\tilde{V}}_{sHt} - \iota_{Ht}\right)
    \end{pmatrix}.
\end{align}


Our estimator is the minimum distance estimator where the moment condition is specified in (\ref{eq:moment_conditions}). When we write the asymptotic variance of $\hat{\tilde{U}}_{sht}$ and $\hat{\tilde{V}}_{sht}$ by $S^U_{t}$ and $S^V_{t}$, the optimally weighted minimum distance estimator is defined as follows: 
\begin{align}
    \hat{\theta} \equiv \argmin_{\theta}\ g^{\prime}(\theta) S^{-1} g(\theta),
\end{align}
where
\begin{align}
    S = \begin{pmatrix}
        \frac{1}{T^{2}}\sum_t S^U_{t} & 0\\
        0 & \frac{1}{T^{2}}\sum_t S^V_{t}
    \end{pmatrix}.
\end{align}

We can compute the asymptotic distribution of the estimator as follows, and the standard error can be obtained directly\footnote{We assume that the polynomial approximation regarding the systematic utility is correct. When there is a misspecification, we must treat the bias due to the misspecified model, which is beyond the scope of this study.}.
As the Poisson regression in the first step has an explicit form of $S_t^{U}$ and $S_t^{V}$, we can directly compute the estimates of the standard errors for every parameter by inserting the estimated results. 
\begin{theorem}
    Under the regularity conditions, the asymptotic distribution of $\hat{\theta}$ is as follows:
    \begin{align}
        \sqrt{ST}\left( \hat{\theta} - \theta \right)\xrightarrow[]{d} N\left(0, \left(\Gamma^{\prime} S^{-1} \Gamma\right)^{-1}\right),
    \end{align}
    where
    $\Gamma = \frac{\partial}{\partial \theta}g(\theta)$.
\end{theorem}

\section{Additional analysis}
\subsection{Test of implicit tax on urban areas}
\label{sec:appendix_AE_EAE}
Under the excessive competition for the slots in urban counties, it is possible that the surplus generated by matches in urban counties has already been diminished due to some external forces: for example, as the number of slots decreases, it becomes more difficult for residency programs to secure funding.
Given this consideration, the marginal effect estimates for locations in Tokyo or other urban areas may be underestimated. 
Here, we examine whether matches in urban counties are subject to an implicit tax under the current market outcome.

\paragraph{Empirical strategy}
We define an individual level transfer as in the main analysis.
Fix any period $t$ and an implicit taxation policy $w_t = (w_{zt})_z$. We define \textit{individual-level transfer from hospital $h$ to doctor $i$ with a ratio $\alpha$}, denoted by $\tau_{iht}$, as follows:
\begin{align}
    \label{def:transfer}
        \tau_{iht} \coloneqq u_{it} - \left(U^{\mathrm{base}}_{sht} + \epsilon_{iht} - \alpha w_{zt}\right).
\end{align}
Tax $w_{zt}$ is levied on the matched pair of doctor $i$ and hospital $h$. The doctor incurs fraction $\alpha$ of the tax; thus doctor's payoff \emph{without transfer} were to be $U_{ijt}^{\mathrm{base}} - \alpha w_{rt} = U^{\mathrm{base}}_{sht} + \epsilon_{iht} - \alpha w_{rt}$. In equilibrium, doctor $i$ enjoys equilibrium payoff $u_{it}$, which could be different from $U_{ijt}^{\mathrm{base}}$. We interpret the difference between equilibrium payoff and payoff without transfer as the individual-level transfer from the hospital side to the doctor side.

Now we define an \emph{aggregate-level transfer} as the average of the individual-level transfer in a hospital $h$ and denote it by $\iota_{ht}$:
\begin{align}
    \iota_{ht} \coloneqq \frac{1}{\left|D(h)_t \right|} \sum_{i\in D(h)_t}  \tau_{iht},
\end{align}
where $D(h)_{t}$ is the set of doctors matched with any slot of hospital $h$ at time $t$.
We have the same moment conditions for this case as in the main analysis.


As in the main analysis, we model the base utilities as a linear function of observable characteristics. In addition to them, we define $\delta^{U}_{ht} \equiv \alpha w_{r(h)t}$ and $\delta^{V}_{ht} \equiv (1-\alpha) w_{r(h)t}$ as the levied implicit tax on school side and hospital side in period $t$ and treat them as parameters to be estimated.
Our parameters of interest are the following: $\beta_U$, $\beta_V$, $\delta^{U}_{ht}$ for every pair of $h$ and $t$, and $\delta^{V}_{ht}$ for every pair of $z$ and $t$.
We use $\theta$ to indicate the vector of these parameters: $\theta \coloneqq \left(\beta_U, (\delta^{U}_{ht})_{h,t}, \beta_V, (\delta^{V}_{ht})_{h,t} \right)$.

For the estimation of the aggregate-level utilities are same as our main analysis.
For the second step, we construct the following moment conditions for $\theta$:
\begin{equation}
    \begin{aligned}
        \label{eq:moment_conditions_tax_app}
    \sum_{s} \omega_{sht} \left(X_{sht}^{U,\mathrm{base}\prime}\beta_U - \delta^{U}_{ht} \right) = \sum_{s}\omega_{sht} \hat{U}_{sht} - \iota_{ht},\  \forall \ h,t \\
    \sum_{s} \omega_{sht} \left(X_{sht}^{V,\mathrm{base}\prime}\beta_V - \delta^{V}_{ht} \right)= \sum_{s}\omega_{sht} \hat{V}_{sht} + \iota_{ht},\  \forall \ h,t 
    \end{aligned}
\end{equation}
By adopting the same measurement model, the estimating equations are as follows:
\begin{equation}
    \begin{aligned}
    \label{eq:estimating_equation_app}
    &\sum_{s}\omega_{sht} \hat{U}_{sht} = \gamma_{0,U} + \gamma_{1, U}S_{ht} + \sum_{s} \omega_{sht} \left(X_{sht}^{U\prime}\beta_U - \delta^{U}_{ht} \right) + \psi_{ht}^{U}\\
    &\sum_{s} \omega_{sht}\hat{V}_{sht} = \gamma_{0,V} + \gamma_{1,V} S_{ht} + \sum_{s} \omega_{sht} \left(X_{sht}^{V\prime}\beta_V - \delta^{V}_{ht} \right) + \psi_{ht}^{V}.
    \end{aligned}
\end{equation}
We estimate these linear equations using BLP-type IVs.

\begin{table}[tbp]\centering
\def\sym#1{\ifmmode^{#1}\else\(^{#1}\)\fi}
\captionsetup{justification=centering}
\caption{Estimation Result: Tax Parameters\\
Degree of polynomials = 3}
\label{tab:tax_degree3}
\begin{tabular}{l*{4}{c}}
\hline\hline
            &\multicolumn{1}{c}{(1)}   &\multicolumn{1}{c}{(2)}   &\multicolumn{1}{c}{(3)}   &\multicolumn{1}{c}{(4)}   \\
            &  University   &  University   &    Hospital   &    Hospital   \\[1em]
\hline
Constant    &      -6.658***&      -6.704***&       1.875** &       1.966** \\
            &     (0.314)   &     (0.327)   &     (0.848)   &     (0.874)   \\
[1em]
Salary (million Yen)&       2.412***&       2.519***&      -1.579** &      -1.810** \\
            &     (0.445)   &     (0.478)   &     (0.706)   &     (0.793)   \\
[1em]
Urban       &      0.0338   &      0.0429   &       0.114*  &       0.119*  \\
            &    (0.0470)   &    (0.0496)   &    (0.0671)   &    (0.0702)   \\
[1em]
Urban $\times$ 2018&      0.0639   &    -0.00972   &      0.0882   &      0.0805   \\
            &    (0.0577)   &    (0.0644)   &    (0.0786)   &    (0.0878)   \\
[1em]
Urban $\times$ 2019&      0.0752   &      0.0509   &     -0.0886   &     -0.0634   \\
            &    (0.0575)   &    (0.0640)   &    (0.0821)   &    (0.0892)   \\
[1em]
Tokyo       &               &    -0.00704   &               &     -0.0740   \\
            &               &    (0.0751)   &               &     (0.116)   \\
[1em]
Tokyo $\times$ 2018&               &       0.265***&               &      0.0282   \\
            &               &     (0.102)   &               &     (0.144)   \\
[1em]
Tokyo $\times$ 2019&               &      0.0903   &               &     -0.0940   \\
            &               &     (0.102)   &               &     (0.155)   \\[1em]
\hline
\(N\)       &        2627   &        2627   &        2627   &        2627   \\
Other covariates       &        $\surd$   &         $\surd$   &      $\surd$    &   $\surd$    \\
Tokyo   $\times$ Year   &        &         $\surd$   &         &   $\surd$    \\
\hline\hline
\multicolumn{5}{l}{\footnotesize Standard errors in parentheses}\\
\multicolumn{5}{l}{\footnotesize * p<0.1, ** p<0.05, *** p<0.01}\\
\end{tabular}
\vskip 10pt
\end{table}

\paragraph{Estimation results}
We take advantage of the fact that the regional constraints on urban areas are getting the more strict as time goes to clarify the existence of implicit tax.
As we explain in Section \ref{sec:data}, the government lowers the upper bounds on the number of matches in the urban areas by $5\%$ every year.
Hence, if the surplus in urban areas have been decreased due to the constraints, the estimated coefficients on dummy variables of urban or Tokyo will decrease over the years.

Table \ref{tab:tax_degree3} shows the estimation results based on IV estimation: where we include all the covariates in Table \ref{tab:results_degree3} and additioanly the interaction terms between dummy variables of urban and Tokyo and the dummy variables of each year.
As found in every specifications for both sides, we do not find the decrease in the coefficients of dummy variables of urban areas. Furthermore, we do not find any positive impact of living in urban areas except for living in Tokyo in 2018.
Based on these results, we conclude that the current market does not suffer from any implicit tax and the market outcome is the aggregate equilibrium under the reduced capacities.

\subsection{Relative impacts}
\label{sec:appendix_relative_impacts}
\begin{figure}[t]
\begin{center}
    \centering
    \includegraphics[width = 15cm]{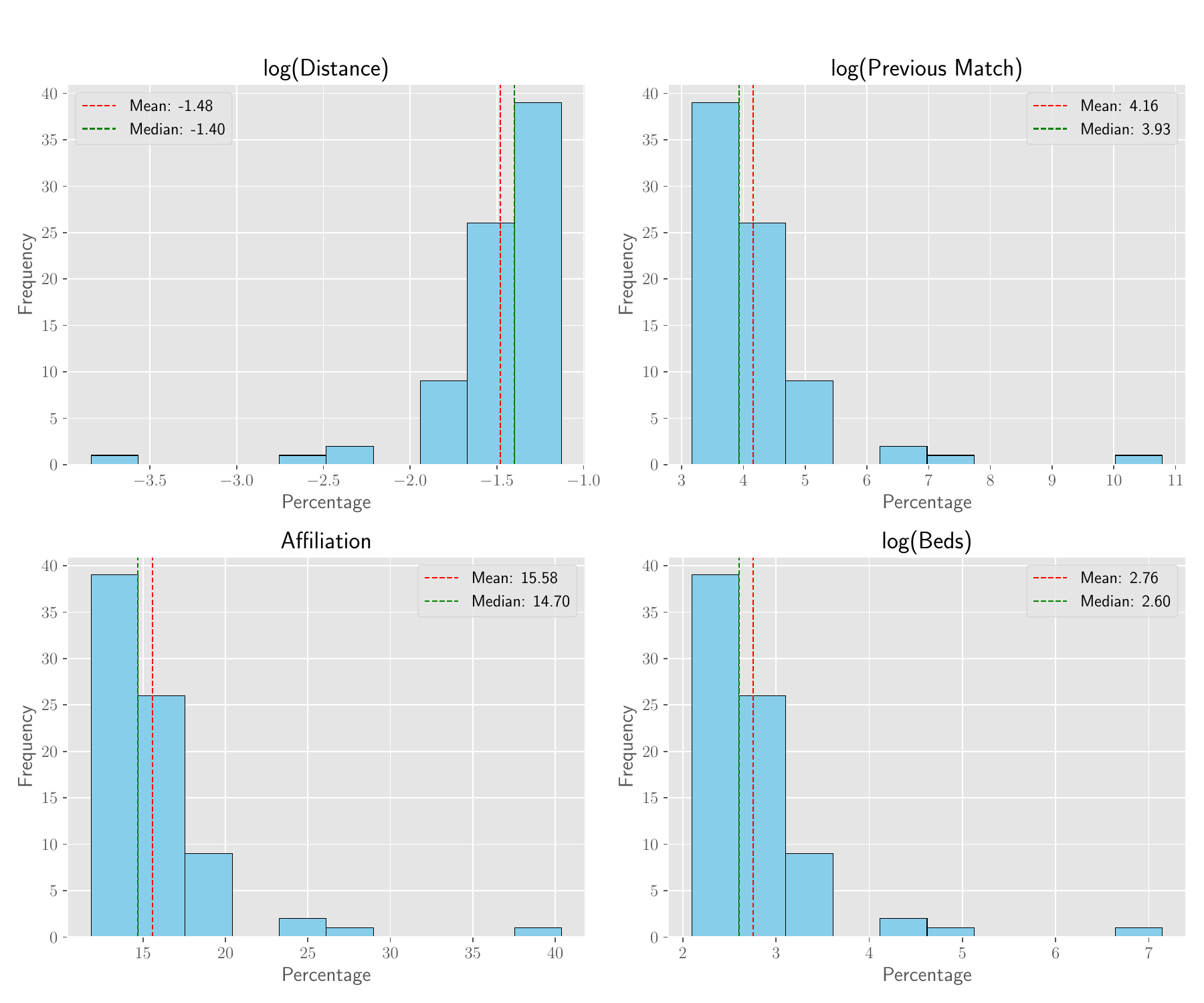}
    \caption{Relative Size of Coefficients in Doctor's Preference.}
    \label{fig:doctor_size_hist}
\end{center}
\end{figure}

\begin{figure}[t]
\begin{center}
    \centering
    \includegraphics[width = 15cm]{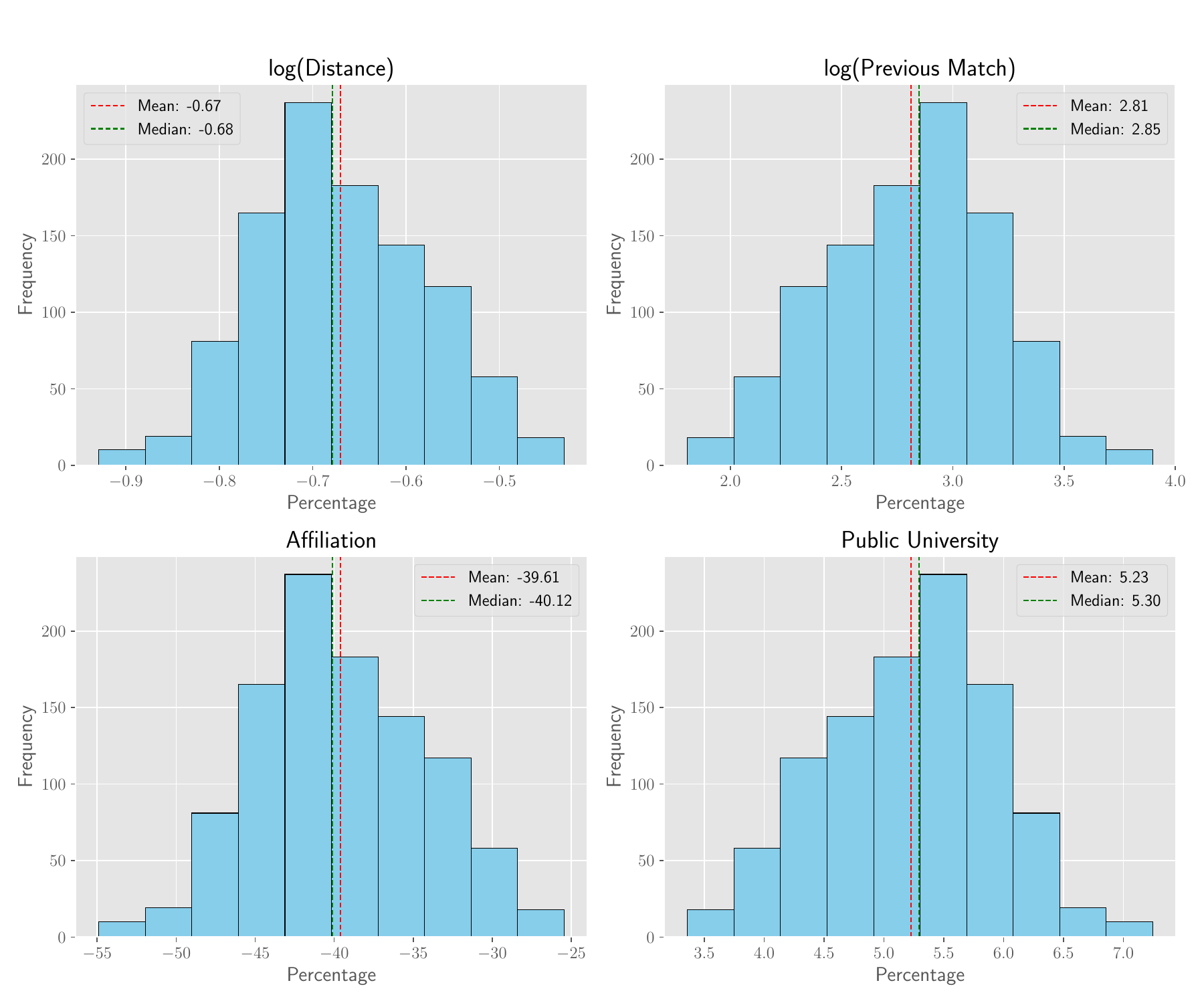}
    \caption{Relative Size of Coefficients in Hospital's Preference.}
    \label{fig:hospital_size_hist}
\end{center}
\end{figure}

To grab the sizes of impacts, we compute the ratio of these coefficients to the aggregate-level utility.
Specifically, we first transform the aggregate-level utilities into monetary unit based on the estimation results in the second stage estimation: for doctor side, we compute $U^{money}_{sht} \equiv \frac{\hat{\tilde{U}}_{sht} - \hat{\gamma}_{0, U, t}}{\hat{\gamma}_{1, U}}$ and for hospital side, we compute $V^{money}_{sht} \equiv \frac{\hat{\tilde{V}}_{sht} - \hat{\gamma}_{0, V, t}}{\hat{\gamma}_{1, V}}$.\footnote{Note that the constants depend on period $t$, i.e. $\hat{\gamma}_{0, U, t}$ and $\hat{\gamma}_{0, V, t}$, because we include the dummy variables of every years.}
Then, we take the average of these transformed aggregate-level utilities with respect to the periods and the institutions in the other side of the market: $\bar{U}^{money}_{s} \equiv \frac{1}{HT}\sum_{h,t}U^{money}_{sht}$ and $\bar{V}^{money}_{h} \equiv \frac{1}{ST}\sum_{s,t}V^{money}_{sht}$.
$\bar{U}^{money}_{s}$ and $\bar{V}^{money}_{s}$ are measures of the expected utilities in the matching market computed for every universities and hospitals.
Finally, we take the ratio of the estimated coefficients to these measures to grab the relative size of the coefficients.
For the logarithm covariates, we compute the relative size of $10\%$ changes of the covariates.

Figure \ref{fig:doctor_size_hist} depicts the histograms of the relative size of the coefficients in doctor's preference for the four covariates which have statistically significant influence in Table \ref{tab:monetary_u_3}: the logarithm of distance, the logarithm of the number of previous matches, the dummy variable of affiliation, and the logarithm of the number of beds.
In each panel, we show the mean and the median of the relative size of impacts. Although there is variation in the utility level among the universities, the distribution of the relative size of impacts has single peak and their means and the medians are not so different.
The average of the relative size of impact of $10\%$ change in distance amounts to $1.4\%$ of the doctor's utility, the same one of the number of previous match amounts to $4.16\%$, and the same one of the number of beds amounts to $2.76\%$. On average, affiliation relationship amounts to about $15.58\%$ of the doctors' average utillities.

Figure \ref{fig:hospital_size_hist} shows the same histograms for the hospital side preference. We plot the historgrams of the four covariates which shows the statistically significance in Table \ref{tab:monetary_v_3}: the logarithm of distance, the logarithm of the number of previous matches, affiliation relationship, and the indicator of the public university.
As the average utilities of hospitals are larger than the ones of doctors in the monetary unit sense, the computed relative size of impacts are likely smaller than the values obtained in the case of doctors.
The average of the relative size of impact of $10\%$ change in distance amounts to $0.67\%$ of the doctor's utility and the same one of the number of previous match amounts to $2.81\%$. On average, graduates from public university, which is usually an elite school, gives $5.23\%$ increase in the utility of hospitals.
Although the affiliation relationship gives the largest negative impact on the utility of hospitals, this estimate is not stable for the choice of the degree of polynomials as shown in Appendix \ref{sec:appendix_empirical_results}.

\subsection{Results when the degree of polynomial is set to 2}
\label{sec:appendix_empirical_results}
Here we show the empirical results obtained when we set the degree of polynomials in the first step to two.
All the tables and figures listed here corresponds to the tables and figures shown in Section \ref{sec:empirical_results} and Appendix \ref{sec:appendix_AE_EAE}.
We do not find any qualittative difference in the main findings from the case where we set the degree of polynomials to three.

\begin{figure}[htbp]
\begin{center}
    \centering
    \begin{subfigure}[b]{\textwidth}
        \centering
        \includegraphics[width=\textwidth]{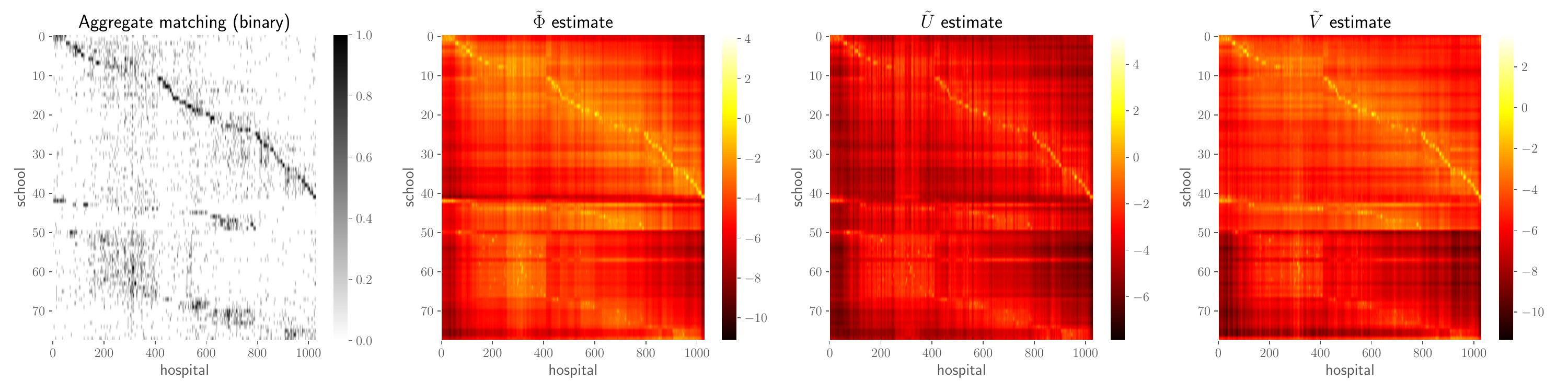}
        \caption{2017}
        \label{fig:app:label-2017}
    \end{subfigure}
    \vfill
    \begin{subfigure}[b]{\textwidth}
        \centering
        \includegraphics[width=\textwidth]{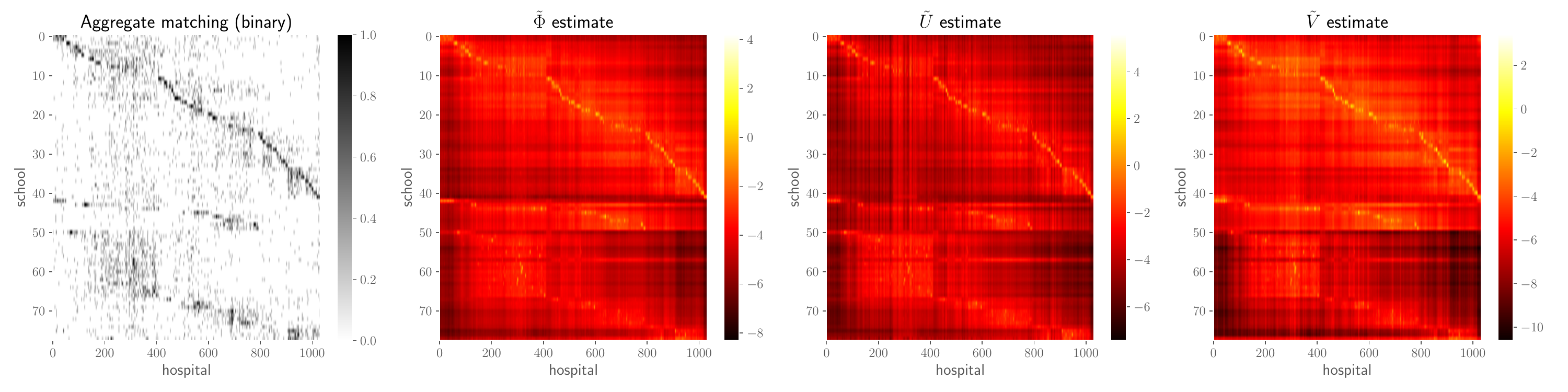}
        \caption{2018}
        \label{fig:app:label-2018}
    \end{subfigure}
    \vfill
    \begin{subfigure}[b]{\textwidth}
        \centering
        \includegraphics[width=\textwidth]{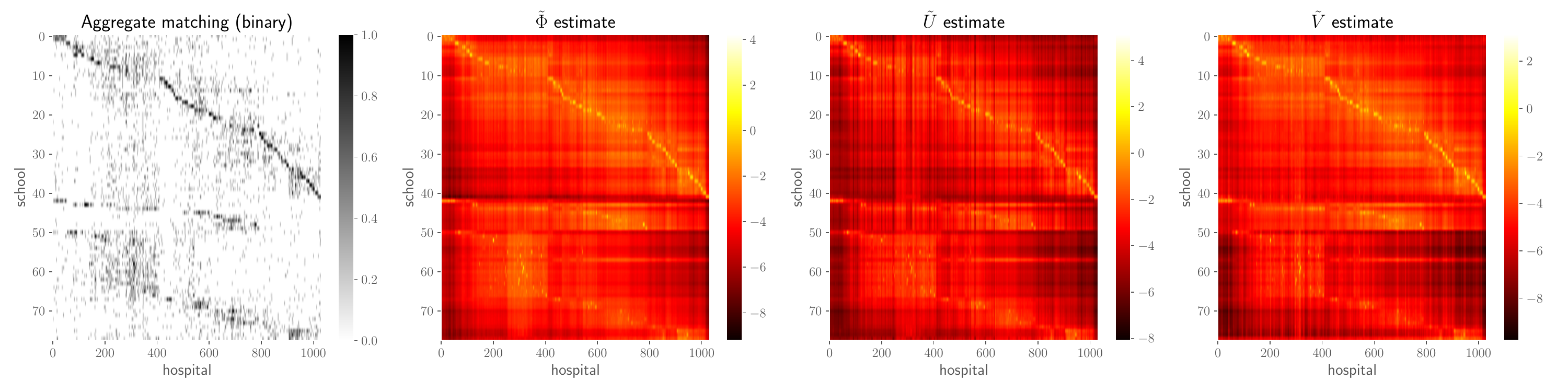}
        \caption{2019}
        \label{fig:app:label-2019}
    \end{subfigure}
    \caption{Aggregate matchings, estimated systematic utilities and estimated social surpluses.}
    \label{fig:first_stage_all_2}
\end{center}
\end{figure}

\begin{table}[htbp]\centering
\def\sym#1{\ifmmode^{#1}\else\(^{#1}\)\fi}
\captionsetup{justification=centering}
\caption{Estimation Result: Preference Parameters \\
Degree of polynomials = 2}
\begin{tabular}{l*{4}{c}}
\hline\hline
            &\multicolumn{1}{c}{(1)}   &\multicolumn{1}{c}{(2)}   &\multicolumn{1}{c}{(3)}   &\multicolumn{1}{c}{(4)}   \\
            &  University   &University (IV)   &    Hospital   &Hospital (IV)   \\[1em]
\hline
Constant    &      -5.917***&      -7.988***&       1.485** &       1.703** \\
            &     (0.225)   &     (0.429)   &     (0.705)   &     (0.787)   \\
[1em]
Salary (million Yen)&       0.564***&       4.129***&       0.704***&      -1.231*  \\
            &     (0.143)   &     (0.627)   &     (0.142)   &     (0.718)   \\
[1em]
Tokyo       &      -0.129***&      0.0102   &       0.100*  &    -0.00294   \\
            &    (0.0492)   &    (0.0604)   &    (0.0524)   &    (0.0621)   \\
[1em]
urban       &      -0.102***&      0.0524   &       0.252***&       0.191***\\
            &    (0.0339)   &    (0.0436)   &    (0.0317)   &    (0.0418)   \\
[1em]
log(Distance)&      -0.380***&      -0.400***&      -0.331***&      -0.304***\\
            &    (0.0157)   &    (0.0158)   &    (0.0150)   &    (0.0172)   \\
[1em]
log(Previous Match)&       1.583***&       1.563***&       1.663***&       1.667***\\
            &    (0.0398)   &    (0.0302)   &    (0.0398)   &    (0.0449)   \\
[1em]
Affiliation &      -0.488** &      -0.431***&      -2.676***&      -2.827***\\
            &     (0.199)   &     (0.146)   &     (0.173)   &     (0.194)   \\
[1em]
University hospital&      -0.199** &      0.0127   &               &               \\
            &    (0.0800)   &     (0.103)   &               &               \\
[1em]
Govermental hospital&      0.0319   &     -0.0589   &               &               \\
            &    (0.0341)   &    (0.0433)   &               &               \\
[1em]
log(Beds)   &       0.511***&       0.628***&               &               \\
            &    (0.0359)   &    (0.0462)   &               &               \\
[1em]
Public university&               &               &       0.182***&       0.176***\\
            &               &               &    (0.0531)   &    (0.0573)   \\
[1em]
Prestige    &               &               &      -1.905***&      -3.125***\\
            &               &               &     (0.668)   &     (0.727)   \\[1em]
\hline
\(N\)       &        2847   &        2627   &        2847   &        2627   \\
\hline\hline
\multicolumn{5}{l}{\footnotesize Standard errors in parentheses. * p<0.1, ** p<0.05, *** p<0.01}\\
\end{tabular}
\end{table}

\begin{table}[htbp]\centering
\def\sym#1{\ifmmode^{#1}\else\(^{#1}\)\fi}
\captionsetup{justification=centering}
\caption{Estimation Result: Tax Parameters\\
Degree of polynomials = 2}
\begin{tabular}{l*{4}{c}}
\hline\hline
            &\multicolumn{1}{c}{(1)}   &\multicolumn{1}{c}{(2)}   &\multicolumn{1}{c}{(3)}   &\multicolumn{1}{c}{(4)}   \\
            &  University   &  University   &    Hospital   &    Hospital   \\[1em]
\hline
Constant    &      -8.028***&      -7.928***&       1.755** &       1.725** \\
            &     (0.415)   &     (0.427)   &     (0.763)   &     (0.784)   \\
[1em]
Salary (million Yen)&       4.298***&       4.121***&      -1.276** &      -1.234*  \\
            &     (0.589)   &     (0.625)   &     (0.643)   &     (0.719)   \\
[1em]
Urban       &      -0.100   &      -0.116*  &       0.106*  &      0.0755   \\
            &    (0.0621)   &    (0.0649)   &    (0.0623)   &    (0.0646)   \\
[1em]
Urban $\times$ 2018&       0.275***&       0.247***&       0.179** &       0.239***\\
            &    (0.0763)   &    (0.0842)   &    (0.0720)   &    (0.0811)   \\
[1em]
Urban $\times$ 2019&       0.217***&       0.256***&      0.0655   &       0.107   \\
            &    (0.0761)   &    (0.0838)   &    (0.0749)   &    (0.0805)   \\
[1em]
Tokyo       &               &      0.0257   &               &       0.119   \\
            &               &    (0.0982)   &               &     (0.101)   \\
[1em]
Tokyo $\times$ 2018&               &      0.0976   &               &      -0.216*  \\
            &               &     (0.133)   &               &     (0.123)   \\
[1em]
Tokyo $\times$ 2019&               &      -0.145   &               &      -0.151   \\
            &               &     (0.133)   &               &     (0.132)   \\
[1em]
\hline
\(N\)       &        2627   &        2627   &        2627   &        2627   \\
Other covariates       &        $\surd$   &         $\surd$   &      $\surd$    &   $\surd$    \\
Tokyo   $\times$ Year   &        &         $\surd$   &         &   $\surd$    \\
\hline\hline
\multicolumn{5}{l}{\footnotesize Standard errors in parentheses. * p<0.1, ** p<0.05, *** p<0.01}\\
\end{tabular}
\end{table}

\begin{table}[htbp]\centering
\def\sym#1{\ifmmode^{#1}\else\(^{#1}\)\fi}
\captionsetup{justification=centering}
\caption{University preference parameters (unit: million Yen) \\
Degree of polynomials = 2}
\begin{tabular}{l*{3}{c}}
\hline\hline
            &\multicolumn{1}{c}{(1)}&\multicolumn{1}{c}{(2)}&\multicolumn{1}{c}{(3)}\\
Coefficient of Salary =           &\multicolumn{1}{c}{$4.129$} &\multicolumn{1}{c}{$ 4.298$} &\multicolumn{1}{c}{$4.121$}\\[1em]
\hline
log(Distance)&      -0.097\sym{***}&      -0.094\sym{***}&      -0.097\sym{***}\\
            &      (0.01)         &      (0.01)         &      (0.01)         \\[1em]
log(Previous Match)&       0.378\sym{***}&       0.363\sym{***}&       0.379\sym{***}\\
            &      (0.06)         &      (0.05)         &      (0.06)         \\[1em]
Affiliation &      -0.104\sym{**} &      -0.102\sym{**} &      -0.106\sym{**} \\
            &      (0.04)         &      (0.04)         &      (0.04)         \\[1em]
University Hospital&       0.003         &       0.006         &       0.003         \\
            &      (0.02)         &      (0.02)         &      (0.02)         \\[1em]
Governmental Hospital&      -0.014         &      -0.015         &      -0.014         \\
            &      (0.01)         &      (0.01)         &      (0.01)         \\[1em]
log(Beds)   &       0.152\sym{***}&       0.148\sym{***}&       0.153\sym{***}\\
            &      (0.02)         &      (0.02)         &      (0.02)         \\[1em]
\hline
\(N\)       &        2627         &        2627         &        2627         \\
Urban $\times$ Year &  & $\surd$ & $\surd$  \\
Tokyo $\times$ Year &  &  & $\surd$  \\
\hline\hline
\end{tabular}
\end{table}

\begin{table}[htbp]\centering
\def\sym#1{\ifmmode^{#1}\else\(^{#1}\)\fi}
\captionsetup{justification=centering}
\caption{Hospital preference parameters (unit: million Yen)\\
Degree of polynomials = 2}
\begin{tabular}{l*{3}{c}}
\hline\hline
            &\multicolumn{1}{c}{(1)}&\multicolumn{1}{c}{(2)}&\multicolumn{1}{c}{(3)}\\
Coefficient of Salary =           &\multicolumn{1}{c}{$1.231$} &\multicolumn{1}{c}{$1.276$} &\multicolumn{1}{c}{$1.234$}\\[1em]
\hline
log(Distance)&      -0.247         &      -0.238         &      -0.247         \\
            &      (0.15)         &      (0.13)         &      (0.15)         \\[1em]
log(Previous Match)&       1.354         &       1.307\sym{*}  &       1.352         \\
            &      (0.79)         &      (0.66)         &      (0.79)         \\[1em]
Affiliation &      -2.296         &      -2.222\sym{*}  &      -2.294         \\
            &      (1.29)         &      (1.07)         &      (1.29)         \\[1em]
Public University&       0.143         &       0.136         &       0.141         \\
            &      (0.09)         &      (0.07)         &      (0.09)         \\[1em]
Prestige    &      -2.537         &      -2.451\sym{*}  &      -2.522         \\
            &      (1.37)         &      (1.15)         &      (1.36)         \\[1em]
\hline
\(N\)       &        2627         &        2627         &        2627         \\
Urban $\times$ Year &  & $\surd$ & $\surd$  \\
Tokyo $\times$ Year &  &  & $\surd$  \\
\hline\hline
\end{tabular}
\end{table}

\newpage
\subsection{Counterfactual simulations for the other years}
\label{sec:appendix_cf_other_years}

\begin{table}[tbp]
\centering
\small
\caption{Comparison between Aggregate-level Equilibria}
\label{tab:welfare_AE}
\begin{threeparttable}
\begin{tabularx}{\textwidth}{l @{} RRR}
\toprule 
Policy & \multicolumn{1}{c}{\textsf{AC (Artificial Caps)}} & \multicolumn{1}{c}{\textsf{NC (No Caps)}} & \multicolumn{1}{c}{\textsf{OS (Optimal Subsidy)}} \\

\midrule
Artificial caps & Yes & No & No \\
Floor constraints & Yes & No & Yes \\
Subsidies & No & No & Yes \\
\midrule

\multicolumn{4}{l}{\hspace{-0.5em}\textbf{2017}} \\
Match rate & 0.868 & 0.912 & 0.912 \\
Doctors' welfare & 32795.8 & 33441.7 & 33444.5 \\
Hospitals' welfare & 29586.6 & 31578.1 & 31579.4 \\
Government's revenue & 0.0 & 0.0 & $[-10.5, -7.4]$ \\
Total welfare & 62382.3 & 65019.8 & $[65013.3, 65016.5]$ \\
\#(subsidized regions) & 0 & 0 & 3 \\
Average subsidy & 0.000 & 0.000 & -0.040 \\
\#(constraint violations) & 0 & 3 & 0 \\
\midrule

\multicolumn{4}{l}{\hspace{-0.5em}\textbf{2018}} \\
Match rate & 0.844 & 0.895 & 0.896 \\
Doctors' welfare & 33928.3 & 34272.6 & 34277.1 \\
Hospitals' welfare & 30390.4 & 33700.2 & 33703.4 \\
Government's revenue & 0.0 & 0.0 & $[-18.9, -13.3]$ \\
Total welfare & 64318.7 & 67972.8 & $[67961.6, 67967.1]$ \\
\#(subsidized regions) & 0 & 0 & 5 \\
Average subsidy & 0.000 & 0.000 & -0.038 \\
\#(constraint violations) & 0 & 5 & 0 \\
\midrule

\multicolumn{4}{l}{\hspace{-0.5em}\textbf{2019}} \\
Match rate & 0.869 & 0.912 & 0.912 \\
Doctors' welfare & 33244.7 & 33752.2 & 33755.5 \\
Hospitals' welfare & 30406.1 & 32648.8 & 32649.1 \\
Government's revenue & 0.0 & 0.0 & $[-9.4, -6.6]$ \\
Total welfare & 63650.8 & 66401.0 & $[66395.2, 66398.0]$ \\
\#(subsidized regions) & 0 & 0 & 2 \\
Average subsidy & 0.000 & 0.000 & -0.042 \\
\#(constraint violations) & 0 & 2 & 0 \\
\bottomrule
\end{tabularx}
\begin{tablenotes}
\footnotesize
\item[*] All values except match rates, \#(subsidized regions), and \#(constraint violations) are expressed in units of 1 million JPY per month. The government's revenue is positive when taxes are imposed on doctors and hospitals and negative when subsidies are provided to them. The welfare of doctors and hospitals is scaled according to specification (1) in \Cref{tab:monetary_u_3} and \Cref{tab:monetary_v_3}. We present the bounds of the government's net revenue, scaled by the coefficients on the doctor side and the hospital side, respectively. The total welfare is the sum of doctors' welfare, hospitals' welfare, and the government's revenue. \#(constraint violations) counts the number of prefectures violating the lower bounds (among the 15 rural regions).
\end{tablenotes}
\end{threeparttable}
\end{table}

\end{document}